\documentclass[a4paper,11pt]{article}

\usepackage[utf8]{inputenc} 
\usepackage[T1]{fontenc}    
\usepackage{url}            
\usepackage{booktabs}       
\usepackage{amsfonts}       
\usepackage{nicefrac}       
\usepackage{microtype}      
\usepackage{geometry, xcolor}         
\usepackage{enumitem}
\usepackage{fancyhdr}
\newcommand{\ignore}[1]{}
\RequirePackage{amsthm,amsmath,amsfonts,amssymb}
\RequirePackage{natbib}
\RequirePackage[colorlinks,citecolor=blue,urlcolor=blue]{hyperref}
\usepackage{algorithm}
\usepackage{algorithmic}
\RequirePackage{graphicx,xcolor}
\usepackage{siunitx}
\usepackage{capt-of}
\usepackage{booktabs}
\usepackage{varwidth}
\usepackage{enumitem}   
\usepackage{comment}
\usepackage{lineno}
\usepackage{amsbsy}
\usepackage{blindtext}

\usepackage{graphicx,booktabs,multirow,xcolor}
\usepackage[font=small, labelfont=bf]{caption}
\usepackage{subcaption,wrapfig}

\numberwithin{equation}{section}
\theoremstyle{plain}

\newtheorem{theorem}{Theorem}[section]
\newtheorem{proposition}{Proposition}
\newtheorem{lemma}[theorem]{Lemma}

\newtheorem{assumption}{Assumption}[section]

\theoremstyle{remark}

\newtheorem*{example}{Example}

\newtheorem{remark}{Remark}[section]

\RequirePackage[utf8]{inputenc} 
\RequirePackage[T1]{fontenc}    
\RequirePackage{multicol,booktabs,longtable, multirow, makecell}       
\RequirePackage{bm}  
\RequirePackage{nicefrac}       
\usepackage{textcomp}
\RequirePackage{microtype}      
\RequirePackage{cleveref}       
\RequirePackage{psfrag,epsf}
\RequirePackage{caption,subcaption}
\usepackage{enumitem}

\def\IE{\mathbb{E}}
\def\R{\mathbb{R}}

\def\Z{\mathbb{Z}}
\def\N{\mathbb{N}}

\newcommand{\Acal}{\mathcal{A}}

\DeclareMathOperator*{\argmin}{arg\,min}

\newcommand{\bb}[1]{\boldsymbol{#1}}

\newcommand{\IP}{\mathbb{P}}

\definecolor{mygreen}{rgb}{0.09,0.62,0.45} 
\definecolor{myorange}{rgb}{0.83,0.37,0.10} 
\definecolor{myblue}{rgb}{0.06,0.45,0.69}  
\definecolor{myred}{rgb}{0.94,0.20,0.13}
\definecolor{mypurple}{RGB}{76,0,153} 

\newcommand{\argmax}{\mathop{\rm arg\max}}

\def\limsup{\mathop{\overline{\rm lim}}}
\def\liminf{\mathop{\underline{\rm lim}}}

\title{Fast localization of anomalous patches in spatial data under dependence}

\author{%
\begin{tabular}{ccc}
Soham Bonnerjee\footnote{Equal contributions} & Sayar Karmakar$^*$ & George Michailidis \\
\textit{University of Chicago} & \textit{University of Florida} & \textit{UCLA}
\end{tabular}
}
\date{}

\begin{document}






\maketitle

\begin{abstract}
We propose a scalable, provably accurate method for localizing an unknown number of multiple axis-aligned anomalous patches in spatial data under a general class of spatial dependence. Motivated by the practical need to detect localized changes rather than completely segment large spatial grids, we first introduce both a naive and a significantly faster intelligent-sampling-based estimator for a single patch. We then extend this methodology to the highly challenging multiple-patch setting and propose a two-stage Spatial Patch Localization of Anomalies under DEpendence procedure (SPLADE). Under mild conditions on signal strength, separation from the boundary, inter-patch separation, and a uniform Gaussian approximation, we establish simultaneous consistency for the estimated number of patches and for each individual patch boundary. Extensive numerical results based on synthetic data scenarios demonstrate that the proposed method exhibits significant computational and accuracy gains over competing approaches, as well as robustness to moderate and severe spatial dependence. Finally, we demonstrate the real-world utility of the proposed method by applying it to frame-to-frame video surveillance data, where it accurately detects small, closely separated subjects, a task where existing methods are significantly slower and highly prone to spurious detections due to not accounting for spatial dependence. A second application on 3D fibrous media is deferred to the Appendix.
\end{abstract}


\section{Introduction}
Inference for anomalous patches or clusters in spatial data has a long history, going back at least to \cite{besagnewell1991}. The problem is rooted in the earlier literature on epidemic change-points for time-ordered data \citep{levin1985cusum, yao1993tests, inclan1994use, huvskova1995estimators, csorgo1997limit, chen2016detecting, ravckauskas2004holder,ravckauskas2006testing, ning2012empirical}, aided by mathematical tools such as those developed in \cite{naus1965}. In the spatial setting, one natural analogue is boundary or partition recovery, where separating hypersurfaces divide the domain into heterogeneous regions \citep{hallpengraujrssb, chan-yau, song2011comparative, otto2016detection, fan2018approximate, can-yau-dep}. In many applications, however, the practically relevant departure is not a global partition, but a local hotspot; namely, a small region whose behavior differs from its surroundings. This viewpoint is especially natural in public health \citep{hjalmars1996childhood, souza2019did,  lord2020investigation}, public safety and urban planning \citep{warden2008comparison, gao2013early, zeoli2014homicide, basu2025data}, and environmental monitoring \citep{riitters2005hot, shackelford2015conservation, vega2012cluster}, where the goal is to screen an entire map for localized excess risk without specifying the anomalous region in advance. 

Formally, we model the observed data as a $d$-dimensional lattice field
\begin{equation}\label{eq:spatialmodel}
X_{\boldsymbol i} = \mu_{\boldsymbol i} + \eta_{\boldsymbol i},
\qquad \boldsymbol i \in [\mathbf n],
\end{equation}
where $\mathbf n=(n_1,\dots,n_d)\in\mathbb N^d$ and
$[\mathbf n] := \prod_{\ell=1}^d \{1,2,\dots,n_\ell\}$.
Let $I := \bigcup_{j=1}^{K} I_{j}$ denote the union of $K$ (pairwise disjoint) anomalous patches, with each $I_{j}\subseteq[\mathbf n]$.
We model the mean field by a baseline level $\mu_0$ outside the anomalous region and patch-specific mean shifts inside:
\begin{equation}\label{eq:meanfield_deltas}
\mu_{\boldsymbol i}=
\begin{cases}
\mu_0 + \delta_j, & \boldsymbol i \in I_{j},\quad j=1,\dots,K,\\
\mu_0, & \boldsymbol i \in [\mathbf n]\setminus I,
\end{cases}
\end{equation}
where $\delta_j\neq 0$ for anomalous patches (and $\delta_{j_1}\neq \delta_{j_2}$ is allowed). With this formulation, the two fundamental statistical problems are: (i) to test whether anomalous patches exist at all, and (ii) to localize them when they do.

Likelihood-based spatial scan statistics systematized anomalous patch detection by computing likelihood ratios over candidate regions and using Monte Carlo calibrations to adjust for multiple testing \citep{kulldorffnagarwalla1995, glaz2001scan}. Pioneered by \cite{kulldorff1997} and popularized by the \texttt{SaTScan} software \citep{kulldorff2006satscan, block2007satscan}, this framework has been widely extended. Developments include adaptations for count \citep{neill2004detecting, neil2012fast}, ordinal \citep{jung2007spatial}, and Bernoulli data \citep{boustikas, walther2010}, alongside minimax-optimal procedures for Gaussian settings \citep{castro2005ieeeit, castro2011, chansinica, sharpnackejs, bodhi-scam, walthercalibrating}. Drawing on limit theory for Gaussian random fields \citep{jiang2002, kabluchko2011}, much of this literature primarily focuses on the testing problem \citep{walther2010, walthercalibrating, castro2011, sharpnackejs, gaojmlr2016, castrojasa2018, bodhi-scam, klein2022scan, stoepker2025, kohne2025edge}, though recent advances also include non-parametric, rank-, and permutation-based methods \citep{cucala2014distribution, jung2015nonparametric, castrojasa2018, konigaos2020, stoepker2025}. For comprehensive reviews of the extensive scan-statistics literature, see \cite{castro2011, abolhassani2021up, xie2022statistically}.

Conceptually, the ``scan statistic'' can be argued to adopt a \textit{top-down} perspective: one specifies a class of anomalous patterns a priori, and systematically scans the field to detect regions consistent with that structure. On the other hand, there has been a stream of complementary literature stemming from partition recovery or anomalous sub-graph or cluster detection problems \citep{madrid2021lattice, yu2022optimal, wang2025optimal}, which usually proceeds via a \textit{bottom-up} approach by constructing the anomalous region incrementally from local, fine-scale evidence. In such an approach, one can quantify uncertainty in the boundary by constructing location-wise “membership evidence’’ scores from the ensemble of near-optimal clusters, producing a fuzzy boundary map rather than a single hard contour \cite{oliveira_etal_2018}. These distinctions are not mutually exclusive: for example, several subgraph or sub-matrix recovery problems \cite{castro2011, sharpnack2013near, butucea-submatrix} also perform scans of some statistics over a family of corresponding candidates, blurring the boundary between the two approaches. Nevertheless, we find this distinction convenient for exposition. 
 Our framework aligns closely with the top-down scan-statistic perspective, though we contextualize this work within both frameworks.

\textit{\textbf{The case for dependent-data}}: A vast majority of the literature on the ``bottom-up'' approach, perhaps constrained by the available theoretical guarantees of the clustering algorithms, exclusively consider independently distributed observations. A similar restriction appears in a substantial portion of the scan-statistics literature as well. This assumption is known to be consequential in practice; for example, \cite{loh_zhu_2007} show that when spatial data exhibit unmodeled positive autocorrelation, classical scan statistics that assume independence can yield overly small $p$-values and spurious cluster detections. To the best of our knowledge, only a small number of works consider spatial dependence in anomalous patch detection, including \cite{can-yau-dep, dresvyanskiy2020detecting, makogin2024statistical, kirch2025scan, wang2025optimal}. Among these, \cite{wang2025optimal} is, to our knowledge, the only work that directly addresses the localization problem, providing theoretical guarantees for consistent estimation of both the number of anomalous patches and their locations; rest of these works primarily investigate the testing problem for existence of anomalous patches. Moreover, \cite{dresvyanskiy2020detecting, makogin2024statistical, wang2025optimal} can be viewed as adopting a \textit{bottom-up} perspective, using clustering-based methods to construct anomalous regions from local evidence. In contrast, \cite{can-yau-dep, kirch2025scan} formulate the problem in the scan-statistics framework and concentrate on testing for the presence of anomalous regions. Moreover, \cite{dresvyanskiy2020detecting, kirch2025scan} assume $M$-dependent observations, which substantially restricts the class of admissible dependence structures; for example, standard spatial autoregressive models \citep{anselin1998spatial, smirnov2001fast} generally do not satisfy this assumption. On the other hand, \cite{wang2025optimal} assumes sub-Gaussian tails, which allows the analysis to proceed using techniques similar to those developed for the Gaussian case.

This paper focuses on the \textit{localization of multiple anomalous patches under general spatial dependence}. Crucially, our methodology requires only the existence of \textit{finite $p$-th moments}, thereby seamlessly accommodating \textit{heavy-tailed noise} distributions. As noted by \cite{castro2011}, while testing and localization may share superficial computational traits, they demand fundamentally distinct theoretical assumptions and techniques. By advancing the scan-statistics paradigm, we leverage threshold calibration across spatial scales to establish sharp detection regimes alongside rigorous localization guarantees for structured clusters. Consequently, our methodology delineates the precise conditions under which anomalous boundaries can be reliably localized.

\textit{\textbf{The case for scalability}}: Computational feasibility is a crucial issue often overlooked in the pursuit of optimal detection. For spatially dependent data, \cite{wang2025optimal} proposes a rank-based scan algorithm (similar to \cite{sharpnackejs}) that localizes arbitrarily shaped patches in approximately $O(|\bb{n}|^{3/2})$ time, but lacks accompanying theoretical guarantees. In contrast, the detection algorithms in \cite{can-yau-dep} and \cite{kirch2025scan} require $O(|\bb{n}|^2)$ time. Bottom-up clustering approaches for general anomalous patches typically incur even higher computational cost, an issue highlighted by \cite{yu2022optimal} despite achieving optimal statistical rates. {While \cite{madrid2021lattice} prescribes an $O(|\bb{n}|)$ algorithm, our numerical work reveals its empirical runtime is considerably larger than the method proposed in this paper.} Ultimately, most methods for detecting \textit{irregularly shaped} patches \citep{patil2004detection, tango_takahashi_2005, duczmal2007, takahashi2008flexibly, kim2017evaluation, otani2021flexible, inoue2023detection, oliveira2025arbitrarily} suffer from at least one of four limitations: (i) substantial computational costs, (ii) restrictive assumptions on cluster size corresponding to the anomalous patch, (iii) reliance on knowing the true number of clusters, or (iv) intractable analyses driven by independence assumptions coupled with machine-learning models that are hard to analyze theoretically.

On the other hand, much of the scan-statistics literature alleviates computational costs by restricting the search space to specific geometrical shapes for the anomalies. While the \texttt{SaTScan} software \cite{kulldorff2006satscan} focuses on circular or elliptical patches \citep{tango_takahashi_2005, kulldorff-elliptic}, another prominent approach employs rectangular scan windows. A substantial body of work \citep{neill2004rapid, castro2005ieeeit, walther2010, castro2011, bucchia_jspi, sharpnackejs, bucchia2017jmva, konigaos2020, dresvyanskiy2020detecting, kourect2023, makogin2024statistical} has developed approximately linear-time algorithms for detecting axis-parallel (hyper)-rectangular patches in d-dimensional fields, primarily under independence. Theoretically, as noted by \cite{sharpnackejs, bucchia2017jmva}, scanning over rectangular windows facilitates the derivation of tight asymptotic results even under complex spatial dependency structures.

Moreover, rectangular clusters arise naturally in certain applications, including fibrous media \citep{wirjadi2014characterization, dresvyanskiy2019application, dresvyanskiy2020detecting} and data mining \citep{huo2002multiscale}. In surveillance, criminology, and geography, axis-aligned rectangles and related raster-style grids are also preferred for three primary reasons: (i) they are operationally interpretable, allowing analysts to quickly link hotspots to concrete action levels like street segments \citep{ChaineyTompsonUhlig2008, EckChaineyCameronLeitnerWilson2005}; (ii) regular grids provide comparable micro-units for studying fine-scale heterogeneity, which is essential since crime is often concentrated at specific micro-places \citep{MallesonSteenbeekAndresen2019, Weisburd2015}; and (iii) they permit fast multiscale scanning and rigorous statistical analysis, even if they are more restrictive than irregular shapes \citep{walther2010, PeiJasraHandZhuZhou2009}. This utility is also demonstrated in CAVIAR video surveillance dataset \citep{Fisher2004PETS04, CAVIARD23}, where subjects are identified via rectangular bounding boxes. The rectangular representation is computationally efficient; as noted in the CAVIAR documentation, restricting processing to such detection zones significantly improves speed while maintaining high performance. This shows why such box-based representations are attractive for surveillance analysis in general. Later, we analyze a part of the data using our proposed scalable method and contrast it with other baselines through the lens of precise localization. 

\textbf{Main contributions and Paper Organization:} 
By focusing on axis-parallel rectangles as a canonical structured class of anomalies, we develop a computationally efficient $O(|\bb{n}|)$ algorithm  that provides rigorous detection and localization guarantees under \textit{significantly richer spatial dependence structures} than previously studied. Our core contributions are organized as follows:
\begin{itemize}
    \item \textit{Optimal single-Patch localization (Section \ref{se:model}):} We introduce a fast intelligent sub-sampling approach for localizing a single anomalous patch. Notably, for $d=1$, our method recovers optimal minimax rates, offering a highly scalable solution for epidemic time-series change point localization.
    \item \textit{Scalable multi-Patch detection (Section \ref{se:multiple}):} Building on the single-patch localization algorithm, we develop a two-stage Spatial Patch Localization of Anomalies under DEpendence procedure (SPLADE)
to simultaneously estimate an \textit{unknown number of anomalous patches and their precise boundaries}. SPLADE enjoys strong theoretical guarantees under minimal separation assumptions and is highly scalable, unlike standard binary segmentation techniques that scale poorly in spatial settings.
    \item \textit{Strong empirical performance (Section \ref{sec:simu}):} Extensive synthetic data experiments illustrate that SPLADE exhibits robustness to tuning parameters. More importantly, SPLADE enjoys significant gains in speed and accuracy across various patch configurations, signal strength and types of spatial dependence, compared to existing baselines, including DCART \citep{madrid2021lattice} and recent extensions DPLS-SAD \citep{wang2025optimal}. While competing methods falter under spatial dependence, our approach consistently achieves high Adjusted Rand Indices and low normalized Hausdorff distances, all the while maintaining attractive computational efficiency.
    \item \textit{Real-world efficacy (Section \ref{sec:data} and Appendix \ref{se:appendix real data}):} We validate our framework on video surveillance footage and 3D fibrous media, two domains inherently characterized by strong spatial dependence. For the first application, on video surveillance, we demonstrate the ability to accurately resolve and bound closely situated subjects in complex environments where baseline methods fail to distinguish individual entities. In the second application on fibrous media, we explored beyond 2D images to a 3D use-case, underlining the scalability SPLADE enjoys over other baselines. 
\end{itemize}
All theoretical proofs, auxiliary lemmas, and additional simulation studies are deferred to the Appendix.

\medskip
\noindent \textbf{Notation:} For $\bb{i} \in \Z^d$, $|\bb{i}|=i_1\ldots ,i_d$. Define a \textit{Rectangle} $I_{[\bb{a},\bb{b}]} \subseteq Z^d$ with end-points $\bb{a} \in \Z^d$ and $\bb{b} \in \Z^d$, given by $I_{[\bb{a},\bb{b}]}=\{x \in \Z^d: \bb{a}< x \leq \bb{b}\}$,
where for two vectors $\bb{x}$ and $\bb{y} \in \Z^d$, we say $\bb{x}\leq \bb{y}$ if $x_j\leq y_j$ for $j=1(1)d$. Note that, $|I_{[\bb{a},\bb{b}]}|=|\bb{b}-\bb{a}|$. We denote $I_{[\bb{0}, \bb{n}]}$ as $[\bb{n}]$. We also denote the sample size $n:=|\bb{n}|$. For two points $\bb{a}, \bb{b}\in \Z^d$, let $|\bb{a} / \bb{b}|_{\min}:= \min_{1\leq k \leq d} |a_k/b_k|$, and correspondingly $|\bb{a} / \bb{b}|_{\infty}:= \max_{1\leq k \leq d} |a_k/b_k|$. For sequences $\{a_n\}$ and $\{b_n\}$, $a_n=O(b_n)$ and $a_n = \Theta(b_n)$ imply $\limsup_{n\to \infty} a_n/b_n < \infty$ and $\liminf_{n\to \infty}a_n/b_n \to \infty$ respectively. $A \Delta B$ denotes the symmetric difference of two sets $A$ and $B$.

\section{Localization of a single anomalous patch}
Next, we introduce a generalized framework for spatial dependency. This foundation motivates our two-step localization strategy: first, a "naive" estimator with established theoretical consistency (Section \ref{se:naive estimate}), and second, an intelligent-sampling algorithm that ensures computational scalability while maintaining theoretical validity (Section \ref{se:intelligent-epidemic}). These components form the essential building blocks for our subsequent multiple-patch SPLADE methodology.


\subsection{Preliminaries: dependent spatial fields} \label{se:model}

Let $(\varepsilon_{\bb{i}})_{\bb{i} \in \Z^d}$ be a mean-zero stationary random field. We accommodate a very general dependence structure by assuming only that the $\varepsilon_i$'s satisfy a mild, maximal $\mathcal{L}_p$ bound for some $p>2$. Formally, we require:
\begin{assumption}\label{asmp:doob}
For any rectangle $I \subseteq [\mathbf{n}]$, define its partial sum $S_I^\varepsilon := \sum_{\bb{j}\in I}\varepsilon_j$. Assume that $\|\varepsilon_{\bb{0}}\|_p < \infty$ for some $p>2$. Then, we require that
    \[ \| \max_{I \subseteq [\bb{n}] } |S_I^\varepsilon|\|_p \leq C' |\bb{n}|^{1/2}, \]
    where $C'$ is independent of $\bb{n}$, but may depend on $d$ and $p$.
\end{assumption}
Assumption \ref{asmp:doob} serves as a spatial analogue of  \textit{Rosenthal's} or \textit{Doob's maximal inequality} \citep{liuxiaowu, peligrad2007maximal}. Proving such inequalities for $d\geq 2$ is non-trivial due to lack of total ordering; the independent case was pioneered by \cite{Cairoli1970} and extended to specific dependent martingales by others \citep{christofides1990maximal, hirsch1995potential, walsh2006martingales}. Given its relative obscurity in statistics, we provide a proof in Appendix \S\ref{se:cairoli} that also illustrates the mathematical techniques central to this paper. Next, we show validity of Assumption \ref{asmp:doob} across various spatial dependency mechanisms.

\begin{example}[m-dependent linear field]
Assumption \ref{asmp:doob} easily follows from Cairoli's inequality for $m$-dependent fields (See Lemma \ref{lem:m-dependent}), recently studied by \cite{kirch2025scan}. 
\end{example}

\begin{example}[Linear random fields]\label{rem:ass-1-1} 
Assumption \ref{asmp:doob} also holds (Lemma \ref{lem:linear}) for the broad class of linear random fields that are not $m$- dependent
$$\varepsilon_{\bb{i}} = \sum_{\bb{s}\in \Z^d} a_{\bb{s}} e_{\bb{i}- \bb{s}}, \text{ where } \sum_{\bb{s}\in \Z^d} a_{\bb{s}}< \infty, e_{\bb{k}} \text{ i.i.d. with } \|e_{\bb{0}}\|_p<\infty,$$ 
which encompasses the \textit{Spatial Autoregressive Model} (SAR) (See \eqref{eq:sar}) \citep{ord1975estimation} widely used in econometrics and geography \citep{anselin1998spatial,cressie2015statistics, paul2024spatial}. Later in Section \ref{sec:simu}, we revisit this model for simulations.
\end{example}

Moving beyond such specific classes of dependence, arguably the most general representation for $(\varepsilon_{\bb{i}})_{\bb{i} \in \Z^d}$ is given by the functional form:
\begin{equation}\label{eq:causalspatial}
    \varepsilon_{\bb{i}}=g(e_{\bb{i}-\bb{s}}: \bb{s}\in \Z^d),
\end{equation}
where $g: \bigotimes_{i = 1}^{\infty} \R^d \to \R^d$ is a progressively measurable function and the innovations $(e_{\bb{i}})_{\bb{i} \in \Z^d}$ are i.i.d. This representation intuitively allows for spatial dependence from any direction and is directly inspired from writing the joint distribution of dependent random variables in terms of compositions of conditional quantile functions of i.i.d. uniform random variables.  However, for meaningful analysis, we must control the influence of distant indices by imposing a decay structure on spatial correlations. A standard approach \citep{el2013central, bucchia_jspi, steland2025inference} is to assume a finite long-run variance.
\begin{assumption} \label{asmp:lrv}
    The mean-zero spatial random field $(\varepsilon_{\bb{i}})_{\bb{i} \in \Z^d}$ has a finite long-run variance $\sigma^2$, defined by
    \[ \sigma^2= \IE[\varepsilon_{\bb{0}}^2]+ \sum_{\bb{i}\ne \bb{0}} \IE[\varepsilon_{\bb{0}} \varepsilon_{\bb{i}}] <\infty \]
\end{assumption}
A pertinent question is whether Assumption \ref{asmp:lrv} suffices to guarantee Assumption \ref{asmp:doob}.
\begin{remark}[Assumption \ref{asmp:doob} under general dependency]\label{rem:ass-1-2}
    Under mild regularity conditions like Assumption \ref{asmp:lrv}, fields of type \eqref{eq:causalspatial} satisfy a weaker version of Assumption \ref{asmp:doob}. For instance, \cite{el2013central} establish:
    \begin{lemma}[Abridged from Proposition 1, \cite{el2013central}]
        Under Assumption \ref{asmp:lrv}, for any rectangle $I\subseteq [\bb{n}]$, $\|S_I\|_p \leq C |I|^{1/2}$, where $C$ is independent of $\bb{n}$.   
    \end{lemma}
    However, confirming the full maximal inequality of Assumption \ref{asmp:doob} under general dependence remains an open problem in probability theory, primarily due to the lack of total ordering in $\Z^d$. Existing results typically impose a causal structure to leverage Cairoli's inequality. The most general formulation currently available (derived from \citet{cuny2025weak} and recorded in Lemma~\ref{lem:new-rosenthal}) restricts \eqref{eq:causalspatial} to:
    \[  \varepsilon_{\bb{i}}=g(e_{\bb{i}-\bb{s}}: \bb{s}\in \Z^d, \bb{s} \ge \bb{0}), \]
    which limits dependence to specific axial directions. Despite this technical gap for the fully general case, the extensive evidence validating Assumption~\ref{asmp:doob} across diverse settings makes it a practical and reasonable condition. 
\end{remark}

Under Assumption \ref{asmp:doob}, we next introduce a theoretically valid but \textit{naive} estimator for a single anomalous patch, denoted $I_0$. Analyzing its computational limitations naturally motivates the fast, scalable algorithm developed in Section \ref{se:intelligent-epidemic}.


    




\subsection{A naive estimator}\label{se:naive estimate}
We first establish key notation. Let $\bb{\tau}_{k} = (\tau_{k,1}, \ldots, \tau_{k,d}) \in [0,1]^d$ for $k=1,2$, and denote the true rectangular anomalous patch by $I_0 := \prod_{j=1}^d [n_j \tau_{1, j}, n_j \tau_{2,j}]$. For expositional simplicity, we assume $n_j\tau_{k,j} \in \N$, ignoring negligible fractional rounding since $n^{-1}\lfloor n\gamma \rfloor \sim \gamma$ asymptotically. We start with the ``naive'' least-squares estimator as: 
\begin{equation}\label{eq:ls}
    \hat{I}_{LS}(\lambda_1, \lambda_2) = \argmin_{\substack{I \subset [\bb{n}] \\ n\lambda_1 < |I| < n\lambda_2}} \left( \sum_{k \in I} (X_k - \Bar{X}_{I})^2 + \sum_{k \in I^c} (X_k - \Bar{X}_{I^c})^2 \right),
\end{equation}
where $\lambda_1$ and $\lambda_2$ act as search-space parameters bounding the admissible patch volume. Central to our analysis is an equivalent formulation of \eqref{eq:ls}:
\begin{equation}\label{eq:lsspatial}
    \hat{I}_{LS}(\lambda_1, \lambda_2) = \argmax_{\substack{I \subset [\bb{n}] \\ n\lambda_1 < |I| < n\lambda_2}} \sqrt{\frac{|I|(n-|I|)}{n^2}} |\bar{X}_I - \bar{X}_{I^c}|. 
\end{equation}
This representation streamlines the error analysis for localizing $I_0$ by isolating the contrast between internal (the patch) and external sample means. Additionally, mirroring standard conventions in change-point localization, we assume the true anomaly is bounded away from the spatial domain boundaries:
\begin{assumption}\label{asmp:away-from-boundary}
    Let $\bb{a}$ and $\bb{b}$ denote the endpoints of the rectangle $I_0$. There exists a constant $c>0$ such that $\max\{|\bb{a}/ \bb{n}|_{\min}, 1 - |\bb{b}/\bb{n}|_{\max} \} \geq c$. 
\end{assumption}

\noindent Generalizing the techniques of \cite{bai1994} to spatial random fields, we establish the consistency of $\hat{I}_{LS}$.
\begin{theorem} \label{thm:epiconsistency}
    Consider the model in \eqref{eq:spatialmodel} satisfying Assumption \ref{asmp:doob} for some $p > \sqrt{2(d-1)} \vee 2$. Let $I_0$ satisfy Assumption \ref{asmp:away-from-boundary} and define $c_n = \min\{|I_0|/n, 1-|I_0|/n\}$. Provided $nc_n^2 \delta^2 \to \infty$ as $n \to \infty$, the estimator \eqref{eq:ls} satisfies
    \begin{equation}\label{eq:consistencyrate}
       |I_0 \ \Delta \ \hat{I}_{LS}(\mathcal{C}_0c_n, 1 - \mathcal{C}_1c_n)| = O_{\IP}(r_{n, \delta}^{-1}), \quad r_{n,\delta} := \delta^{\frac{2}{1- 2(d-1)/p^2}} \big(\max\{\log_2 \frac{\sqrt{n}}{\delta}, 1\}\big)^{- \frac{2}{p^2/(d-1) -2}},
    \end{equation}
    for sufficiently small constants $\mathcal{C}_0, \mathcal{C}_1 \in (0,1)$.
\end{theorem}
Theorem \ref{thm:epiconsistency} introduces several important nuances regarding convergence rates, sufficient conditions, and search space selection.
\begin{remark}[Convergence conditions and Assumption \ref{asmp:away-from-boundary}]
    The sufficient condition for convergence is $nc_n^2\delta^2 \to \infty$. When $|I_0| \asymp n$, this recovers the standard optimal condition $n\delta^2 \to \infty$. For context, \cite{wang2025optimal} achieves a localization rate of $(\max_k n_k)\delta^2 \gg \log n$, which is minimax optimal \textit{without} the boundary restrictions of Assumption \ref{asmp:away-from-boundary}. Comparable rates appear in Gaussian subgraph detection \citep{on-combi-testing, yu2022optimal} and time-series literature, where the minimax bound $n\delta^2 \gg \log n$ relaxes to $n\delta^2 \to \infty$, when change-points are bounded away from the edges \citep{wang2020univariate}. Crucially, while such minimax bounds typically rely on strict (sub-)Gaussian data, our results hold under the significantly milder Assumption \ref{asmp:doob}, completely avoiding any reliance on fast tail decay.
\end{remark}
Of particular interest is the fixed-alternative regime ($\delta \asymp 1$), where consistent localization requires $c_n \gg n^{-1/2}$. Conversely, for degenerate ``flat'' rectangles (i.e., $\min_{k\in [d]} |b_k - a_k|=0$), the patch volume satisfies $|I| \leq |\bb{n}|_{\min}$. Consistency here necessitates a diverging signal $\delta > n^{1/2} |\bb{n}|_{\min}^{-1}$. For instance, if $\bb{n} = t_0 \bb{c}$ for some $t_0 \in \N$ and $\bb{c} \in (0,1]^d$, consistency requires $\delta \to \infty$ whenever $d > 2$. Intuitively, the sparse signal from a lower-dimensional flat patch is easily overwhelmed by surrounding baseline noise. More generally, if $c_n \asymp n^{-\gamma}$, the signal must be sufficiently strong such that $\gamma < \frac{1}{2} + \frac{\log \delta}{\log n}$.

\begin{remark}[Duality of anomalous patch localization]
    The rate $nc_n^2 \delta^2$, governed by $c_n = \min\{|I_0|/n, 1-|I_0|/n\}$, reveals an interesting \textit{duality}: localization is equally difficult whether the patch is vanishingly small ($|I_0|/n \to 0$) or overwhelmingly large ($|I_0|/n \to 1$). While upper bounds on patch size have appeared in the literature \citep[e.g., Assumption 2.iii in][]{wang2025optimal}, this explicit duality is rarely emphasized, though minimum size constraints occasionally appear \cite{walther2010, sharpnackejs}. Conceptually, if an anomaly dominates the spatial domain, the baseline effectively corresponds to the true ``anomalous'' patch.
\end{remark}
\begin{remark}(Impact of spatial dimension and convergence rates)
    The rate $r_{n, \delta}$ reflects an impact of $d$. For $d=1$ (dependent time series), Theorem \ref{thm:epiconsistency} recovers the standard $\delta^{-2}$ localization rate. See \cite{huskova95, bucchia_jspi} Further, for a fixed $d$, if the spatial field $(\varepsilon_{\bb{i}})_{\bb{i} \in \Z^d}$ possesses sufficiently many moments, such as sub-Weibull tails \citep{kontorovich2014concentration}, we again recover the near-optimal $\delta^{-2}$ rate, up to logarithmic factors.
\end{remark}
 The primary limitation of \eqref{eq:lsspatial} is computational: an exhaustive search over a grid $\prod_{k=1}^d \{1, \cdots, n_0\}$ requires $O(n_0^{2d})$ operations, which is severely prohibitive for real-world applications.
\subsection{Intelligent Sampling for Spatial Patch Localization}\label{se:intelligent-epidemic}
To overcome the $O(n_0^{2d})$ computational bottleneck of $\hat{I}_{LS}$, we develop an efficient sub-sampling algorithm to localize $I_0$. This method will also serve as a foundational building block for our subsequent multi-patch algorithm. 
Building on the ``intelligent sampling'' concept introduced by \cite{lu2017intelligent} for univariate time series, we substantially generalize the framework to multi-dimensional spatial random fields. The core idea relies on a two-stage process: first, we apply the naive estimator to a coarsely sampled spatial grid to identify high-probability candidate regions containing the patch boundaries. Second, we restrict the refined search space exclusively to these localized subsets. This strategy yields massive computational speedups without sacrificing statistical accuracy.

\begin{figure}[htbp]
    \centering
    \includegraphics[width=\linewidth]{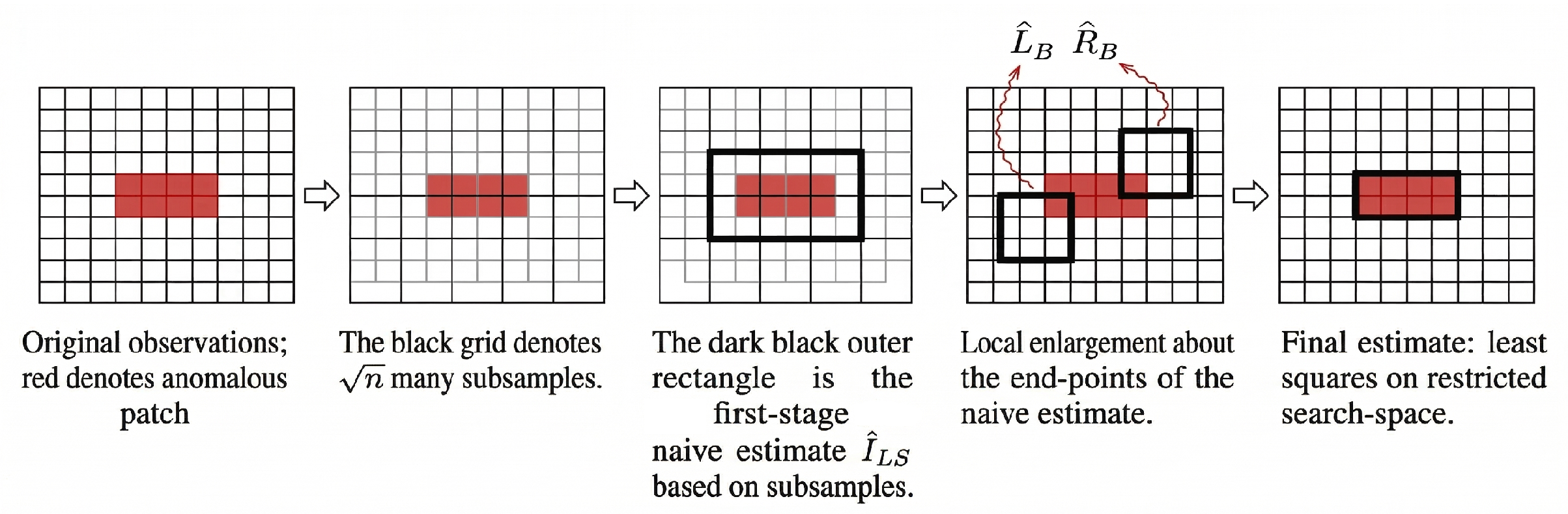}
    \caption{Illustration of the two-stage intelligent subsampling procedure detailed in Algorithm \ref{algo:subsample}.}
    \label{fig:algo-1}
\end{figure}

To build intuition before formalizing the general algorithm, consider the $d=1$ case under a fixed-alternative regime ($\delta \geq c > 0$). We first construct a coarse subsample $\mathcal{Y} = \{X_1, X_{\lfloor \sqrt{n} \rfloor}, X_{\lfloor 2\sqrt{n} \rfloor}, \ldots\}$ of size $O(\sqrt{n})$. Computing the naive estimate $\hat{I}_1 := \hat{I}_{LS}(\mathcal{Y})$ requires only $O(n)$ operations, and by Theorem \ref{thm:epiconsistency}, its endpoints are guaranteed to lie within $O(\sqrt{n})$ of the true anomaly $I_0$.
Therefore, restricting our second-stage search to $O(\sqrt{n}\log n)$ neighborhoods around the endpoints of $\hat{I}_1$ guarantees optimal localization of $I_0$ with high probability. This restricted search reduces the second-stage complexity to just $O(n \log^2 n)$, a stark improvement over the $O(n^2)$ naive approach. Generalizing this two-stage refinement strategy to $d$-dimensional fields yields the proposed algorithm.
Figure \ref{fig:algo-1} provides a schematic illustration of Algorithm \ref{algo:subsample}.
\begin{algorithm}[H]
\caption{Single spatial patch localization}\label{algo:subsample}
\begin{algorithmic}[1]
\STATE \textbf{Input:} $X=(X_{\bb{i}})_{\bb{i} \in [\bb{n}]}$, $\alpha, \kappa$.
\FOR{$k = 1$ \TO $d$}
  \STATE $L_k \leftarrow \lfloor n_k^{\alpha} \rfloor$, \quad $M_k \leftarrow \lceil n_k/L_k \rceil$.
  \STATE $\mathcal{Y}_k \leftarrow \{ (s-1)L_k + 1: s\in [M_k] \}$.
\ENDFOR
\STATE Sub-sampled dataset: $\mathcal{Y} \leftarrow \{X_{\bb{i}}: \bb{i}\in \prod_{k=1}^d \mathcal{Y}_k\}$, \quad $m \leftarrow |\mathcal{Y}|$.
\STATE $\hat{I}_{[a_I, b_I]}\leftarrow$ the preliminary \textit{naive} estimate based on $\mathcal{Y}$.
\STATE $\hat{L}_B \leftarrow \prod_{k=1}^d [L_k a_{I,k}- CL_k\bb{n}_k^{\kappa}(\log n)^{1/d} , \ L_ka_{I,k}+CL_k \bb{n}_k^{\kappa}(\log n)^{1/d}]$.
\STATE $\hat{R}_B \leftarrow \prod_{k=1}^d[L_kb_{I,k}- CL_k \bb{n}_k^{\kappa}(\log n)^{1/d}, \  L_kb_{I,k}+CL_k\bb{n}_k^{\kappa} (\log n)^{1/d}]$.
\STATE $\displaystyle
\tilde{I}:= \argmax_{\bb{i} \in \hat{L}_B,\, \bb{j} \in \hat{R}_B}
\sqrt{\frac{|\bb{j}-\bb{i}|(n-|\bb{j}-\bb{i}|)}{n^2}}
\big|\bar{X}_{I_{[\bb{i}, \bb{j}]}} - \bar{X}_{I_{[\bb{i}, \bb{j}]}^c}\big|
$.
\end{algorithmic}
\end{algorithm}

In particular, the validity and consistency of $\tilde{I}$ based on Algorithm \ref{algo:subsample} depend heavily on the coarseness of the initial grid. Intuitively, a coarser subsample degrades the first-stage estimate, increasing the risk that the local enlargement sets $\hat{L}_B$ and $\hat{R}_B$ fail to capture the true rectangle's endpoints. Conversely, the initial grid size directly dictates the computational speedup of Algorithm \ref{algo:subsample} over $\hat{I}_{LS}$, introducing a fundamental computational-statistical trade-off. We formalize this trade-off in the remarks following our next result, which establishes the consistency of Algorithm \ref{algo:subsample}.
\begin{theorem} \label{thm:single-consistency}
    Consider the model in \eqref{eq:spatialmodel}. Let $\bb{n}$ be sufficiently large such that $r_{n, \delta} \geq (\min_k n_k)^{-\kappa}$ for some $\kappa>0$, and assume $n^{1-\alpha}c_n^2 \delta^2 \to \infty$ as $n \to \infty$. Let $\tilde{I}$ be the output of Algorithm \ref{algo:subsample} with parameters $\alpha, \kappa > 0$ satisfying $\alpha + \kappa < 1$. Under the assumptions of Theorem \ref{thm:epiconsistency}, it holds that
    \begin{equation}
        |\tilde{I} \Delta I_0| = O_\IP(r_{n, \delta}^{-1}).
    \end{equation}
\end{theorem}
The parameter $\alpha$ in Theorem \ref{thm:single-consistency} governs a crucial statistical-computational tradeoff. While a smaller $\alpha$ relaxes the sufficient conditions for statistical consistency, a larger $\alpha$ reduces the computational cost of the preliminary estimate $\hat{I}_{[a_I, b_I]}$ at the expense of a more costly second-stage refinement for $\tilde{I}$. We formalize this tradeoff below.
\begin{remark}[Optimal choice of $\alpha$]\label{rem:single}
    For simplicity, assume $p \gg \sqrt{d}$ (so $r_{n,\delta} \asymp \delta^2$ up to logarithmic factors) and uniform dimensions $n_1 \asymp \ldots \asymp n_d$. Algorithm  \ref{algo:subsample}'s computational complexity is $O(n^{2(1-\alpha)} + n^{2(\alpha + \kappa)} \log^2 n)$. If the patch size scales as $c_n \asymp n^{-\gamma}$ for some $\gamma \in (0,1)$, consistency requires $\alpha \in (0, 1 - 2\gamma + 2\frac{\log \delta}{\log n})$. Balancing computational efficiency with this statistical constraint yields the optimal choice $\alpha^\star$:
    \[
    \alpha^{\star} = \begin{cases}
        \frac{1-\kappa}{2}, & \gamma \in \big(0, \frac{1+\kappa}{4} + \frac{\log \delta}{\log n} \big), \\
        1 - 2\gamma + 2\frac{\log \delta}{\log n}, & \gamma \in \big[\frac{1+\kappa}{4} + \frac{\log \delta}{\log n}, \frac{1}{2} + \frac{\log \delta}{\log n}\big). 
    \end{cases}
    \]
    Since $\delta^2 \gtrsim n^{-\kappa/d}$, it follows that $\frac{1+\kappa}{4} + \frac{\log \delta}{\log n} \geq \frac{1}{4} + \kappa(\frac{1}{4} - \frac{1}{2d})$. We detail the following observations for regimes where the patch size is sufficiently large ($\gamma \in (0,1/4)$). Note that this still accommodates vanishing patches ($c_n \to 0$ when $\gamma > 0$), provided they do not vanish too rapidly.\\
$\bullet$ For $d=2$, arguably the most common practical setting, we have $\alpha^{\star} = \frac{1-\kappa}{2}$ for all $\gamma \in (0,1/4)$, provided $\delta^2 \gtrsim n^{-\kappa/2}$ for $\kappa>0$. Consequently, Algorithm \ref{algo:subsample} achieves a computational complexity of $O(n^{1+\kappa})$. This represents a massive speed-up over the naive estimator $\hat{I}_{LS}$, while preserving the optimal statistical consistency rate for realistically sized anomalous patches. \\
$\bullet$ As expected, computational complexity naturally increases as the signal strength $\delta$ decreases (which corresponds to a larger $\kappa$). \\
$\bullet$ Under a fixed alternative ($\delta \asymp 1$), we set $\kappa\approx 0$, making the ideal first-stage block length $L_k(I) \approx \sqrt{n_k}$. In this regime, Algorithm \ref{algo:subsample} achieves a near-linear runtime of $O(n)$ up to logarithmic factors. In contrast, under similar settings, existing approaches \citep[e.g.,][]{wang2025optimal} require $O(n^{3/2})$ operations for $d=2$, often without explicit theoretical guarantees for consistency.
\end{remark}
In summary, while Algorithm \ref{algo:subsample} successfully leverages subsampling to yield a computationally efficient and provably valid estimator, its current formulation fundamentally assumes the presence of only a single anomalous patch. Since the true number of anomalies is rarely known \textit{a priori} in practice, this framework must be extended. Building on the foundations established in Section \ref{se:intelligent-epidemic}, the subsequent section develops a generalized algorithm capable of localizing multiple spatial patches.

\section{Multiple spatial patch localization }\label{se:multiple}

Next, we address the generalized multi-patch localization problem introduced in \eqref{eq:meanfield_deltas}. To leverage the computational and statistical efficiency of Algorithm \ref{algo:subsample}, we decompose the multiple-patch problem into several disjoint single-patch localization tasks that can be solved in parallel.  This decomposition relies on a preliminary block-based testing procedure. Despite complex spatial dependence, functional central limit theorems typically yield an asymptotic Gaussian structure for the random field. Assuming for simplicity that the long-run variance is known, we derive an asymptotic threshold for a coarse screening step. This screening isolates a set of disjoint candidate regions, each containing a single true anomalous patch. We can then apply Algorithm \ref{algo:subsample} piecemeal to each candidate region. We introduce our proposed method, for multiple patch detection, as Spatial Patch Localization of Anomalies under DEpendence (SPLADE), detailing its practical nuances and formal theoretical guarantees below.

Figure \ref{fig:algo-2} provides a schematic of Algorithm \ref{algo:multiple-subsample}. Assuming the true anomalous patches are sufficiently large and \textit{well-separated} (for instance, overlapping along the y-axis but separated along the x-axis as in Figure \ref{fig:algo-2}, formalized later in Assumption \ref{ass:min-sep}) the key steps of Algorithm \ref{algo:multiple-subsample} proceed as follows:
\begin{itemize}
    \item Let $\mu_0$ denote the baseline mean outside the anomalies. The algorithm begins with a block-based testing strategy using a threshold $\mathbf{Q}$. The sample space $[\bb{n}]$ is partitioned into $n^{1-\alpha}$ equal-sized rectangles, and a simultaneous test of $\IE[X_{\bb{i}}]=\mu_0$ is performed on each block. Since the patches are large relative to the block size, uniform Gaussian approximations (formalized in Assumption \ref{ass:fclt}) ensure that all blocks substantially overlapping with an anomaly are flagged with probability approaching one.
    \item Because the threshold controls a pre-specified Type I error, isolated false positive blocks will naturally occur. However, the probability of these false positives forming large connected components vanishes as $n$ increases. Thus, we isolate only connected regions of blocks, denoted $C_j$, that contain a sufficient number of samples. Given the separation assumption, each selected region $C_j$ captures exactly one anomalous patch, simultaneously yielding an accurate estimate of the total number of anomalies.
    \item Each connected region $C_j$ is then enveloped by a slightly larger rectangle $D_j$ such that all $D_j$ remain disjoint. Since each $D_j$ isolates a single true anomaly, we can deploy Algorithm \ref{algo:subsample} in parallel across all $D_j$ to accurately localize the patches.
\end{itemize}
\begin{figure}
    \centering    \includegraphics[width=0.8\linewidth]{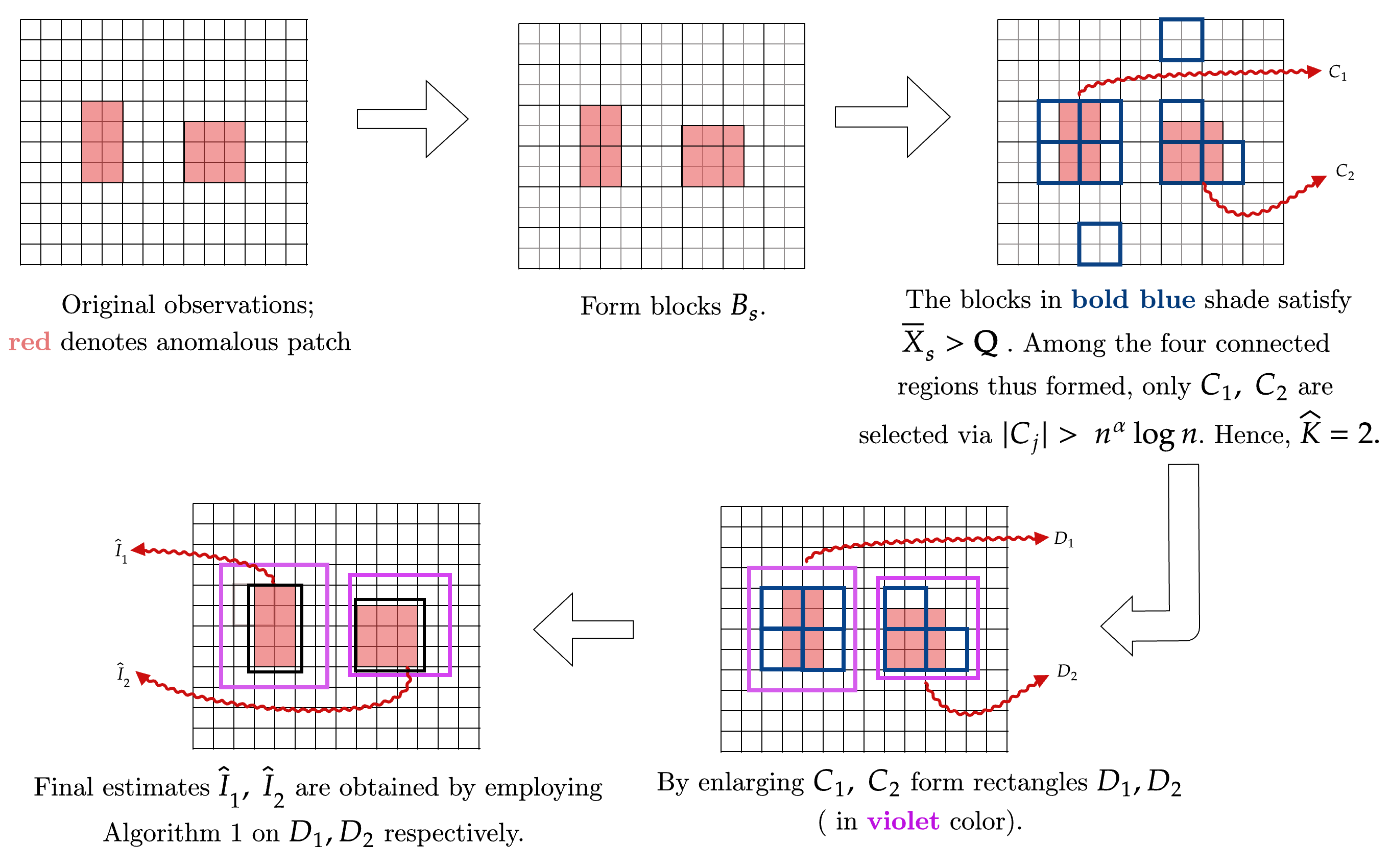}
    \caption{Illustration of the workings of Algorithm \ref{algo:multiple-subsample}.}
    \label{fig:algo-2}
\end{figure}
\begin{remark}[Computational complexity of Algorithm \ref{algo:multiple-subsample}]
As in Remark \ref{rem:single}, assume uniform dimensions $n_1 \asymp \ldots \asymp n_d \asymp n^{1/d}$. The first-stage testing requires $O(n)$ computations, and the connected components can be identified in $O(|M|)$ time using standard graph traversal techniques \citep[e.g.,][]{cormen2022introduction}. Suppose the mean-shift $\delta_j$ of each patch $I_j$ satisfies $\delta_j^2 \gg n^{-\kappa_j/d}$. Applying Algorithm \ref{algo:subsample} with parameters $(\alpha_j, \kappa_j)$ independently to each enveloping region $D_j$ yields a computational complexity of
$$O\bigg(\sum_{j=1}^K |I_j|^{2(1-\alpha_j)} + |I_j|^{2(\alpha_j+\kappa_j)} \log^2 n \bigg).$$
Choosing the optimal $\alpha_j^\star = \frac{1- \kappa_j}{2}$ in accordance with Remark \ref{rem:single} reduces the total complexity to $O(n + \sum_{j=1}^K |I_j|^{1+\kappa_j})$. If the number of anomalous patches $K = O(1)$, applying the trivial bound $\max_j |I_j| \leq n$ simplifies this to $O(n^{1 + \max_j \kappa_j})$. 
\begin{algorithm}[H]
\caption{SPLADE: Spatial Patch Localization of Anomalies under DEpendence}\label{algo:multiple-subsample}
\begin{algorithmic}[1]
\STATE \textbf{Input:} $X=(X_{\bb{i}})_{\bb{i} \in [\bb{n}]}$, block length parameter $\alpha$, first-stage threshold $\mathbf{Q}$.
\FOR{$k = 1$ \TO $d$}
  \STATE $L_k \leftarrow \lfloor n_k^{\alpha} \rfloor$, \quad $M_k \leftarrow \lceil n_k/L_k \rceil$.
  \STATE For $s\in [M_k]$, set
  \STATE \hspace{1.5em}$I_k(s) := \big\{ (s-1)L_k + 1,\,(s-1)L_k + 2,\,\dots,\,\min\{sL_k,\,n_k\} \big\}$.
\ENDFOR
\STATE Form blocks $B_{\bb{s}} := \bigotimes_{k=1}^d I_k(s_k)$ for $\bb{s}=(s_1,\dots,s_d)\in \bigotimes_{k=1}^d [M_k]$.
\STATE Set $\lvert B_{\bb{s}} \rvert = \prod_{k=1}^d \Big(\min\{s_k L_k,\,n_k\} - (s_k-1)L_k\Big)$.
\STATE Define block means $\overline{X}_{\bb{s}} := \frac{1}{\lvert B_{\bb{s}} \rvert}\sum_{\bb{i}\in B_{\bb{s}}} X_{\bb{i}}$, \; $\bb{s}\in \bigotimes_{k=1}^d [M_k]$.
\STATE Initialize $\tilde{M}$.
\FOR{each $\bb{s}\in \bigotimes_{k=1}^d [M_k]$}
  \IF{$|\bar{X}_{\bb{s}}| > \mathbf{Q}$}
    \STATE $\Tilde{M} \gets B_{\bb{s}}$.
  \ENDIF
\ENDFOR
\STATE Let $\{C_1,\dots,C_{\hat{K}}\}$ be the connected components of $\tilde{M}$, with $\min_j |C_j|> n^{\alpha} \log n$.
\FOR{$j=1$ \TO $\hat{K}$}
  \FOR{$k=1$ \TO $d$}
    \STATE $\ell_k^j \gets \min_{\bb{s}\in C_j} s_k$; \quad $r_k^j \gets \max_{\bb{s}\in C_j} s_k$.
  \ENDFOR
  \STATE $\bb{l}_j\gets [L_1 \ell_1^j - c L_1 \log^{} n, \ldots, L_d \ell_d^j - c L_d \log^{} n]$, \quad
  $\bb{r}_j\gets [L_1 r_1^j + c L_1 \log^{} n, \ldots, L_d r_d^j + c L_d \log^{} n]$.
  \STATE $D_j \gets I_{[\bb{l}_j, \bb{r}_j]}$, \quad $d_j \gets |D_j|$.
  \STATE $\displaystyle
  \hat{I}_j \gets \text{Algorithm 1}(X_{\bb{i}}; \bb{i}\in D_j).
  $
\ENDFOR
\RETURN number of patches $\hat{K}$; estimated patches $\hat{I}_j, j \in [\hat{K}]$.
\end{algorithmic}
\end{algorithm}
Crucially, for fixed alternatives where all $\delta_j \asymp 1$ (implying $\kappa_j=0$), the overall runtime is strictly $O(n)$. To the best of our knowledge, this is the \textit{only algorithm that achieves linear-time computation} for the spatial anomaly localization problem. In contrast, existing approaches such as \cite{wang2025optimal} require $O(n^{3/2})$ operations even when $K$ is bounded by a constant.
\end{remark}
As anticipated in our informal discussion, the validity of Algorithm \ref{algo:multiple-subsample} hinges on two key conditions: (i) the anomalous patches must be well-separated, and (ii) a sufficiently sharp Gaussian approximation must hold to enable simultaneous testing of the first-stage block means. We formalize these assumptions below.
\begin{assumption}[Minimum separation]\label{ass:min-sep}
    For any two rectangles $I_{[\bb{i}_1, \bb{i}_2]}$ and $I_{[\bb{j}_1, \bb{j}_2]}$ in $[\bb{n}]$, define the pseudo-metric 
    \[
    \rho(I_{[\bb{i}_1, \bb{i}_2]}, I_{[\bb{j}_1, \bb{j}_2]}) = \max_{k\in [d]} \max\{0, j_{1,k} - i_{2,k}, i_{1,k} - j_{2,k}\},
    \]
    and let $\nu(I_{[\bb{i}_1, \bb{i}_2]}, I_{[\bb{j}_1, \bb{j}_2]})$ denote the corresponding maximizer over $k \in [d]$. There exists an $\alpha \in (0,1)$ such that the $K$ disjoint anomalous patches $I_1, \ldots, I_K$ satisfy
    \[
    \min_{j\neq k} \rho(I_j, I_k) \geq c_0 n_{\nu^\star_{jk}}^{\alpha}\log n,
    \]
    where $\nu_{jk}^\star = \nu(I_j, I_k)$ and $I_j = \prod_{k=1}^d [n_k\tau_{1,k}^j, n_k\tau_{2,k}^j]$.
\end{assumption}
Assumption \ref{ass:min-sep} ensures that any pair of anomalous patches is sufficiently separated along at least one of the $d$ coordinate axes. The separation parameter $\alpha$ strictly governs the first-stage block-based testing procedure in Algorithm \ref{algo:multiple-subsample}. Specifically, if we partition the domain into rectangular blocks with side lengths $n_k^\alpha$ for $k \in [d]$, Assumption \ref{ass:min-sep} guarantees that any two distinct patches $I_j$ and $I_k$ are separated by at least $\Theta(\log n)$ many blocks along their axis of maximal separation, $\nu^\star_{jk}$.
\begin{assumption}[Uniform Gaussian approximation]\label{ass:fclt}
    Let $\mathcal{A}$ be the collection of rectangles in $[0,1]^d$. For some $n_0 \in \N$, define $A_{n_0}(A) := \sum_{\bb{i}\in \{1, \ldots, n_0\}^d} n_0 A \cap I_{[\bb{i}-1, \bb{i}]}$. In a possibly enriched probability space, there exists a standard Brownian sheet $\mathbb{W}$ on $[0,1]^d$ such that 
    \begin{equation}
        \sup_{A \in \mathcal{A}} \big| S_n(A_{n_0}(A)) - \sigma n_0^{d/2} \mathbb{W}(A) \big| = o_{\IP}(n_0^{d/q}), \label{eq:fclt}
    \end{equation}
    for some $q>2$, where $\sigma>0$ is defined as in Assumption \ref{asmp:lrv}, and $S_n(A_{n_0}(A)) = \sum_{\bb{i}\in A_{n_0}(A)} \varepsilon_{\bb{i}}$.
\end{assumption}

Assumption \ref{ass:fclt} plays a central role in our subsequent analysis by enabling a simultaneously valid, block-based testing procedure. Next, we discuss its implications.
\begin{remark}[Discussion on Assumption \ref{ass:fclt}]
    Assumption \ref{ass:fclt} functions as a strengthened \textit{weak invariance principle} (or functional central limit theorem) for spatial random fields. The $q=2$ case in \eqref{eq:fclt} is well-established for a broad class of stationary spatial random fields satisfying \eqref{eq:causalspatial} \citep{el2013central, bucchia_jspi}, building upon a rich history of classical results \citep{wichura1969inequalities, poghosyan1998invariance, doi:10.1142/6555}. 
    Stronger uniform Gaussian approximations requiring $q>2$ have been extensively developed primarily for the $d=1$ (time series) setting. Beginning with the seminal work of \citet{komios1975approximation}, a large body of literature has established optimal exponents $q=p$ (with $p$ defined in Assumption \ref{asmp:doob}) for general stationary time series \citep{sakhanenko1984, sakhanenko1989, zaitsev1998, Sakhanenko2006, gotzezaitsev, LiuLin2009, zhouzhousinica, Wu14, KarmakarWu2020, bonnerjeekarmakarwu}. 
    
    Given this substantial evidence in one dimension, we adopt Assumption \ref{ass:fclt} as a standing condition for spatial fields ($d>1$). We do not strictly require the optimal exponent $q=p$, although such rates may well be attainable in higher dimensions. Rigorously establishing these uniform approximations for $d>1$ remains a highly non-trivial open problem in probability theory, though recent advances in multiscale and high-dimensional spatial approximations \citep[e.g.,][]{proksch2018multiscale, shao_cck} offer promising steps towards this direction.
\end{remark}
The next result establishes the theoretical consistency guarantees of Algorithm \ref{algo:multiple-subsample}.
\begin{theorem}\label{thm:multiple}
    Suppose Assumptions \ref{asmp:doob}, \ref{asmp:lrv}, and \ref{ass:fclt} hold for the underlying spatial random field $(\varepsilon_{\bb{i}})_{\bb{i}\in \Z^d}$. Further, assume the baseline mean is $\mu_0=0$. Recall the definition of blocks $B_{\bb{s}}$ in Algorithm \ref{algo:multiple-subsample}. Let $\mathbf{Q}$ be the $(1-\kappa)$-th quantile of $\max_{\bb{s}\in \prod_{k=1}^d [M_k(I)]} |B_{\bb{s}}|^{-1}\big|\sigma \mathbb{W}(B_{\bb{s}})\big|$ for some $\kappa \in (0,1)$. Consider the model \eqref{eq:spatialmodel} with a bounded number of anomalous patches $K \leq C$ for some constant $C>0$, and let these patches $\{I_j\}_{j\in [K]}$ satisfy Assumption \ref{asmp:away-from-boundary}.
    \begin{itemize}
        \item \textit{(First stage guarantee)} Assume that the parameter $\alpha \in (0,1)$ in the first stage of Algorithm \ref{algo:multiple-subsample} is small enough to satisfy Assumption \ref{ass:min-sep} and 
        \begin{equation}
            \min_{j\in[K]} \min_{k\in[d]} n_k^{1-\alpha}|\tau_{2,k}^j - \tau_{1,k}^j| = \Theta(\log^{1/d}n), \label{eq:min-size}
        \end{equation} 
        yet large enough to satisfy
        \begin{equation}
            \min \bigg\{ \frac{n^{\alpha}}{\log n}\min_{j \in [K]} \delta_j^2, \ n^{\alpha/2}\sqrt{\log n} |\bb{n}|_{\infty}^{-d/q} \bigg\} \to \infty \quad \text{as } n\to \infty, \label{eq:jump-min}
        \end{equation}
        where $q>2$ is as defined in Assumption \ref{ass:fclt}. Then, for Algorithm \ref{algo:multiple-subsample}, it holds that
        \begin{equation}
            \IP(\hat{K}=K)\to 1 \quad \text{as } n \to \infty. \label{eq:correct-number-est}
        \end{equation}
        \item \textit{(Second stage guarantee)} For each $j\in [K]$, suppose there exists $\kappa_j>0$ such that $r_{|I_j|, \delta_j} \gg (\min_{k \in [d]}|I_j|_k)^{-\kappa_j}$. Moreover, suppose each individual application of Algorithm \ref{algo:subsample} within Algorithm \ref{algo:multiple-subsample} is performed with parameters $\alpha_j$ small enough such that $\alpha_j + \kappa_j < 1$, and
        \begin{equation}
            \min_{j\in [K]} \frac{n^{2\alpha} \log^{3d} n}{|I_j|^{1+\alpha_j}}\delta_j^2 \to \infty. \label{eq:indv-I_j-tradeoff}       
        \end{equation}
        Then, for any $\eta>0$, there exists an $M_{\eta}>0$ such that the estimates $\hat{K}$ and $\hat{I}_j$ (for $j\in [\hat{K}]$) from Algorithm \ref{algo:multiple-subsample} satisfy
        \begin{equation}
            \IP\big( |I_j \Delta \hat{I}_j| > M_{\eta}r_{n, \delta_j}^{-1} \text{ for all } j\in [K] \ \big| \ \hat{K}=K \big) < \eta. \label{eq:multpl-consistency}
        \end{equation} 
    \end{itemize}
\end{theorem}
Theorem \ref{thm:multiple} ensures that Algorithm \ref{algo:multiple-subsample} accurately estimates the true number of anomalous patches while localizing each individual patch at the optimal rate $r_{n, \delta_j}$, matching the performance of the naive estimator. To the best of our knowledge, this establishes Algorithm~\ref{algo:multiple-subsample} as the \textit{only method capable of consistently localizing rectangular anomalies under spatial dependence, while potentially achieving an $O(n)$ computational complexity}.

The theoretical guarantees of Theorem \ref{thm:multiple} rely on specific technical conditions governing the choice of $\alpha$ in the first stage and the parameters $\{\alpha_j, \kappa_j\}_{j=1}^K$ in the second stage. We detail these requirements below. Throughout this discussion, we treat the true number of anomalies $K$ as fixed relative to $n$ and $\bb{n}$, and we condition on the event that $\hat{K}=K$.
\begin{remark}[Choice of first-stage block-size parameter $\alpha$]\label{rem:mult-alpha}
    The choice of $\alpha$ is closely tied to the true patch sizes $|I_j|$ and their mean-shifts $\delta_j$. Following Remark \ref{rem:single}, assume uniform dimensions $n_1 \asymp \ldots \asymp n_d \asymp n^{1/d}$, and suppose the relative boundaries $\tau_{1,k}^j, \tau_{2,k}^j$ are independent of $n$. Under these conditions, Assumption \ref{ass:min-sep} and \eqref{eq:min-size} are trivially satisfied for any $\alpha \in (0,1)$, and \eqref{eq:jump-min} simplifies to
    \begin{equation}
        \alpha \in \bigg( \max\bigg\{\frac{2}{q}, -2\frac{\min_{j\in [K] }\log |\delta_j|}{\log n}\bigg\}, 1 \bigg). \label{eq:range of alpha}
    \end{equation}
    Equation \eqref{eq:range of alpha} requires $q>2$ (from Assumption \ref{ass:fclt}) and $\min_{j\in [K]} |\delta_j| \gg n^{-1/2}$. This is an extremely mild lower bound, ensuring the anomalous patches remain discernible at the scale of the lattice $\Z^d$. As we detail next, this choice of $\alpha$ also fundamentally constrains the individual patch localizations in the second stage of Algorithm \ref{algo:multiple-subsample}.
\end{remark}
\begin{remark}[Choice of second-stage tuning parameters $\alpha_j$]
    Analyzing \eqref{eq:indv-I_j-tradeoff} under the conditions of Remark \ref{rem:mult-alpha}, the parameters $\alpha$ and $\alpha_j$ must satisfy $n^{2\alpha-1 -\alpha_j}\delta_j^2 \to \infty$. This condition is strictly stronger than the $n\delta^2 \to \infty$ requirement in Theorem \ref{thm:epiconsistency}, reflecting the fundamental trade-off necessary to achieve computational efficiency. Ignoring logarithmic factors, this requires
    $$ \alpha_j \in \bigg(0, 2\bigg(\alpha + \frac{\log |\delta_j|}{\log n}\bigg) - 1\bigg), $$
    which implicitly necessitates $\alpha > 1/2 - (\log |\delta_j|)/\log n$. Our empirical ablation studies, presented in \S\ref{ssc:simuablation}, highlight the robustness of SPLADE across different choices of $\alpha$ for various light-tailed settings,  provided the theoretical constraints are satisfied. To further align with Theorem \ref{thm:single-consistency}, assume $|\delta_j| \asymp n^{-\kappa_j/(2d)}$ for $\kappa_j \in (0, 2d)$, and let $\min_{j \in [K]} \kappa_j = 0$ (meaning at least one patch exhibits a constant-order mean shift). For sufficiently large $n$, we have $1/q \gg \max_{j\in [K]} \kappa_j/(2d \log n)$, which simplifies the feasible parameter range to:
    \begin{equation}
        \alpha \in (2/q, 1) \quad \text{and} \quad \alpha_j \in \bigg(0, 2\bigg(\alpha - \frac{\kappa_j}{2d} - \frac{1}{2}\bigg)\bigg). \label{eq:final-choice}
    \end{equation}
    While the $O(n)$ complexity of the first stage is independent of $\alpha$, the efficiency of the second stage relies heavily on $\alpha_j$. Having argued $\kappa_j \approx 0$ in \S\ref{se:intelligent-epidemic}, in light of \eqref{eq:final-choice} we choose $\alpha_j=1/2$, which yields a fully linear-time algorithm. 
\end{remark}
 
The discussion on estimating the parameters $\mu_0$ and $\sigma$ is deferred to Appendix \ref{ssc:musigma}.

\section{Performance Evaluation}\label{sec:simu}

In this section, we provide empirical evidence corroborating the established theoretical guarantees. We first provide a sensitivity analysis on the stability of SPLADE across different choices of tuning parameters. Subsequently, we compare it extensively against three baselines: DCART \citep{madrid2021lattice}, a fused-lasso approach denoted as TV \citep{tansey2015fast}, and, in selected settings, DPLS-SAD \citep{wang2025optimal}. An implementation of SPLADE is publicly available at \href{https://github.com/sohamb01/SPLADE}{https://github.com/sohamb01/SPLADE}. We evaluate computational efficiency, localization accuracy, and robustness across varying grid sizes and anomalous patch layouts, spatial dependence structures, and signal strengths. Performance is assessed using the following metrics: the mean number of detected patches, the empirical probability of correctly estimating the true number of patches, the Adjusted Rand Index (ARI), the average runtime per iteration (in seconds\footnote{ All computations are run on 13th Gen intel(R) Core(TM) i9-13900K}), and the normalized Hausdorff distance defined as follows: For $\mathcal G := \{1,\dots,N\}\times\{1,\dots,N\},
\Lambda_0 := \mathcal G \setminus \bigcup_{k=1}^K \Lambda_k,
\widehat{\Lambda}_0 := \mathcal G \setminus \bigcup_{k=1}^{\widehat K} \widehat{\Lambda}_k$,
let $\mathcal C := \{\Lambda_k : 0 \le k \le K,\ \Lambda_k \neq \varnothing\},
\qquad
\widehat{\mathcal C} := \{\widehat{\Lambda}_k : 0 \le k \le \widehat K,\ \widehat{\Lambda}_k \neq \varnothing\}.$
For \(A,B \subseteq \mathcal G\), define the Jaccard distance
$d_J(A,B) := |A \triangle B|/|A \cup B|$
with the convention \(d_J(\varnothing,\varnothing)=0\). The normalized two-sided Hausdorff distance between \(\mathcal C\) and \(\widehat{\mathcal C}\) is
\[
d_H(\mathcal C,\widehat{\mathcal C})
:=
\max\Biggl\{
\max_{C\in\mathcal C}\min_{\widehat C\in\widehat{\mathcal C}} d_J(C,\widehat C),
\;
\max_{\widehat C\in\widehat{\mathcal C}}\min_{C\in\mathcal C} d_J(\widehat C,C)
\Biggr\}.
\] All reported results are averaged over 100 independent replicates across diverse experimental conditions.

\subsection{Sensitivity analysis for SPLADE}\label{ssc:simuablation}
In this section, we consider the following setting.

\begin{itemize}
    \item \textit{Grid sizes ($N \times N$):} For ablation, $N=500,750$ and $1000$.
   \item \textit{\textbf{Configuration 1} of anomalous patches:}  Layout of 3 true patches (Figure \ref{fig:3patches}).
   \begin{figure}[htbp]
    \centering
    \begin{subfigure}[t]{0.42\linewidth}
        \centering
        \includegraphics[width=\linewidth]{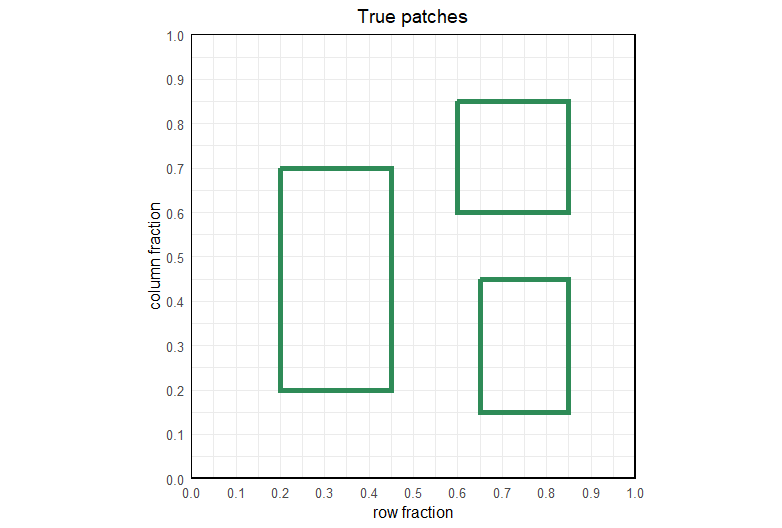}
        \caption{Individual jumps are $\delta_{\mu}$ (left), $\delta_{\mu}$ (top right), and $-\delta_{\mu}$ (bottom right)}
        \label{fig:3patches}
    \end{subfigure}
    \hfill
    \begin{subfigure}[t]{0.50\linewidth}
        \centering
        \includegraphics[width=0.76\linewidth]{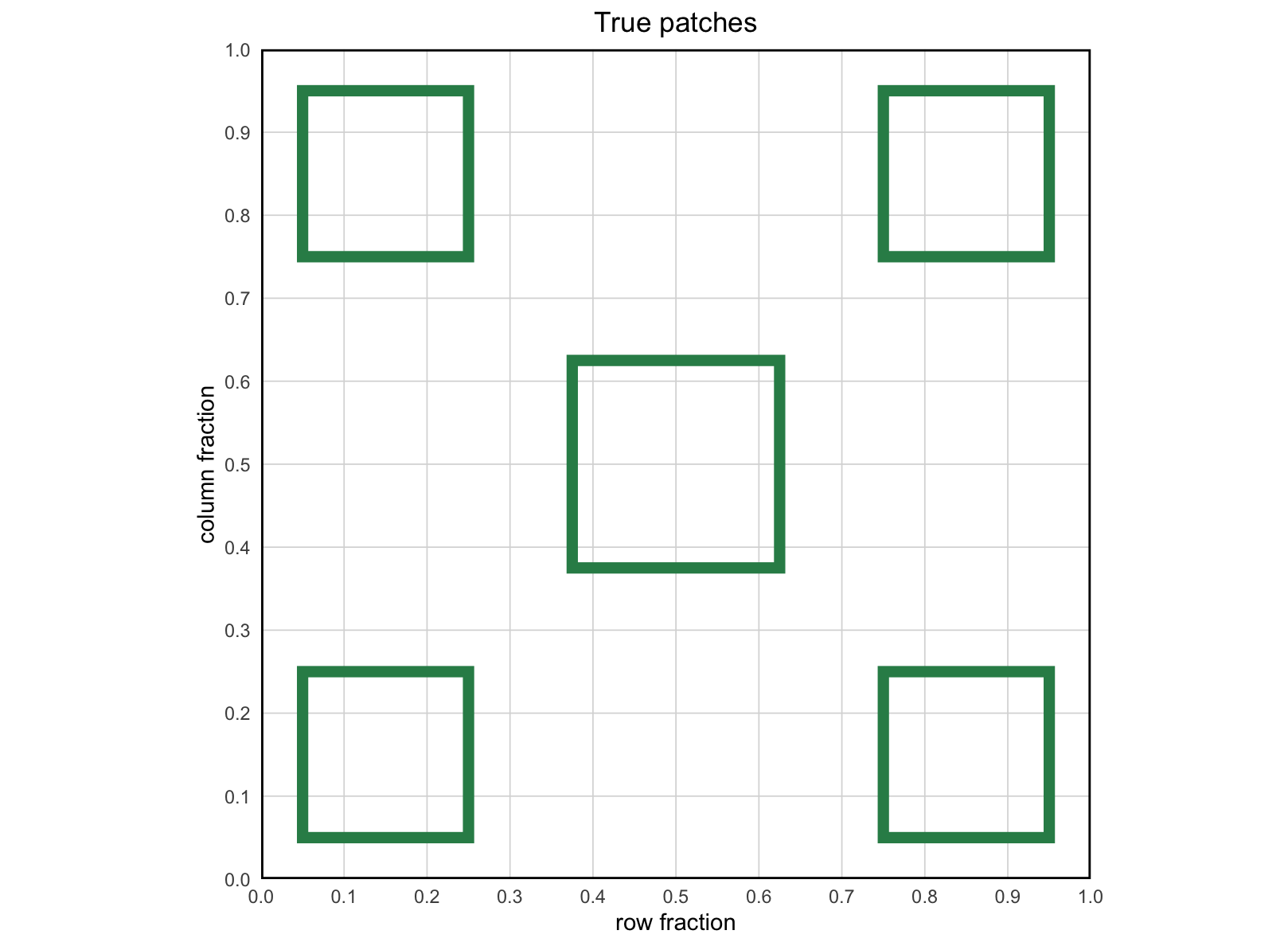}
        \caption{Individual jumps are $\delta_{\mu}$ (bottom left), $2\delta_{\mu}$ (top left),
        $3\delta_{\mu}$ (top right), $4\delta_{\mu}$ (bottom right), and $5\delta_{\mu}$ (center)}
        \label{fig:5patches}
    \end{subfigure}
    \caption{Illustration of anomalous patch configurations used for performance assessment.}
    \label{fig:multiple_patches}
\end{figure}   
    \item \textit{Spatial dependence structure:} We consider a SAR($\rho$) process defined as
\begin{equation}\label{eq:sar}
    \varepsilon_{\bb{i}}= \rho \sum_{\bb{j} \in \mathcal{N}(\bb{i})} w_{\bb{ij}}\varepsilon_{\bb{j}} + e_{\bb{i}},\end{equation}
 with $e_{\bb{i}} \stackrel{\mathrm{iid}}{\sim} N(0,1)$ and $w_{\bb{ij}}=\frac{\mathbf{1}\{\bb{j}\in\mathcal{N}(\bb{i})\}}{|\mathcal{N}(\bb{i})|},$ where neighborhood $\mathcal{N}(\bb{i})$ can have cardinality of 2, 3 or 4 depending on $\bb{i}$-th pixel's position at corner, edge or interior respectively. 
For the sensitivity analysis, we vary $\rho \in \{0.25,0.5 \}$, and keep the signal strength $\delta_\mu$ in the anomalous patches fixed at $1$. Some additional ablation studies for a non-linear distribution are deferred to Appendix \ref{ssc:simuablation_app}.
\end{itemize}
Across all experimental settings, we implement SPLADE (Algorithm \ref{algo:multiple-subsample}) with varying $\alpha \in \{0.4,0.5,0.6\}$, and we fix the second-stage tuning parameters (Algorithm \ref{algo:subsample}) to $\alpha_j=0.5$ for all $j \in [\hat{K}]$. The corresponding results are presented in Table~\ref{tab:splade_ablation_config1_sar}. SPLADE seems to perform equally well for all three choices of $\alpha$. This result is not surprising based on \eqref{eq:range of alpha} and \eqref{eq:final-choice}, whose prescribed ranges are being further widened by the light-tailed Gaussian distribution of the SAR($\rho$) errors. Importantly, our theoretical results require only the minimal assumption of a finite $p$-th moment, and so the theoretically motivated choices of $\alpha$ may, prima facie, seem somewhat conservative. Nevertheless, SPLADE is robust across a wide range of $\alpha$, highlighting its stability and further strengthening the consistent improvements in both accuracy and speed over competitors in \S\ref{ssc:simucomparative}.

\begin{table}[!htbp]
\centering
\scriptsize
\setlength{\tabcolsep}{4.5pt}
\renewcommand{\arraystretch}{1.15}
\caption{Sensitivity analysis of SPLADE (on $\alpha$ from Stage 1 of Algorithm \ref{algo:multiple-subsample}) for Config.~1 with $\delta_{\mu}=1$ under SAR $\rho$. Each cell reports average over 100 replicates in the order $\alpha=0.4$ / $\alpha=0.5$ / $\alpha=0.6$.}
\label{tab:splade_ablation_config1_sar}
\begin{tabular}{|c| c| c| c| c| c|}
\hline
$N$ & SAR($\rho$) & $\hat K$ mean & $I(\hat K=3)$ & ARI & Hausdorff \\
\hline
\multirow{2}{*}{500}
& 0.25
& 3.01 / 3.02 / 3.00
& 0.99 / 0.98 / 1.00
& 0.986 / 0.975 / 0.977
& 0.051 / 0.090 / 0.063 \\
\cline{2-6}
& 0.50
& 3.01 / 3.02 / 3.00
& 0.99 / 0.98 / 1.00
& 0.977 / 0.967 / 0.969
& 0.078 / 0.111 / 0.085 \\
\hline
\multirow{2}{*}{750}
& 0.25
& 3.00 / 3.00 / 3.00
& 1.00 / 1.00 / 1.00
& 0.982 / 0.985 / 0.983
& 0.034 / 0.044 / 0.050 \\
\cline{2-6}
& 0.50
& 3.00 / 3.00 / 3.00
& 1.00 / 1.00 / 1.00
& 0.979 / 0.984 / 0.977
& 0.049 / 0.048 / 0.065 \\
\hline
\multirow{2}{*}{1000}
& 0.25
& 3.00 / 3.00 / 2.98
& 1.00 / 1.00 / 0.96
& 0.993 / 0.996 / 0.980
& 0.018 / 0.014 / 0.059 \\
\cline{2-6}
& 0.50
& 3.00 / 3.00 / 2.97
& 1.00 / 1.00 / 0.95
& 0.991 / 0.991 / 0.976
& 0.024 / 0.033 / 0.071 \\
\hline
\end{tabular}
\end{table}

\subsection{Comparative studies}\label{ssc:simucomparative}

In this section, we provide a detailed comparison against other baseline approaches across diverse settings. In particular, we consider the following.

\begin{itemize}
    \item \textit{Grid sizes ($N \times N$):} Owing to scalability and other restrictions in the comparative study (eg. the restriction of $2^k \times 2^k$ for DCART), we choose $N=256$ and $512$.
    \item \textit{\textbf{Configuration 2} of anomalous patches:} In addition to Configuration 1 in \S\ref{ssc:simuablation}, we also consider a layout of 5 true patches (Figure \ref{fig:5patches}), closely mirroring Scenario 4 from DCART proposed in \citet{madrid2021lattice}.

    \item \textit{Spatial dependence structure:} Similar to \S\ref{ssc:simuablation}, we consider SAR($\rho$) process with 
 $\rho \in \{0.04, 0.4, 0.8\}$. Here, the $\rho=0.04$ approximates the i.i.d. setting assumed by DCART. Some additional simulations for a non-linear distribution are deferred to Appendix \ref{se:appendix simulation}.
    \item \textit{Signal strengths:} We vary the signal strength $\delta_{\mu} \in \{0.2, 0.4, 0.6, 0.8, 1\}$ for an exhaustive picture of the performance of SPLADE in both low and high SNR regimes.
\end{itemize}

\textcolor{black}{For this section, across all experimental settings, we implement SPLADE (Algorithm \ref{algo:multiple-subsample}) with the first-stage parameter set to $\alpha=0.5$, and we fix the second-stage tuning parameters (Algorithm \ref{algo:subsample}) to $\alpha_j=0.5$ for all $j \in [\hat{K}]$. } For the baseline methods, DCART and TV, we define the penalty parameter grids for $\lambda$ as $\{5, 6.78, \ldots, 30\}$ and $\{10^c : c \in \{-1, -0.785, \ldots, 3\}\}$, respectively, following the default recommendations in \citet{madrid2021lattice}.

\newcommand{\Khat}{\hat K}
\newcommand{\Pfive}{\Pr(\Khat=5)}
\newcommand{\Haus}{\text{Haus}_1}

\begin{table}[htbp]
\centering
\scriptsize
\setlength{\tabcolsep}{4.2pt}
\renewcommand{\arraystretch}{1.18}
\caption{Comparison of DCART, SPLADE and TV across grid sizes, jump sizes, and SAR $\rho$ for Config.\ 1 (Fig~\ref{fig:3patches}). Each cell reports average over 100 replicates in the order DCART / SPLADE / TV.}
\label{tab:grid256_512_rho_stacked_k3}
\resizebox{\textwidth}{!}{%
\begin{tabular}{|c|c|c|c|c|c|}
\hline
Jump & $\hat K$ & $I(\hat K=3)$ & ARI & Hausdorff distance & time/iter (sec) \\
\hline

\multicolumn{6}{|c|}{\textbf{Grid = $256 \times 256$}} \\
\hline
\multicolumn{6}{|c|}{\textbf{$\rho = 0.04$}} \\
\hline
0.2 & 4.78 / \textbf{3.66} / 1.34 & 0.09 / \textbf{0.43} / 0.03 & 0.187 / \textbf{0.490} / 0.006 & \textbf{0.86} / 0.86 / 0.96 & 8.88 / \textbf{2.58} / 4.92 \\
\hline
0.4 & 3.72 / \textbf{3.00} / 1.18 & 0.25 / \textbf{1.00} / 0.00 & 0.353 / \textbf{0.702} / 0.026 & 0.93 / \textbf{0.52} / 0.95 & 8.81 / \textbf{4.24} / 5.36 \\
\hline
0.6 & 2.83 / \textbf{3.00} / 1.31 & 0.31 / \textbf{1.00} / 0.03 & 0.275 / \textbf{0.792} / 0.057 & 0.95 / \textbf{0.38} / 0.95 & 8.85 / \textbf{5.05} / 5.87 \\
\hline
0.8 & 2.63 / \textbf{3.00} / 1.42 & 0.34 / \textbf{1.00} / 0.01 & 0.249 / \textbf{0.866} / 0.097 & 0.96 / \textbf{0.25} / 0.96 & 8.99 / \textbf{6.04} / 6.18 \\
\hline
1.0 & \textbf{2.99} / 2.99 / 1.38 & 0.37 / \textbf{0.99} / 0.04 & 0.274 / \textbf{0.886} / 0.140 & 0.97 / \textbf{0.20} / 0.95 & 8.98 / \textbf{6.37} / 6.72 \\
\hline
\multicolumn{6}{|c|}{\textbf{$\rho = 0.40$}} \\
\hline
0.2 & 7.74 / \textbf{1.68} / 1.32 & 0.04 / \textbf{0.12} / 0.04 & 0.040 / \textbf{0.108} / 0.005 & 0.99 / \textbf{0.95} / 0.95 & 16.38 / \textbf{1.21} / 8.85 \\
\hline
0.4 & 7.02 / \textbf{3.07} / 1.39 & 0.09 / \textbf{0.93} / 0.05 & 0.052 / \textbf{0.685} / 0.022 & 0.99 / \textbf{0.58} / 0.96 & 8.93 / \textbf{3.77} / 5.60 \\
\hline
0.6 & 6.21 / \textbf{3.00} / 1.37 & 0.07 / \textbf{1.00} / 0.04 & 0.013 / \textbf{0.789} / 0.049 & 1.00 / \textbf{0.39} / 0.96 & 9.00 / \textbf{4.89} / 6.07 \\
\hline
0.8 & 5.82 / \textbf{3.00} / 1.45 & 0.09 / \textbf{1.00} / 0.03 & 0.022 / \textbf{0.858} / 0.083 & 1.00 / \textbf{0.26} / 0.96 & 9.03 / \textbf{5.56} / 6.33 \\
\hline
1.0 & 5.26 / \textbf{3.00} / 1.40 & 0.14 / \textbf{1.00} / 0.05 & 0.028 / \textbf{0.892} / 0.123 & 1.00 / \textbf{0.20} / 0.95 & 9.12 / \textbf{6.04} / 6.83 \\
\hline
\multicolumn{6}{|c|}{\textbf{$\rho = 0.80$}} \\
\hline
0.2 & 28.95 / 0.10 / \textbf{2.00} & 0.00 / 0.00 / \textbf{0.24} & $-0.002$ / \textbf{0.002} / 0.001 & 1.00 / \textbf{0.94} / 0.98 & 36.00 / \textbf{0.14} / 6.09 \\
\hline
0.4 & 28.77 / 0.79 / \textbf{2.15} & 0.00 / 0.04 / \textbf{0.31} & $-0.002$ / \textbf{0.049} / 0.008 & 1.00 / \textbf{0.94} / 0.98 & 36.05 / \textbf{0.37} / 6.43 \\
\hline
0.6 & 28.51 / \textbf{2.92} / 2.00 & 0.00 / \textbf{0.30} / 0.24 & $-0.002$ / \textbf{0.289} / 0.019 & 1.00 / \textbf{0.93} / 0.98 & 36.17 / \textbf{1.38} / 6.60 \\
\hline
0.8 & 28.25 / \textbf{3.66} / 2.03 & 0.00 / \textbf{0.43} / 0.22 & $-0.002$ / \textbf{0.656} / 0.034 & 1.00 / \textbf{0.75} / 0.98 & 36.20 / \textbf{3.17} / 6.82 \\
\hline
1.0 & 28.12 / \textbf{3.21} / 2.30 & 0.00 / \textbf{0.80} / 0.28 & $-0.003$ / \textbf{0.815} / 0.052 & 1.00 / \textbf{0.43} / 0.99 & 36.71 / \textbf{4.62} / 7.18 \\
\hline

\multicolumn{6}{|c|}{\textbf{Grid = $512 \times 512$}} \\
\hline
\multicolumn{6}{|c|}{\textbf{$\rho = 0.04$}} \\
\hline
0.2 & \textbf{3.60} / 5.71 / 1.10 & \textbf{0.33} / 0.04 / 0.01 & 0.348 / \textbf{0.369} / 0.003 & \textbf{0.92} / 0.97 / 0.98 & 68.31 / \textbf{9.34} / 31.20 \\
\hline
0.4 & 2.53 / \textbf{3.00} / 1.12 & 0.27 / \textbf{1.00} / 0.01 & 0.203 / \textbf{0.782} / 0.010 & 0.97 / \textbf{0.43} / 0.98 & 40.44 / \textbf{20.33} / 22.91 \\
\hline
0.6 & 2.04 / \textbf{3.00} / 1.30 & 0.23 / \textbf{1.00} / 0.04 & 0.211 / \textbf{0.891} / 0.022 & 0.95 / \textbf{0.23} / 0.97 & 40.79 / \textbf{26.08} / 25.13 \\
\hline
0.8 & 1.97 / \textbf{3.00} / 1.13 & 0.15 / \textbf{1.00} / 0.00 & 0.200 / \textbf{0.948} / 0.037 & 0.96 / \textbf{0.11} / 0.96 & 65.32 / 43.92 / \textbf{37.55} \\
\hline
1.0 & 2.17 / \textbf{3.00} / 1.25 & 0.24 / \textbf{1.00} / 0.01 & 0.306 / \textbf{0.967} / 0.054 & 0.93 / \textbf{0.06} / 0.95 & 37.92 / 27.52 / \textbf{23.73} \\
\hline
\multicolumn{6}{|c|}{\textbf{$\rho = 0.40$}} \\
\hline
0.2 & 9.57 / 0.06 / \textbf{1.25} & 0.01 / 0.00 / \textbf{0.02} & \textbf{0.107} / 0.002 / 0.002 & 1.00 / \textbf{0.94} / 0.98 & 68.99 / \textbf{0.36} / 31.79 \\
\hline
0.4 & 8.40 / \textbf{3.23} / 1.28 & 0.05 / \textbf{0.81} / 0.02 & 0.019 / \textbf{0.764} / 0.009 & 1.00 / \textbf{0.53} / 0.98 & 40.68 / \textbf{15.34} / 23.71 \\
\hline
0.6 & 6.99 / \textbf{3.00} / 1.24 & 0.08 / \textbf{1.00} / 0.02 & 0.036 / \textbf{0.877} / 0.021 & 1.00 / \textbf{0.28} / 0.97 & 40.84 / \textbf{22.12} / 24.99 \\
\hline
0.8 & 7.12 / \textbf{3.00} / 1.26 & 0.14 / \textbf{1.00} / 0.04 & 0.045 / \textbf{0.934} / 0.035 & 1.00 / \textbf{0.15} / 0.96 & 66.66 / 40.63 / \textbf{40.27} \\
\hline
1.0 & 6.36 / \textbf{3.00} / 1.32 & 0.13 / \textbf{1.00} / 0.05 & 0.061 / \textbf{0.963} / 0.052 & 1.00 / \textbf{0.07} / 0.95 & 38.13 / 26.40 / \textbf{25.07} \\
\hline
\multicolumn{6}{|c|}{\textbf{$\rho = 0.80$}} \\
\hline
0.2 & 104.32 / 0.00 / \textbf{1.56} & 0.00 / 0.00 / \textbf{0.04} & $-0.001$ / 0.000 / \textbf{0.001} & 1.00 / \textbf{0.94} / 0.98 & 565.48 / \textbf{0.42} / 49.88 \\
\hline
0.4 & 103.90 / 0.00 / \textbf{1.52} & 0.00 / 0.00 / \textbf{0.12} & $-0.001$ / 0.000 / \textbf{0.005} & 1.00 / \textbf{0.94} / 0.98 & 425.19 / \textbf{0.36} / 48.41 \\
\hline
0.6 & 103.82 / 0.52 / \textbf{1.70} & 0.00 / 0.00 / \textbf{0.12} & $-0.001$ / \textbf{0.026} / 0.013 & 1.00 / \textbf{0.94} / 0.98 & 437.93 / \textbf{0.71} / 52.72 \\
\hline
0.8 & 103.54 / 4.88 / \textbf{1.68} & 0.00 / \textbf{0.18} / 0.14 & $-0.001$ / \textbf{0.363} / 0.023 & 1.00 / \textbf{0.96} / 0.98 & 471.83 / \textbf{4.88} / 55.70 \\
\hline
1.0 & 102.66 / 4.44 / \textbf{1.66} & 0.00 / \textbf{0.18} / 0.12 & $-0.001$ / \textbf{0.901} / 0.035 & 1.00 / \textbf{0.66} / 0.97 & 409.01 / \textbf{13.68} / 56.61 \\
\hline
\end{tabular}%
}
\end{table}
\begin{table}[htbp]
\centering
\scriptsize
\setlength{\tabcolsep}{4.2pt}
\renewcommand{\arraystretch}{1.18}
\caption{Comparison of DCART, SPLADE and TV across grid sizes, jump sizes, and SAR $\rho$ for Config.\ 2 (5 patches). Each cell reports average over 100 replicates in the order DCART / SPLADE / TV.}
\label{tab:grid256_512_rho_stacked_k5}
\resizebox{\textwidth}{!}{%
\begin{tabular}{|c|c|c|c|c|c|}
\hline
Jump & $\hat K$ & $I(\hat K=5)$ & ARI & Hausdorff distance & time/iter (sec) \\
\hline

\multicolumn{6}{|c|}{\textbf{Grid = $256 \times 256$}} \\
\hline
\multicolumn{6}{|c|}{\textbf{$\rho = 0.04$}} \\
\hline
0.2 & 6.30 / \textbf{4.78} / 1.33 & 0.14 / \textbf{0.68} / 0.00 & 0.621 / \textbf{0.752} / 0.059 & 0.84 / \textbf{0.75} / 0.97 & 9.94 / \textbf{4.01} / 4.93 \\
\hline
0.4 & 5.72 / \textbf{5.00} / 1.54 & 0.34 / \textbf{1.00} / 0.00 & 0.797 / \textbf{0.877} / 0.194 & 0.89 / \textbf{0.40} / 0.98 & 9.69 / \textbf{4.99} / 5.38 \\
\hline
0.6 & 5.52 / \textbf{5.00} / 6.00 & 0.34 / \textbf{1.00} / 0.34 & 0.851 / \textbf{0.930} / 0.831 & 0.73 / \textbf{0.21} / 0.99 & 9.76 / 5.36 / \textbf{3.76} \\
\hline
0.8 & 5.86 / \textbf{5.00} / 6.72 & 0.27 / \textbf{1.00} / 0.18 & 0.878 / \textbf{0.951} / 0.852 & 0.73 / \textbf{0.13} / 0.99 & 9.67 / 5.65 / \textbf{4.00} \\
\hline
1.0 & 6.16 / \textbf{5.00} / 7.32 & 0.22 / \textbf{1.00} / 0.05 & 0.905 / \textbf{0.958} / 0.902 & 0.74 / \textbf{0.10} / 0.99 & 9.85 / 5.77 / \textbf{3.97} \\
\hline
\multicolumn{6}{|c|}{\textbf{$\rho = 0.40$}} \\
\hline
0.2 & 9.82 / \textbf{4.24} / 1.49 & 0.04 / \textbf{0.24} / 0.00 & 0.402 / \textbf{0.715} / 0.051 & 1.00 / \textbf{0.90} / 0.98 & 9.62 / \textbf{3.57} / 4.92 \\
\hline
0.4 & 8.91 / \textbf{5.01} / 1.68 & 0.04 / \textbf{0.95} / 0.01 & 0.528 / \textbf{0.865} / 0.172 & 1.00 / \textbf{0.50} / 0.98 & 9.79 / \textbf{4.85} / 5.47 \\
\hline
0.6 & 9.02 / \textbf{5.00} / 2.20 & 0.05 / \textbf{1.00} / 0.02 & 0.651 / \textbf{0.920} / 0.302 & 0.99 / \textbf{0.26} / 0.99 & 9.80 / \textbf{5.28} / 5.72 \\
\hline
0.8 & 9.18 / \textbf{5.00} / 7.01 & 0.02 / \textbf{1.00} / 0.10 & 0.668 / \textbf{0.944} / 0.834 & 1.00 / \textbf{0.16} / 0.99 & 9.71 / 5.57 / \textbf{4.07} \\
\hline
1.0 & 9.51 / \textbf{5.00} / 7.67 & 0.01 / \textbf{1.00} / 0.04 & 0.713 / \textbf{0.954} / 0.855 & 1.00 / \textbf{0.11} / 0.99 & 9.95 / 5.72 / \textbf{4.00} \\
\hline
\multicolumn{6}{|c|}{\textbf{$\rho = 0.80$}} \\
\hline
0.2 & 29.07 / \textbf{2.85} / 2.13 & 0.00 / \textbf{0.03} / 0.00 & $-0.002$ / \textbf{0.456} / 0.020 & 1.00 / \textbf{0.96} / 0.99 & 37.35 / \textbf{2.23} / 5.34 \\
\hline
0.4 & 28.72 / \textbf{4.13} / 2.29 & 0.00 / \textbf{0.18} / 0.01 & $-0.002$ / \textbf{0.738} / 0.078 & 1.00 / \textbf{0.92} / 0.99 & 37.34 / \textbf{3.86} / 5.83 \\
\hline
0.6 & 28.79 / \textbf{4.43} / 2.89 & 0.00 / \textbf{0.39} / 0.05 & $-0.001$ / \textbf{0.818} / 0.147 & 1.00 / \textbf{0.80} / 1.00 & 37.58 / \textbf{4.57} / 6.10 \\
\hline
0.8 & 28.83 / \textbf{4.74} / 3.82 & 0.00 / \textbf{0.68} / 0.19 & $-0.001$ / \textbf{0.865} / 0.203 & 1.00 / \textbf{0.61} / 1.00 & 38.00 / \textbf{4.96} / 6.40 \\
\hline
1.0 & 29.41 / \textbf{4.97} / 5.71 & 0.00 / \textbf{0.89} / 0.20 & 0.043 / \textbf{0.908} / 0.239 & 1.00 / \textbf{0.38} / 1.00 & 38.60 / \textbf{5.28} / 6.37 \\
\hline

\multicolumn{6}{|c|}{\textbf{Grid = $512 \times 512$}} \\
\hline
\multicolumn{6}{|c|}{\textbf{$\rho = 0.04$}} \\
\hline
0.2 & 5.86 / \textbf{4.89} / 1.22 & 0.30 / \textbf{0.67} / 0.00 & 0.729 / \textbf{0.822} / 0.023 & 0.82 / \textbf{0.78} / 0.97 & 36.43 / \textbf{18.76} / 20.76 \\
\hline
0.4 & 5.91 / \textbf{5.00} / 1.27 & 0.27 / \textbf{1.00} / 0.00 & 0.849 / \textbf{0.944} / 0.059 & 0.79 / \textbf{0.27} / 0.97 & 37.15 / 23.49 / \textbf{22.81} \\
\hline
0.6 & 6.03 / \textbf{5.00} / 1.39 & 0.25 / \textbf{1.00} / 0.00 & 0.895 / \textbf{0.969} / 0.080 & 0.78 / \textbf{0.14} / 0.97 & 36.94 / \textbf{24.09} / 24.46 \\
\hline
0.8 & 5.86 / \textbf{5.00} / 1.40 & 0.30 / \textbf{1.00} / 0.00 & 0.929 / \textbf{0.981} / 0.103 & 0.59 / \textbf{0.07} / 0.97 & 37.57 / \textbf{24.61} / 26.88 \\
\hline
1.0 & 5.87 / \textbf{5.00} / 1.43 & 0.19 / \textbf{1.00} / 0.00 & 0.940 / \textbf{0.986} / 0.162 & 0.47 / \textbf{0.04} / 0.97 & 39.75 / \textbf{28.11} / 30.21 \\
\hline
\multicolumn{6}{|c|}{\textbf{$\rho = 0.40$}} \\
\hline
0.2 & 11.58 / \textbf{4.10} / 1.24 & 0.05 / \textbf{0.10} / 0.00 & 0.428 / \textbf{0.776} / 0.022 & 1.00 / \textbf{0.95} / 0.97 & 36.81 / \textbf{16.70} / 21.00 \\
\hline
0.4 & 10.70 / \textbf{5.07} / 1.36 & 0.03 / \textbf{0.94} / 0.00 & 0.527 / \textbf{0.936} / 0.061 & 1.00 / \textbf{0.32} / 0.97 & 37.37 / \textbf{22.87} / 23.54 \\
\hline
0.6 & 9.67 / \textbf{5.00} / 1.46 & 0.03 / \textbf{1.00} / 0.00 & 0.629 / \textbf{0.966} / 0.087 & 0.98 / \textbf{0.16} / 0.97 & 37.15 / \textbf{24.02} / 24.92 \\
\hline
0.8 & 9.62 / \textbf{5.00} / 1.57 & 0.02 / \textbf{1.00} / 0.00 & 0.691 / \textbf{0.978} / 0.113 & 0.97 / \textbf{0.09} / 0.98 & 37.74 / \textbf{24.64} / 26.92 \\
\hline
1.0 & 9.53 / \textbf{5.00} / 1.53 & 0.00 / \textbf{1.00} / 0.00 & 0.731 / \textbf{0.985} / 0.168 & 1.00 / \textbf{0.04} / 0.98 & 39.96 / \textbf{27.99} / 31.05 \\
\hline
\multicolumn{6}{|c|}{\textbf{$\rho = 0.80$}} \\
\hline
0.2 & 103.40 / \textbf{2.25} / 1.72 & 0.00 / \textbf{0.02} / 0.00 & $-0.001$ / \textbf{0.360} / 0.014 & 1.00 / \textbf{0.96} / 0.98 & 404.87 / \textbf{4.33} / 42.81 \\
\hline
0.4 & 102.85 / \textbf{3.71} / 1.79 & 0.00 / \textbf{0.04} / 0.00 & $-0.001$ / \textbf{0.732} / 0.051 & 1.00 / \textbf{0.95} / 0.98 & 404.71 / \textbf{14.33} / 48.50 \\
\hline
0.6 & 103.18 / \textbf{4.04} / 1.92 & 0.00 / \textbf{0.04} / 0.00 & $-0.001$ / \textbf{0.845} / 0.095 & 1.00 / \textbf{0.94} / 0.99 & 419.63 / \textbf{20.06} / 53.21 \\
\hline
0.8 & 106.10 / \textbf{4.70} / 2.00 & 0.00 / \textbf{0.58} / 0.00 & $-0.001$ / \textbf{0.876} / 0.139 & 1.00 / \textbf{0.77} / 0.99 & 383.13 / \textbf{20.81} / 55.35 \\
\hline
1.0 & 112.32 / \textbf{5.20} / 2.16 & 0.00 / \textbf{0.82} / 0.00 & 0.019 / \textbf{0.963} / 0.187 & 1.00 / \textbf{0.28} / 0.99 & 317.32 / \textbf{25.55} / 51.02 \\
\hline
\end{tabular}%
}
\end{table}

Tables \ref{tab:grid256_512_rho_stacked_k3} and \ref{tab:grid256_512_rho_stacked_k5} summarize the comparative performance of DCART, SPLADE, and TV across a diverse range of experimental configurations. Notably, DPLS-SAD \citep{wang2025optimal} is excluded from these comprehensive evaluations due to its prohibitive computational runtime; in many of our experimental configurations, the algorithm failed to terminate within a practical time frame.

Regarding computational efficiency, SPLADE is consistently the fastest method across nearly all settings; in the rare instances where TV marginally outperforms it, TV incurs a severe cost in accuracy. While average iteration times naturally scale with grid size and spatial dependence, SPLADE maintains a distinct advantage in these more challenging scenarios, achieving speed-ups of up to 15--20x over DCART and 4x over TV. Further, while the 5-patch layout (Configuration 2) increases the computational burden for all methods, it concurrently highlights the most substantial relative speed improvements for SPLADE. It is important to note that DCART was implemented with an \texttt{Rcpp} accelerator, which makes the contrast in speed even more stark. 

In terms of statistical performance, SPLADE generally yields the most accurate estimation of the true number of patches, struggling only when the baseline jump size is exceptionally small. Strong spatial dependence ($\rho=0.8$) degrades the performance of all methods; however, as the jump size (and correspondingly, the signal-to-noise ratio) increases, SPLADE recovers its count accuracy much more rapidly than competing methods. A similar trend is evident in the Adjusted Rand Index (ARI). SPLADE comprehensively outperforms the competing methods, and it is the only approach whose ARI reliably approaches 1 under high dependence as the signal strength grows. Finally, evaluating localization accuracy via the normalized Hausdorff distance reveals that SPLADE's estimation error decreases significantly as the jump size increases, whereas the error rates of DCART and TV stagnate or exhibit only marginal improvements. 
%








\begin{table}[htbp]
\centering
\scriptsize
\setlength{\tabcolsep}{4.2pt}
\renewcommand{\arraystretch}{1.18}
\caption{Comparison of DPLS-SAD and SPLADE across grid sizes, jump sizes, and SAR $\rho$ for Config.\ 1 (3 patches). Each cell reports average over 100 replicates in the order DPLS-SAD / SPLADE.}
\label{tab:dpls_SPLADE_k3_stacked}
\resizebox{\textwidth}{!}{%
\begin{tabular}{|c|c|c|c|c|c|}
\hline
$\text{Jump}$ & $\hat K$ & $I(\hat K=3)$ & ARI & Hausdorff distance & time/iter (sec) \\
\hline

\multicolumn{6}{|c|}{\textbf{Grid = $64 \times 64$}} \\
\hline
\multicolumn{6}{|c|}{\textbf{$\rho = 0.2$}} \\
\hline
0.50 & 0.56 / \textbf{2.29} & 0.00 / \textbf{0.44} & 0.003 / \textbf{0.372} & 0.96 / \textbf{0.83} & 41.40 / \textbf{1.15} \\
\hline
0.75 & 0.69 / \textbf{2.77} & 0.00 / \textbf{0.75} & 0.006 / \textbf{0.534} & 0.96 / \textbf{0.64} & 38.10 / \textbf{0.72} \\
\hline
1.00 & 0.76 / \textbf{2.40} & 0.00 / \textbf{0.50} & 0.010 / \textbf{0.511} & 0.96 / \textbf{0.71} & 43.51 / \textbf{2.42} \\
\hline
\multicolumn{6}{|c|}{\textbf{$\rho = 0.4$}} \\
\hline
0.50 & 0.50 / \textbf{1.51} & 0.00 / \textbf{0.09} & 0.003 / \textbf{0.239} & 0.95 / \textbf{0.90} & 40.77 / \textbf{0.64} \\
\hline
0.75 & 0.63 / \textbf{2.69} & 0.00 / \textbf{0.67} & 0.005 / \textbf{0.474} & 0.96 / \textbf{0.73} & 37.74 / \textbf{0.59} \\
\hline
1.00 & 0.71 / \textbf{2.76} & 0.00 / \textbf{0.76} & 0.009 / \textbf{0.568} & 0.96 / \textbf{0.60} & 38.25 / \textbf{0.80} \\
\hline
\multicolumn{6}{|c|}{\textbf{$\rho = 0.6$}} \\
\hline
0.50 & 0.36 / \textbf{0.67} & 0.00 / \textbf{0.01} & 0.002 / \textbf{0.088} & 0.95 / \textbf{0.92} & 38.63 / \textbf{0.30} \\
\hline
0.75 & 0.50 / \textbf{1.80} & 0.00 / \textbf{0.15} & 0.004 / \textbf{0.306} & 0.95 / \textbf{0.88} & 40.15 / \textbf{0.86} \\
\hline
1.00 & 0.69 / \textbf{2.58} & 0.00 / \textbf{0.57} & 0.008 / \textbf{0.484} & 0.96 / \textbf{0.77} & 13.35 / \textbf{0.24} \\
\hline

\multicolumn{6}{|c|}{\textbf{Grid = $128 \times 128$}} \\
\hline
\multicolumn{6}{|c|}{\textbf{$\rho = 0.2$}} \\
\hline
0.50 & 0.00 / \textbf{3.02} & 0.00 / \textbf{0.96} & 0.000 / \textbf{0.597} & 0.93 / \textbf{0.64} & 121.46 / \textbf{3.87} \\
\hline
0.75 & 0.00 / \textbf{3.00} & 0.00 / \textbf{0.98} & 0.000 / \textbf{0.749} & 0.93 / \textbf{0.42} & 168.22 / \textbf{2.89} \\
\hline
1.00 & 0.00 / \textbf{2.98} & 0.00 / \textbf{0.96} & 0.000 / \textbf{0.820} & 0.93 / \textbf{0.31} & 110.97 / \textbf{2.49} \\
\hline
\multicolumn{6}{|c|}{\textbf{$\rho = 0.4$}} \\
\hline
0.50 & 0.00 / \textbf{3.02} & 0.00 / \textbf{0.70} & 0.000 / \textbf{0.526} & 0.93 / \textbf{0.75} & 109.54 / \textbf{1.35} \\
\hline
0.75 & 0.00 / \textbf{3.00} & 0.00 / \textbf{0.98} & 0.000 / \textbf{0.736} & 0.93 / \textbf{0.45} & 167.38 / \textbf{2.71} \\
\hline
1.00 & 0.00 / \textbf{3.00} & 0.00 / \textbf{0.98} & 0.000 / \textbf{0.808} & 0.93 / \textbf{0.32} & 167.80 / \textbf{3.14} \\
\hline
\multicolumn{6}{|c|}{\textbf{$\rho = 0.6$}} \\
\hline
0.50 & 0.00 / \textbf{1.79} & 0.00 / \textbf{0.23} & 0.000 / \textbf{0.240} & 0.93 / \textbf{0.93} & 126.89 / \textbf{1.30} \\
\hline
0.75 & 0.00 / \textbf{3.12} & 0.00 / \textbf{0.80} & 0.000 / \textbf{0.649} & 0.93 / \textbf{0.64} & 69.53 / \textbf{0.96} \\
\hline
1.00 & 0.00 / \textbf{3.01} & 0.00 / \textbf{0.97} & 0.000 / \textbf{0.782} & 0.93 / \textbf{0.38} & 160.13 / \textbf{2.71} \\
\hline
\end{tabular}%
}
\end{table}

Table \ref{tab:dpls_SPLADE_k3_stacked} compares DPLS-SAD and SPLADE on a reduced $64 \times 64$ grid, a necessary constraint to bypass the severe computational bottlenecks DPLS-SAD encounters on larger domains. Because the original code is unavailable, we evaluated DPLS-SAD using a custom implementation based directly on the authors' pseudocode. For this targeted experiment, we test dependence levels $\rho \in \{0.2, 0.4, 0.6\}$ and jump sizes $\delta_{\mu} \in \{0.5, 0.75, 1\}$. 

The trends across all accuracy metrics align closely with our broader findings against DCART and TV. Computationally, SPLADE achieves massive speedups of 50 to 100$\times$ over our pure \texttt{R} implementation of DPLS-SAD (noting that, unlike the provided DCART package, our DPLS-SAD implementation lacks \texttt{Rcpp} acceleration). We further contextualize this stark difference in time complexity in the real data analysis (Section \ref{sec:data}).

 \ignore{
 \textcolor{red}{Do we need this paragraph? Or we can delegate it to the Conclusions section? \\ 
 Finally, we also attempted comparing with the methods presented in \cite{kirch2025scan} as they dealing with spatially $m-dependent$ scan statistics. However, since their focus is on testing existence of fissure or bubbles, their methods yield poor results as they presently are and need significant modifications for axis-aligned rectangles to actuate a fair comparison with our method and thus we decide to not present the comparative performance here. }
}

\section{Real-world data application: video surveillance footage}\label{sec:data}

\noindent In this section, we demonstrate the practical utility of our proposed method through the analysis of video surveillance footage. Specifically, we apply SPLADE to capture the spatial dynamics of two individuals meeting, evaluating the method's resolution and localization accuracy as the subjects physically approach one another. An additional real-world application concerning anomaly detection in fiber systems is deferred to Appendix \ref{se:appendix real data}.

The CAVIAR project\footnote{EC Funded CAVIAR project/IST 2001 37540, available at: http://homepages.inf.ed.ac.uk/rbf/CAVIAR/.} serves as a foundational benchmark in the field of public surveillance. It provides staged indoor video sequences featuring realistic scenarios, frame-level annotations, and semantically labeled behaviors such as walking alone, meeting, and window shopping. Due to its high-quality ground truth, the dataset has been widely adopted for the reproducible evaluation of detection, tracking and high-level activity analysis \citep{Fisher2004PETS04,CAVIARLabelSchema2005}. Further, it has significantly influenced the development of context-aware perception for ambient intelligence and human-centered video understanding \citep{CrowleyReignier2003VSCA}. 

Subsequent research has utilized CAVIAR for diverse tasks, including short-term activity recognition \citep{RibeiroSantosVictor2005HAREM}, anomalous trajectory detection \citep{SillitoFisher2008BMVC}, and multi-target tracking with social grouping cues \citep{QinShelton2012CVPR} and change point analysis \citep{bai2020multiple}. It remains a staple benchmark in contemporary surveys of surveillance-oriented activity recognition \citep{Chaquet2013Survey}. In this study, we focus specifically on the clip \emph{"Two people meet and walk together."} This scenario provides a natural testbed for our methodology, as the interaction between individuals can be effectively modeled and localized as axis-aligned anomalous spatial patches within the video frames. The selected CAVIAR sequences were recorded at the entrance lobby of the INRIA Labs in Grenoble, France, using a wide-angle camera. The footage was captured at half-resolution PAL quality ($384 \times 288$ pixels) at 25 frames per second and compressed via MPEG2. Our analysis focuses on 501 frames 1000-1550, which we process in \textsf{R} using the \texttt{readJPEG} function from the \texttt{jpeg} package. 

This yields a $288 \times 384 \times 3$ array of normalized pixel intensities in $[0, 1]$, from which the red, green, and blue channels are extracted as distinct matrices. For clarity of presentation in this paper, all images are shown in their transposed orientation. To isolate motion-driven foreground variations from the static background, we center each channel by subtracting a baseline mean image, calculated by averaging frames 1000--1150. While this centering approach is similar in spirit to the preprocessing in \citet{patra2020semi}, we maintain the spatial matrix structure of the data rather than vectorizing the RGB channels and assuming inependence across co-ordinates, thereby preserving the inherent spatial dependence across coordinates. 
\begin{wrapfigure} 
{r}{0.4\linewidth}
    \centering
\includegraphics[width=1\linewidth]{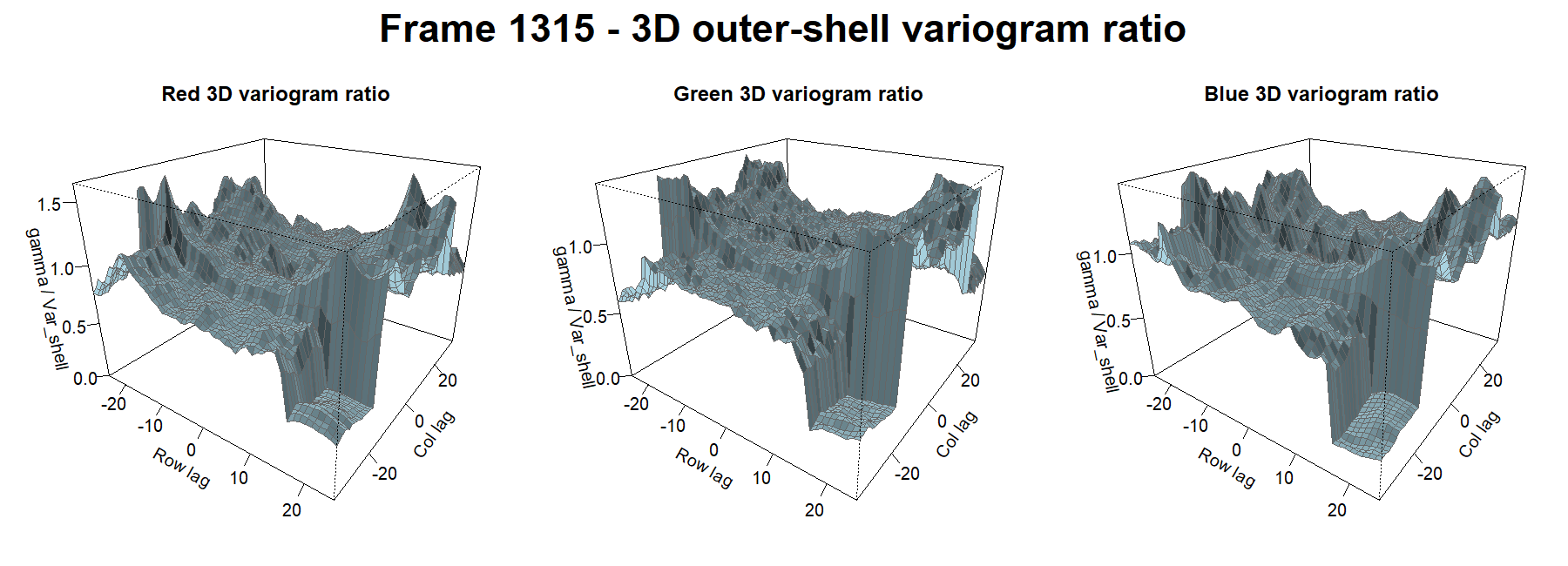}
    \caption{Variogram in both directions compared to $\gamma(\textbf{0})$ for Frame 1315.}
    \label{fig:variogram}
\end{wrapfigure}
In our support, we found significant spatial correlation as we exhibit a variogram for frame 1315 in Figure \ref{fig:variogram} which also establishes a need for developing a method that can handle spatial dependence. Our findings about detected patch boundaries are summarized in Figure \ref{fig:rgb_detection_frames}.

\begin{figure}[!htbp]
    \centering
    \captionsetup[subfigure]{font=small}

    \begin{subfigure}[t]{0.45\textwidth}
        \centering
        \includegraphics[height=0.155\textheight,keepaspectratio]{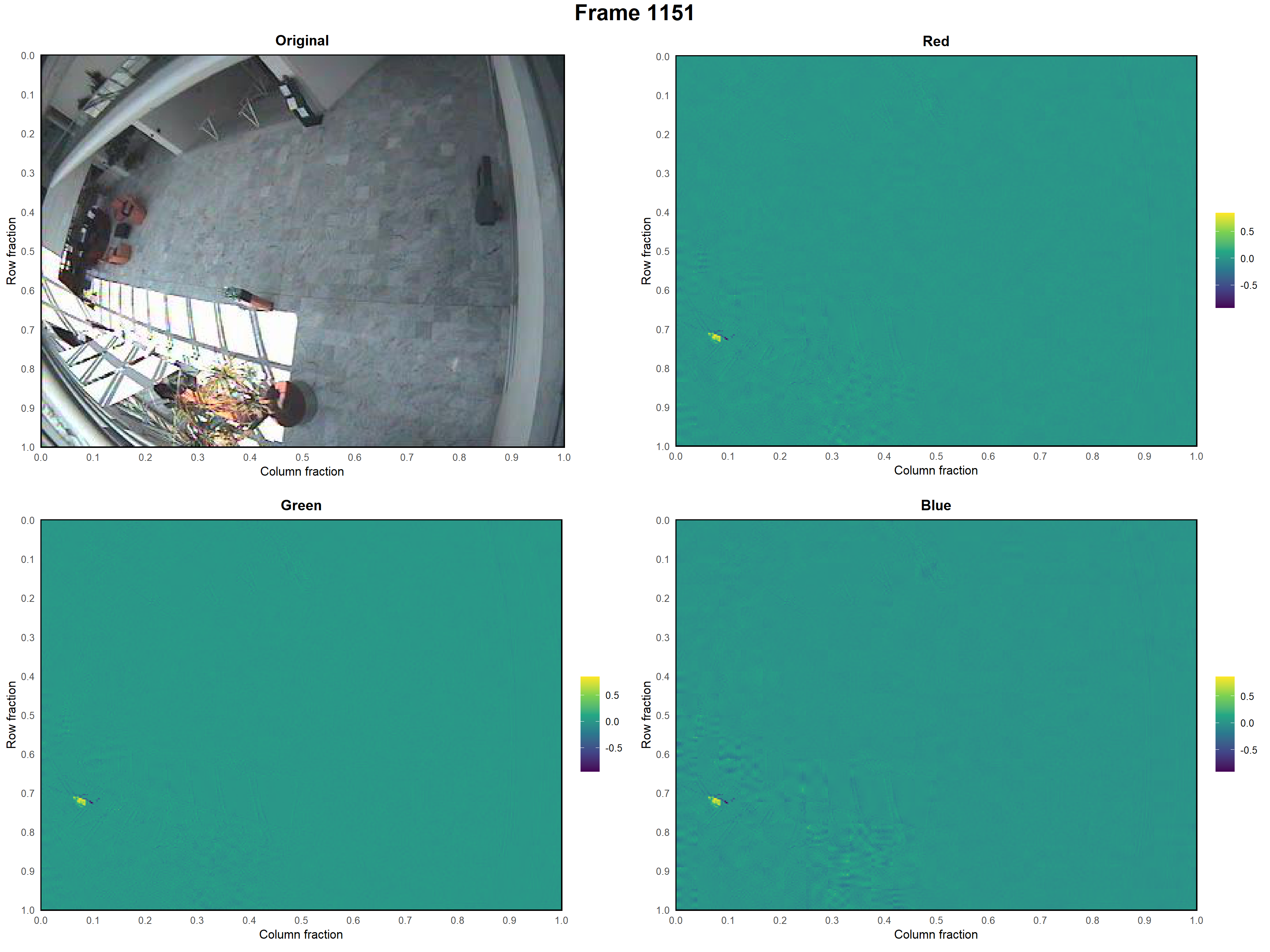}
        \caption{1151: No one in the image}
    \end{subfigure}
    \hfill
    \begin{subfigure}[t]{0.45\textwidth}
        \centering
        \includegraphics[height=0.155\textheight,keepaspectratio]{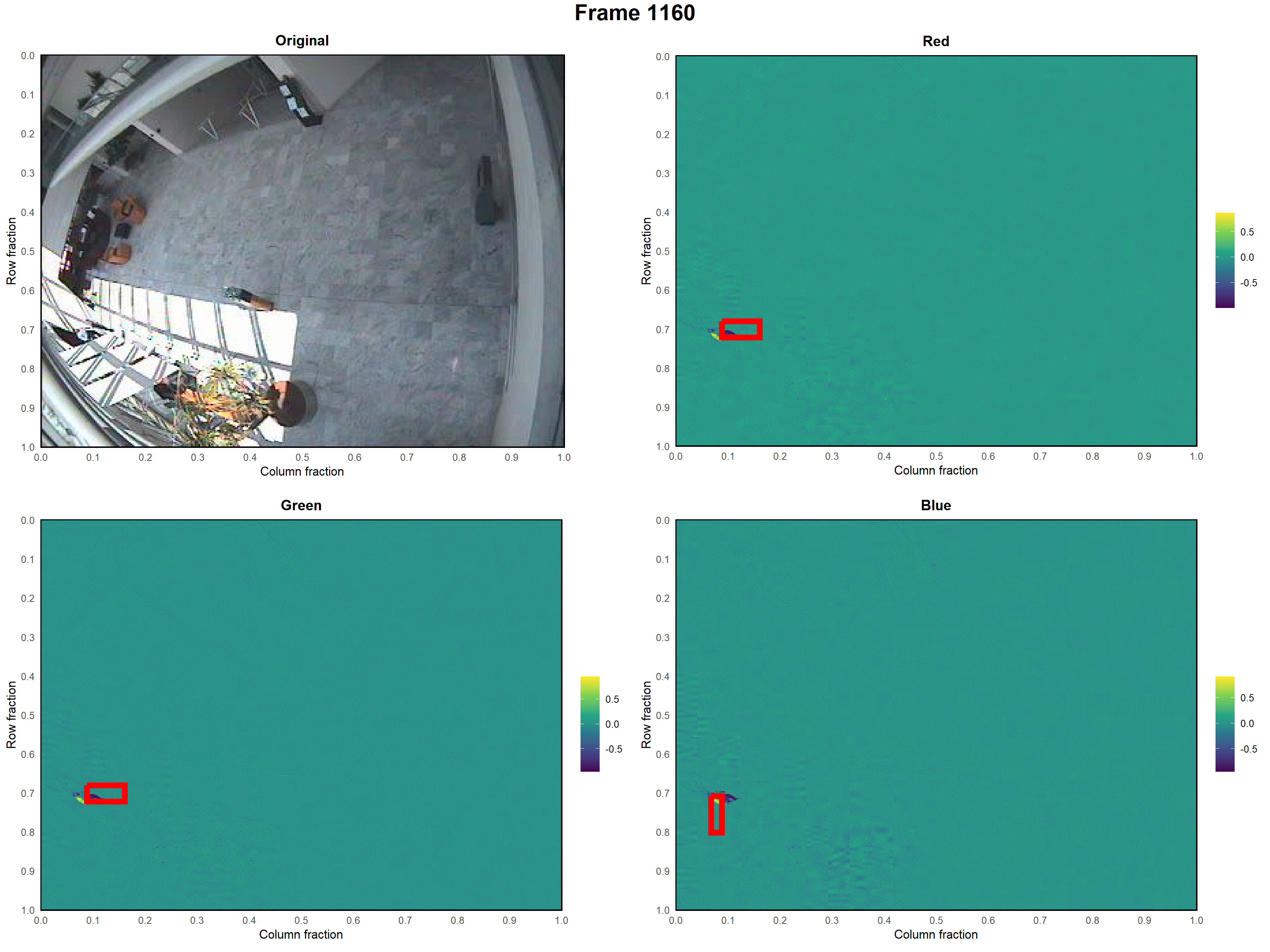}
        \caption{1160: One person in the image}
    \end{subfigure}

    \begin{subfigure}[t]{0.45\textwidth}
        \centering
        \includegraphics[height=0.155\textheight,keepaspectratio]{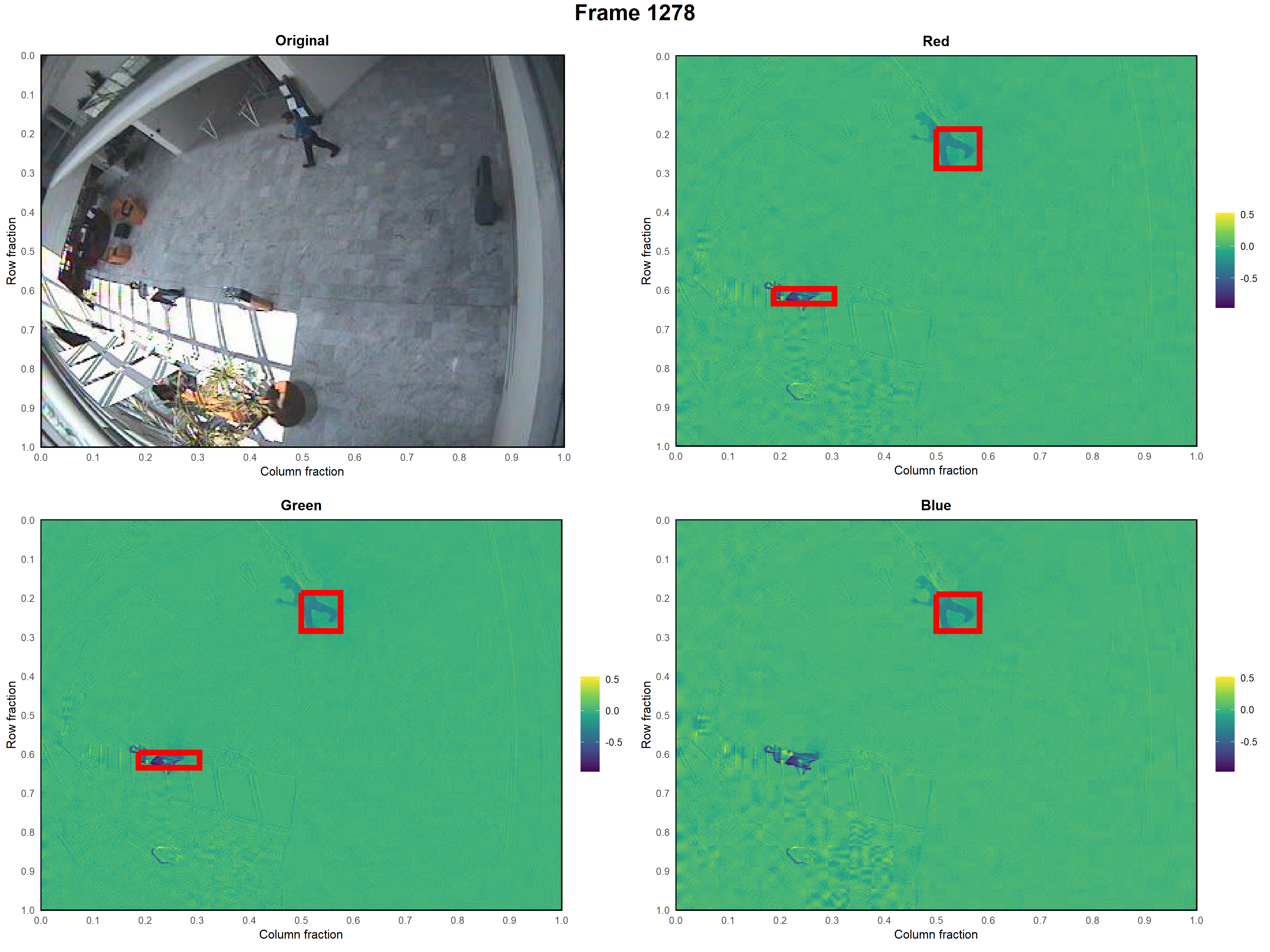}
        \caption{1278: Second person in the image}
    \end{subfigure}
    \hfill
    \begin{subfigure}[t]{0.45\textwidth}
        \centering
        \includegraphics[height=0.155\textheight,keepaspectratio]{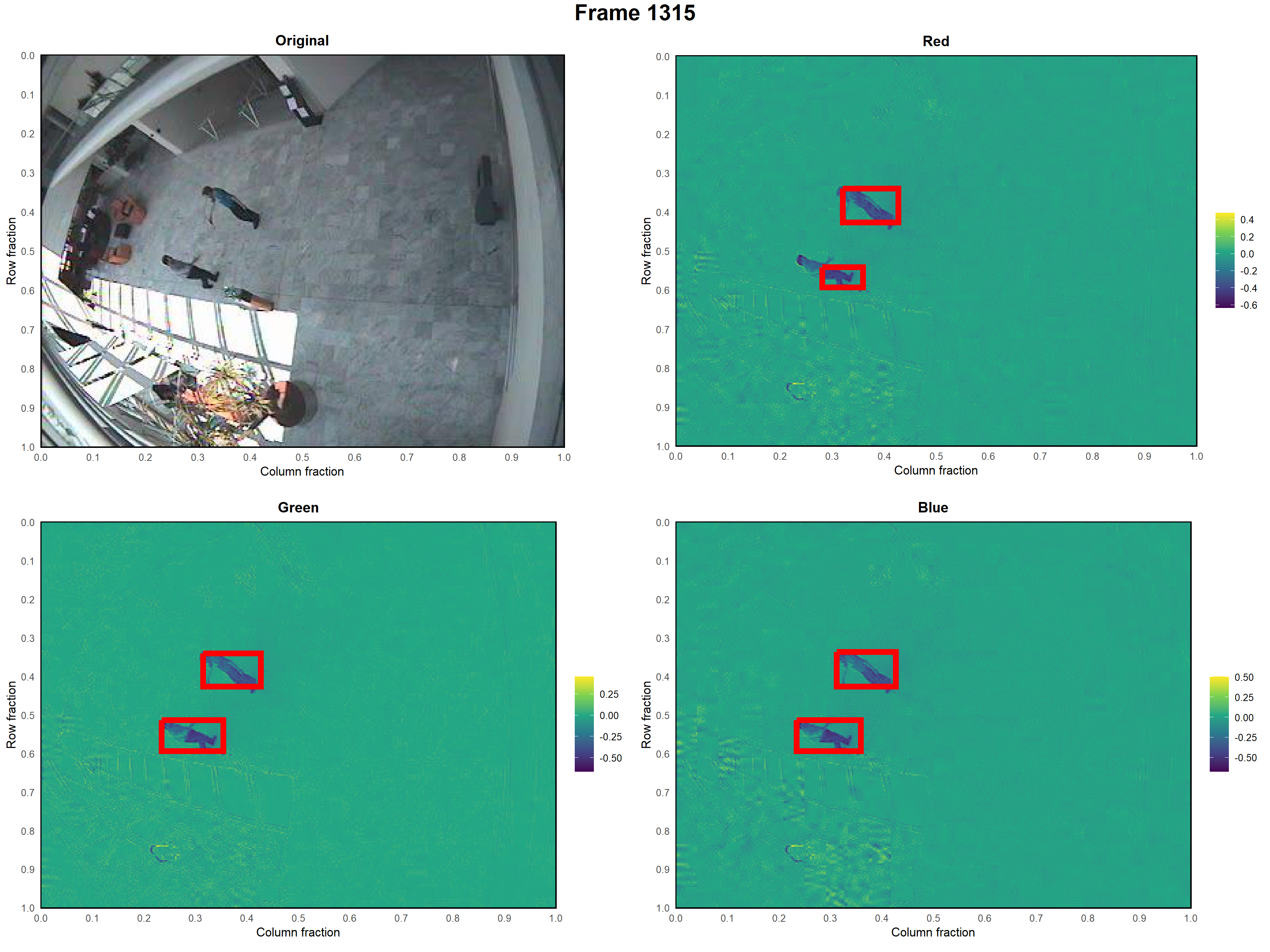}
        \caption{1315: Two persons quite close}
    \end{subfigure}

    \begin{subfigure}[t]{0.45\textwidth}
        \centering
        \includegraphics[height=0.155\textheight,keepaspectratio]{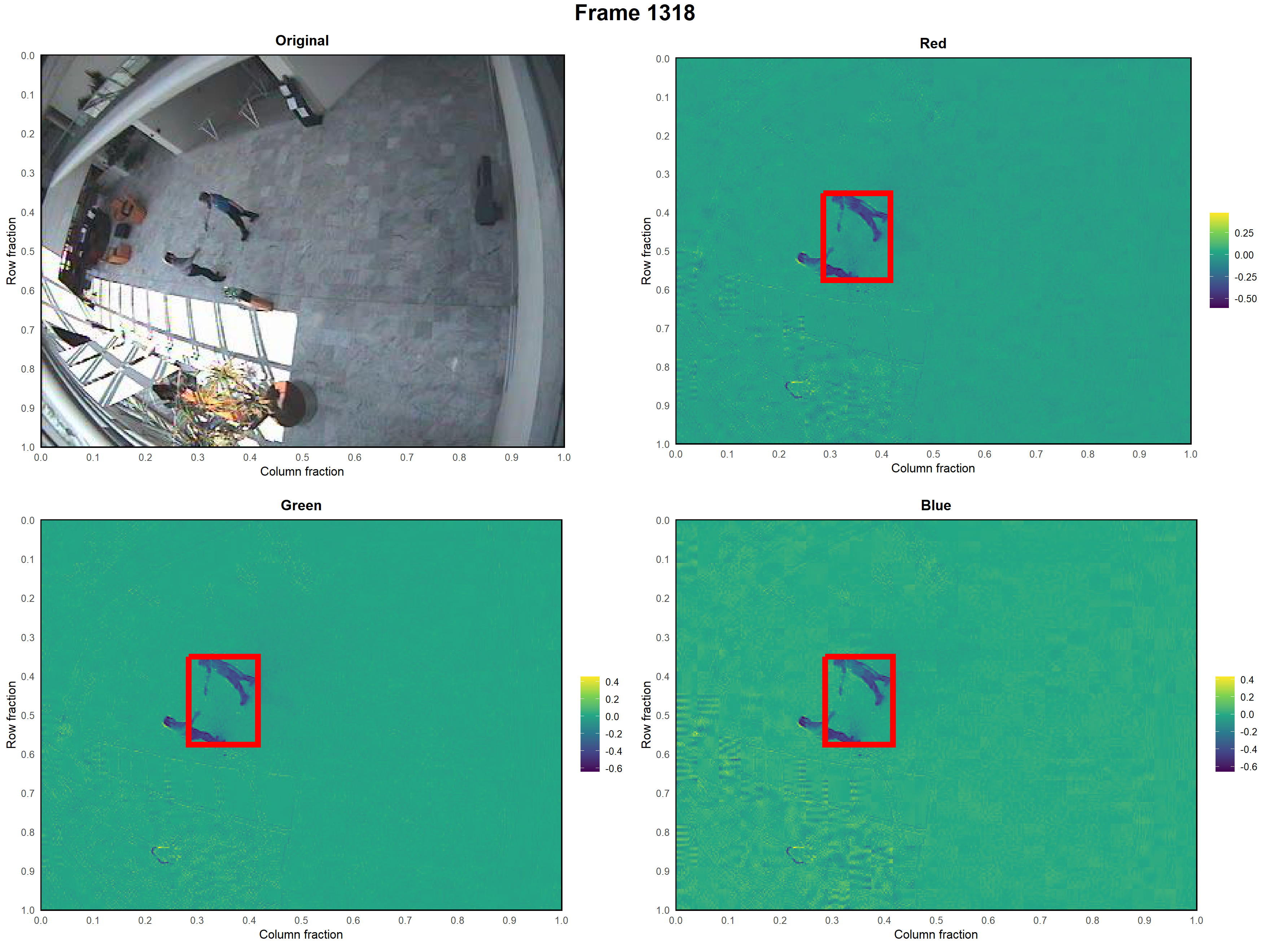}
        \caption{1318: So close that one box is detected}
    \end{subfigure}
    \hfill
    \begin{subfigure}[t]{0.45\textwidth}
        \centering
        \includegraphics[height=0.155\textheight,keepaspectratio]{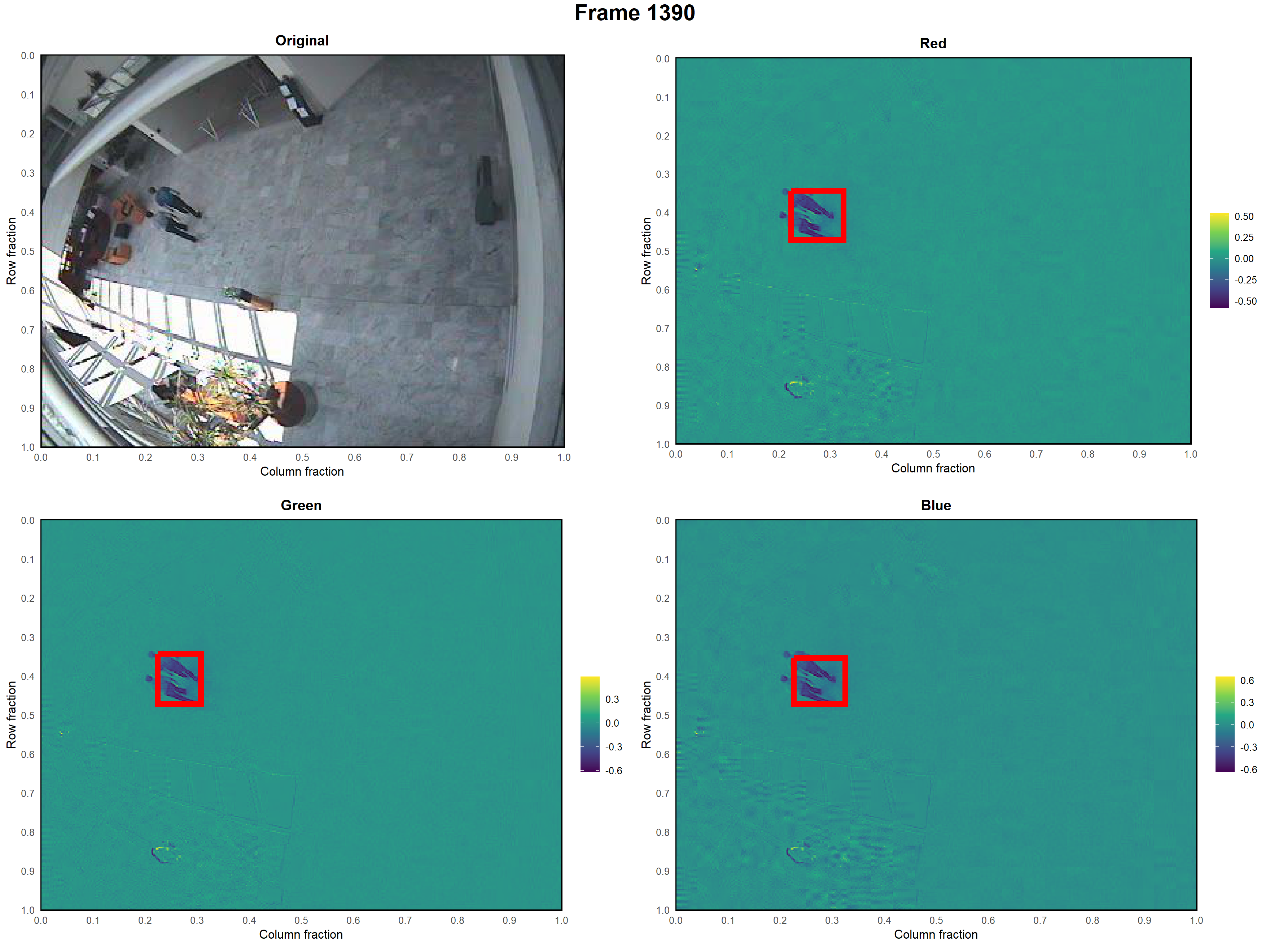}
        \caption{1390: Two persons but far apart}
    \end{subfigure}

    \begin{subfigure}[t]{0.45\textwidth}
        \centering
        \includegraphics[height=0.155\textheight,keepaspectratio]{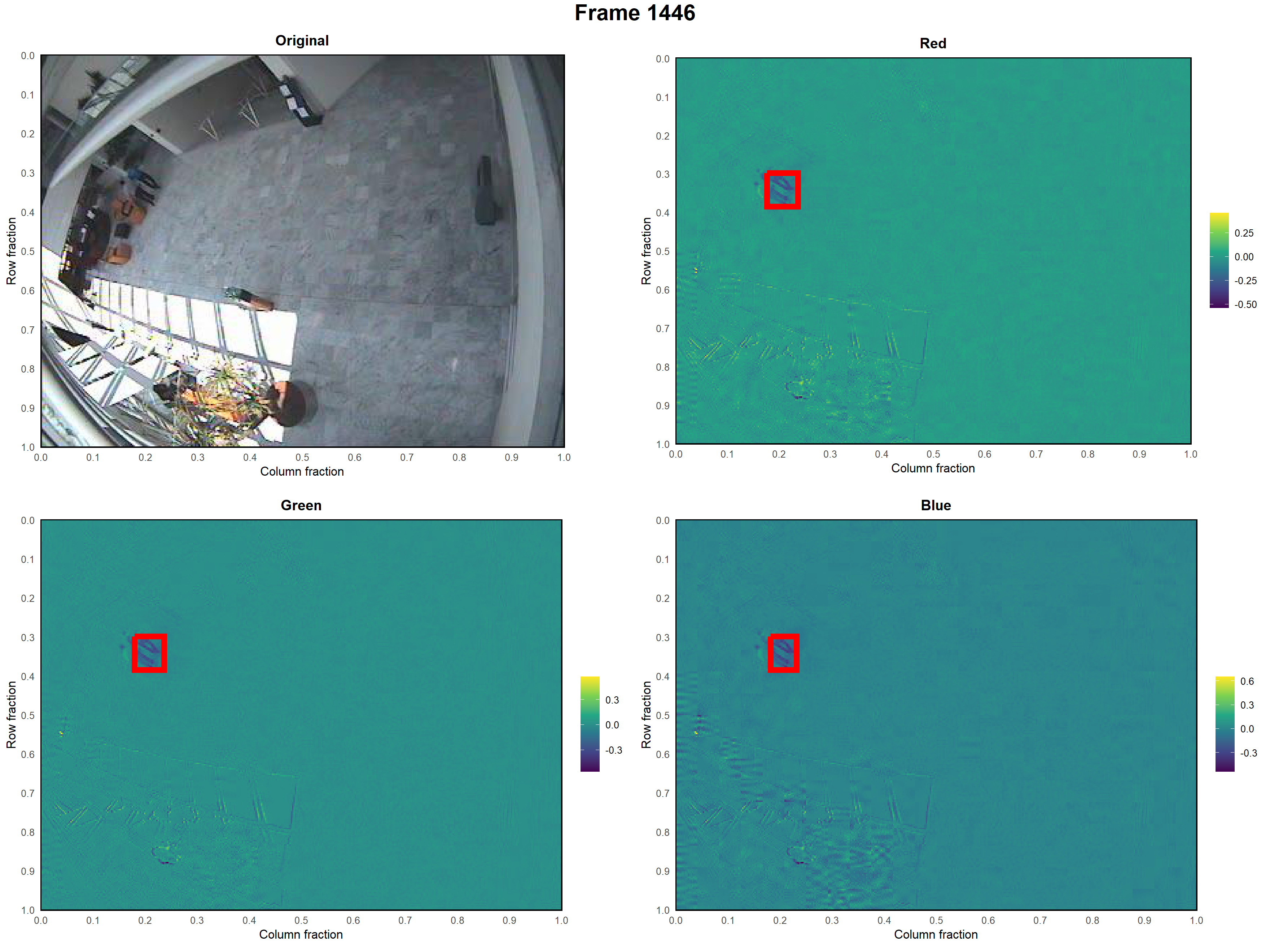}
        \caption{1446: Two persons but far apart}
    \end{subfigure}
    \hfill
    \begin{subfigure}[t]{0.45\textwidth}
        \centering
        \includegraphics[height=0.155\textheight,keepaspectratio]{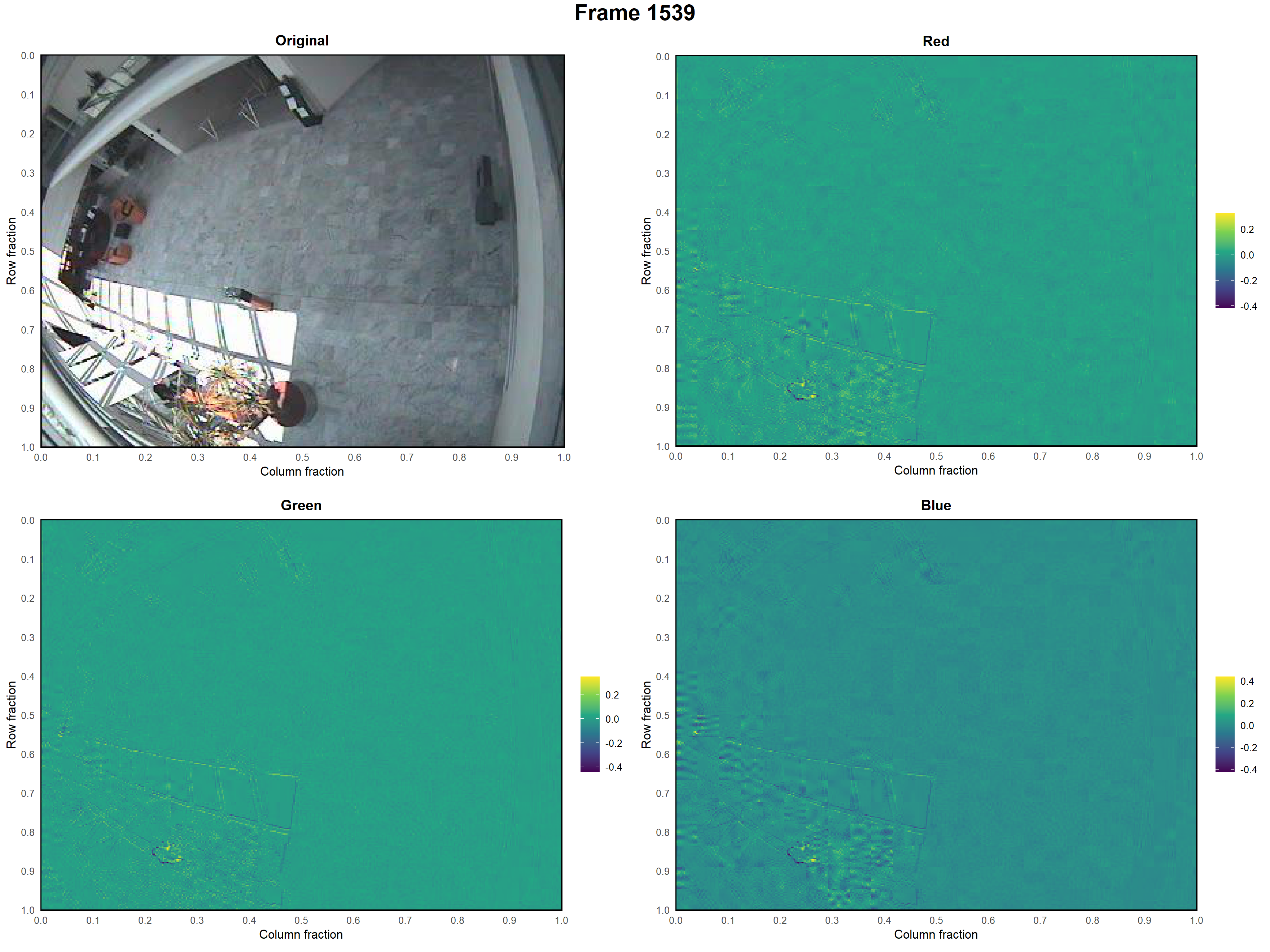}
        \caption{1539: No dynamics}
    \end{subfigure}

    \caption{RGB detection results for selected frames.}
    \label{fig:rgb_detection_frames}
\end{figure}
In Frame 1151, no anomalies are detected as a subject is just beginning to enter the scene. From Frame 1155 onward, a single block is consistently identified despite the subject moving through challenging lighting conditions (sunlight). By Frame 1250, a second individual enters the field of view; our method successfully captures both individuals as distinct entities by Frame 1278. SPLADE maintains this two-block detection with high precision even as the subjects approach one another, successfully resolving them as separate patches until Frame 1316. From Frame 1318, as they meet and walk together, the algorithm transitions to detecting a single merged block. This detection persists until Frame 1446, after which the subjects recede from the camera and the frames return to the baseline static background.

In contrast, DCART and DPLS-SAD fail to achieve this level of precision. To accommodate their inherent limitations, we provided both baselines with significant advantages: DCART was restricted to a $256 \times 256$ bottom-left subgrid (their algorithm is restricted to $2^k \times 2^k$ lattices), and DPLS-SAD was applied to a hand-cropped $111 \times 121$ subgrid specifically centered on the subjects at Frame 1315 to mitigate its lack of scalability. Despite these favorable settings, Figure \ref{fig:othermethoddata} illustrates their poor performance. Across several grid choices for their tuning parameter $\lambda$, DCART identifies numerous spurious patches that fail to intersect with the subjects, while DPLS-SAD fails to detect any anomalies entirely. These failures highlight the inability of the baseline methods to account for the significant spatial correlation present in real-world surveillance data.
\begin{figure}[htbp]
    \centering
    \begin{subfigure}[t]{0.38\linewidth}
        \centering
        \includegraphics[width=\linewidth]{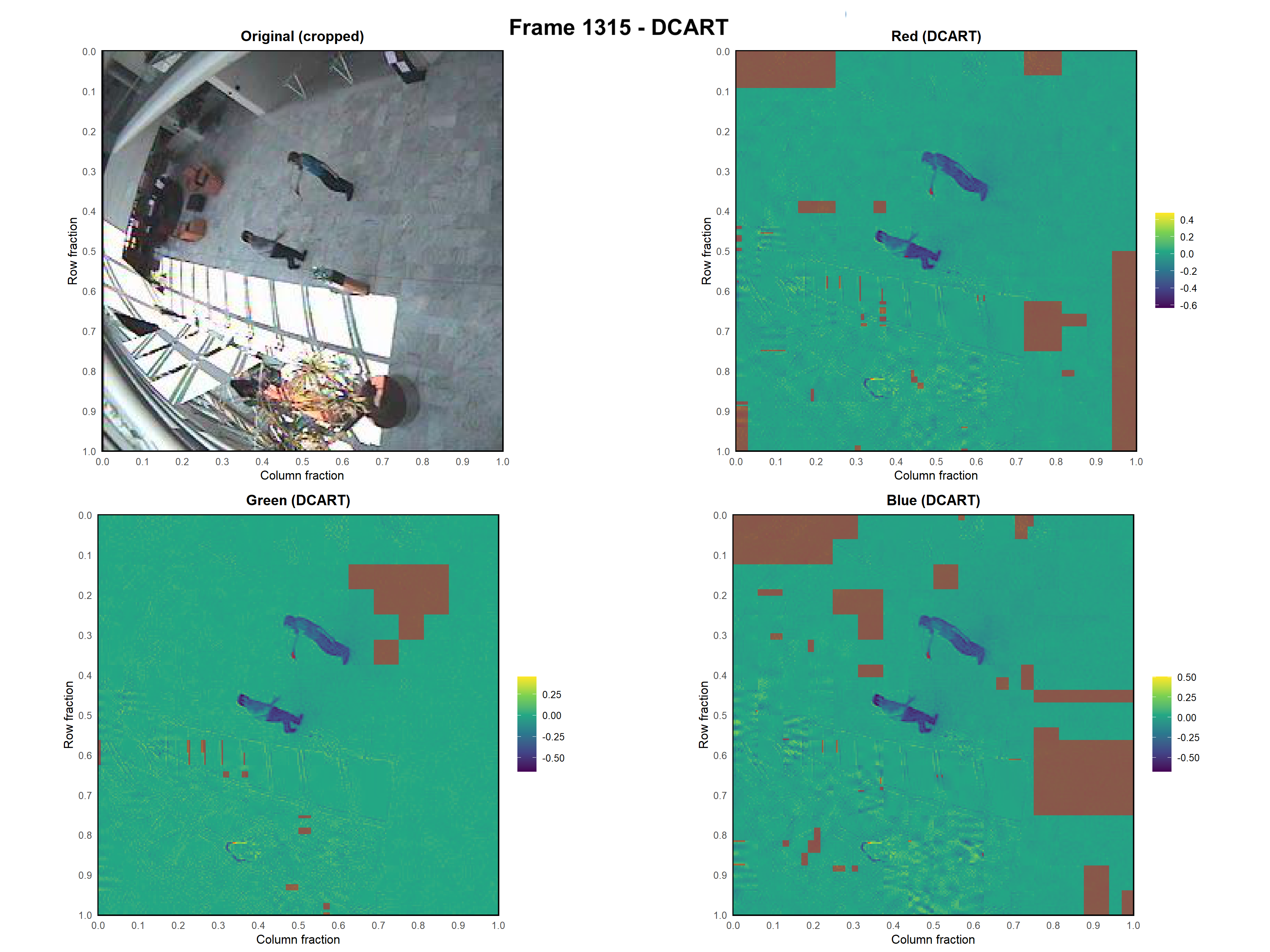}
        \caption{DCART on cropped data}
        \label{fig:DCART on cropped data}
    \end{subfigure}
    \hspace*{2cm}
    \begin{subfigure}[t]{0.38\linewidth}
        \centering
        \includegraphics[width=\linewidth]{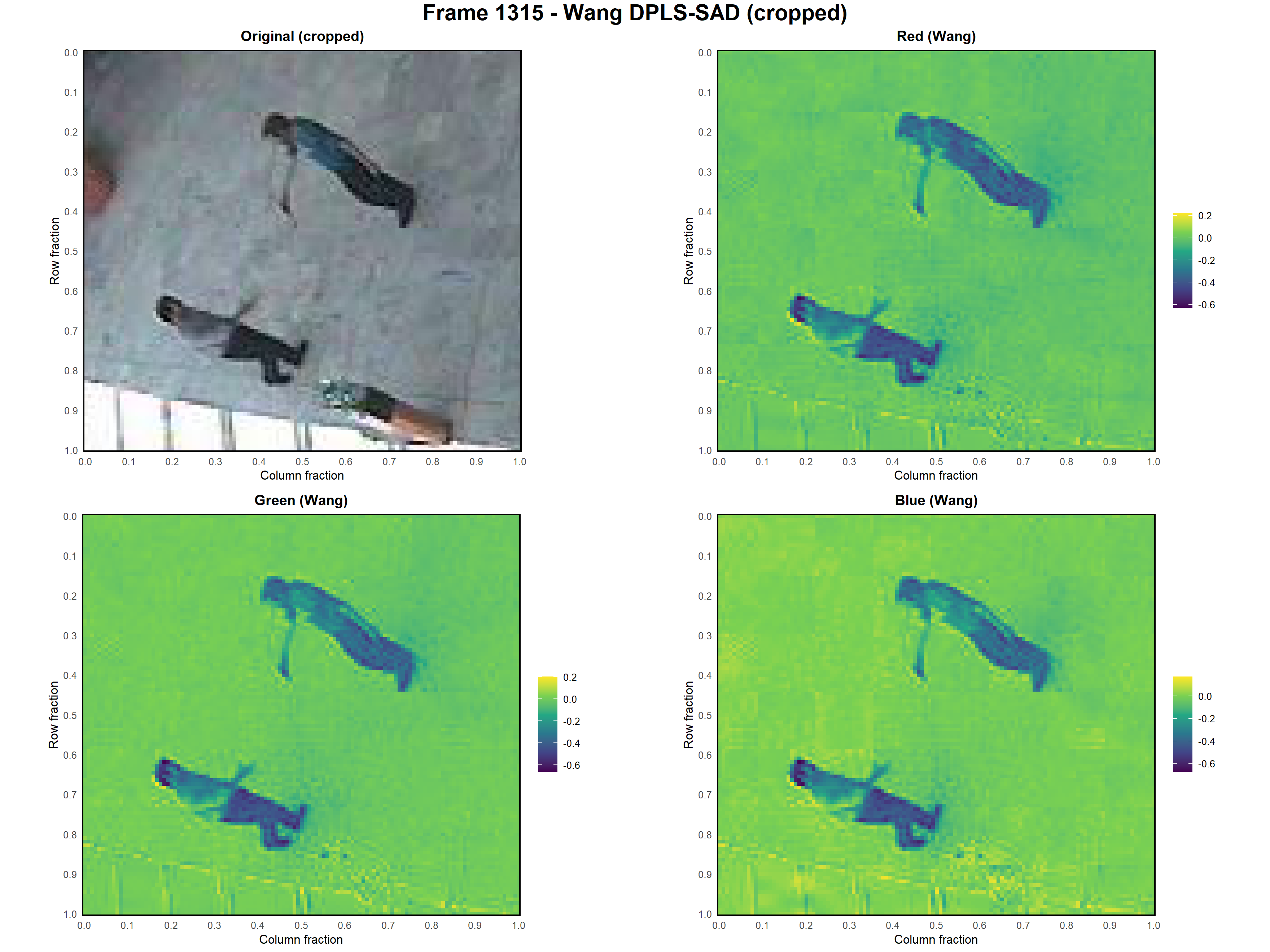}
        \caption{DPLS on cropped data}
        \label{fig:DPLS on cropped data}
    \end{subfigure}
    \caption{DCART (256 x 256) and DPLS (111 x 121) on cropped image}
    \label{fig:othermethoddata}
\end{figure}

\section{Conclusion}\label{sec:conclusion}
Despite recent attention, scalable spatial anomaly localization under general forms of spatial dependence remains a challenging problem. Prior work largely focuses on testing for the mere existence of anomalies, assuming restrictive structures such as $m$-dependence, and sacrificing computational feasibility for shape generality. In contrast, this paper focuses on identifying axis-aligned rectangular anomalous patches, introducing SPLADE: a fast, statistically accurate localization procedure robust to a wide class of spatially dependent data generating mechanisms.

SPLADE’s two-stage architecture leverages intelligent block-based sub-sampling, yielding massive computational speed-ups while facilitating rigorous theoretical guarantees under spatial dependence. Extensive experiments based on synthetic data across diverse dependence structures, anomalous patch configurations and signal strengths, not only validate our theoretically established guarantees, but also highlight SPLADE's significant computational efficiency advantage over competing methods. SPLADE’s performance does not hinge on any specific choice of tuning parameters; instead, it remains robust across a fairly broad range of settings, as supported by theoretical analysis and corroborated by ablation studies. A natural avenue for future research is extending this framework to other parametric shapes, such as ellipsoids, which we anticipate would primarily require careful modifications to the second stage of our algorithm.

\bibliographystyle{abbrvnat}
\bibliography{references}
\newpage
\appendix
\begin{center}\begin{large}\textbf{Appendix}\end{large}\end{center}
The Appendix contains all deferred discussions, including theoretical proofs and additional numerical experiments. In particular, Appendix \ref{se:cairoli} contains examples satisfying our key assumption quantifying a general dependence structure. Some auxiliary results follow in Appendix \ref{app:auxiliary} that will be used in proofs subsequently. Appendixes \ref{app:proof-1}, \ref{app:proof-2}, and \ref{app:proof-3} contains the proofs of the Theorems \ref{thm:epiconsistency}, \ref{thm:single-consistency} and \ref{thm:multiple} respectively. Appendix \ref{ssc:musigma} contains some deferred, finer details regarding implementation of SPLADE; Finally Appendices \ref{se:appendix simulation} and \ref{se:appendix real data} contain some additional simulation study, and one interesting application of SPLADE on fibre anomaly detection, respectively. \section{Assumption \ref{asmp:doob} and deferred discussion} \label{se:cairoli}
In the following, we illustrate the ubiquity of Assumption \ref{asmp:doob} through examples drawn from commonly occurring spatial processes.
\begin{lemma}\label{lem:m-dependent}
    Consider an $m$-dependent random field $(\varepsilon_{\bb{i}})_{\bb{i}\in \Z^d}$ satisfying $\varepsilon_{\bb{i}}$ and $\varepsilon_{\bb{j}}$ are independent if $|\bb{i}- \bb{j}|_{\infty}>m$. Then Assumption \ref{asmp:doob} is satisfied for $(\varepsilon_{\bb{i}})_{\bb{i}\in \Z^d}$.
\end{lemma}
\begin{proof}[of Lemma \ref{lem:m-dependent}]
Let $q := m+1$. For each $\bb{a} \in \{0,\dots,m\}^d$, define
\[
\Gamma_{\bb{a}} := \{\bb{i} \in \mathbb{Z}^d : i_r \equiv a_r \pmod{q},\ r \in [d]\}.
\]
Then $(\Gamma_{\bb{a}})_{\bb{a} \in \{0,\dots,m\}^d}$ forms a partition of $\mathbb{Z}^d$. Hence for every rectangle $I \subseteq [\bb{n}]$,
\begin{align}\label{eq:m-dep-1}
S_I^\varepsilon
&= \sum_{\bb{i} \in I} \varepsilon_{\bb{i}}
 = \sum_{\bb{a} \in \{0,\dots,m\}^d}
   \sum_{\bb{j} \in I \cap \Gamma_{\bb{a}}} \varepsilon_{\bb{j}} .
\end{align}
Fix $\bb{a}$. If $\bb{i} \neq \bb{j} \in \Gamma_{\bb{a}}$, then there exist
$k_{\bb{i}}, k_{\bb{j}} \in \mathbb{Z}^d$ such that
$\bb{s} = q k_{\bb{s}} + \bb{a}$ for $\bb{s} \in \{\bb{i},\bb{j}\}$. Consequently,
\begin{align*}
|\bb{i}-\bb{j}|_{\infty}
= q |k_{\bb{i}} - k_{\bb{j}}|_{\infty}
\ge q = m+1 .
\end{align*}
By $m$-dependence, $\varepsilon_{\bb{i}}$ and $\varepsilon_{\bb{j}}$ are independent. Define the coarse-lattice field
\[
\varepsilon_{\bb{k}}^{(\bb{a})} := \varepsilon_{\bb{a} + q \bb{k}},
\qquad \bb{k} \in \mathbb{Z}^d .
\]
Then $(\varepsilon_{\bb{k}}^{(\bb{a})})_{\bb{k} \in \mathbb{Z}^d}$ is an independent field. Let $I = \prod_{r=1}^d [u_r,v_r] \cap \mathbb{Z}^d$. Define
\begin{align*}
\alpha_r(I,\bb{a}) &= \Big\lceil \frac{u_r-a_r}{q} \Big\rceil, \ \
\beta_r(I,\bb{a}) = \Big\lfloor \frac{v_r-a_r}{q} \Big\rfloor .
\end{align*}
Set
\[
J(I,\bb{a})
:= \prod_{r=1}^d
[\alpha_r(I,\bb{a}), \beta_r(I,\bb{a})] \cap \mathbb{Z}^d .
\]
Then $I \cap \Gamma_{\bb{a}}
= \{\bb{a} + q \bb{k} : \bb{k} \in J(I,\bb{a})\},$
and therefore
\begin{align*}
\sum_{\bb{j} \in I \cap \Gamma_{\bb{a}}} \varepsilon_{\bb{j}}
=
\sum_{\bb{k} \in J(I,\bb{a})} \varepsilon_{\bb{k}}^{(\bb{a})}.
\end{align*}
Hence, \eqref{eq:m-dep-1} can be re-written as
\begin{align*}
S_I^\varepsilon
=
\sum_{\bb{a} \in \{0,\dots,m\}^d}
\sum_{\bb{k} \in J(I,\bb{a})}
\varepsilon_{\bb{k}}^{(\bb{a})}.
\end{align*}
Since $J(I,\bb{a})$ ranges over rectangles contained in a box with side lengths at most
$\lceil n_r/q \rceil$, Cairoli's maximal inequality for independent random fields yields

\begin{align*}
\Big\| \max_{I \subseteq [\bb{n}]} |S_I^\varepsilon| \Big\|_p
&\le
\sum_{\bb{a} \in \{0,\dots,m\}^d}
\Big\|
\max_{I \subseteq [\bb{n}]}
\Big|
\sum_{\bb{k} \in J(I,\bb{a})}
\varepsilon_{\bb{k}}^{(\bb{a})}
\Big|
\Big\|_p \\
&\le
C_{p,d} (m+1)^d \|\varepsilon_0\|_p
\prod_{r=1}^d \Big\lceil \frac{n_r}{m+1} \Big\rceil^{1/2}
= O(|\bb{n}|^{1/2}),
\end{align*}
which completes the proof.
\end{proof}
\begin{lemma}\label{lem:linear}
Consider the linear random field $\varepsilon_{\bb{i}}= \sum_{\bb{s}\in \Z^d} a_{\bb{s}} e_{\bb{i}- \bb{s}}$, where $(e_{\bb{s}})_{\bb{s}\in \Z^d}$ are i.i.d. mean-zero random variables and $\sum_{\bb{s}\in \Z^d}|a_{\bb{s}}|<\infty$. Then Assumption \ref{asmp:doob} is satisfied by $(\varepsilon_{\bb{i}})_{\bb{i}\in \Z^d}$. 
\end{lemma}
\begin{proof}[of Lemma \ref{lem:linear}]
    Observe that $ S_I^\varepsilon = \sum_{\bb{s}\in \Z^d} a_{\bb{s}} \sum_{\bb{i}\in I} e_{\bb{i}- \bb{s}}$, which immediately implies, via another application of Cairoli's maximal inequality for independent random fields,
    \begin{align*}
        \|\max_{I \subseteq [\bb{n}]} \ |S_I^\varepsilon|  \|_p \leq \sum_{\bb{s}\in \Z^d} |a_{\bb{s}}| \Big\|\max_{I \subseteq [\bb{n}]} \sum_{\bb{i}\in I} e_{\bb{i}-\bb{s}} \Big\|_p = O(|\bb{n}|^{1/2}).
    \end{align*}
\end{proof}

\section{Auxiliary Results} \label{app:auxiliary}
In this section we record some crucial auxiliary results in support of our theoretical arguments and broader analysis. Firstly, we address the feasibility of Assumption \ref{asmp:doob} by deriving it for a relatively broad class of spatial dependence. Specifically, \cite{cuny2025weak} establishes a Rosenthal inequality for the following dependence class:
\[\varepsilon_{\bb{i}}=g(e_{\bb{i}-\bb{s}}: \bb{s}\in \Z^d, \bb{s} \ge \bb{0}).\]
In \S\ref{se:model} we briefly discuss this model as an example of additional structures ( for example, $\bb{s}\ge \bb{0}$) imposed on spatial dependence in order to derive control on maximal partial sums. In Lemma \ref{lem:new-rosenthal}, we formalize this by proving Assumption \ref{asmp:doob} for this class. 

\begin{lemma}\label{lem:new-rosenthal}
    For mean-zero spatial stationary random field $(\varepsilon_i)_{i \in \Z^d}$, and a rectangle $I\subseteq [\bb{n}]$, let $S_I^\varepsilon := \sum_{\bb{j}\in I}\varepsilon_j$, and suppose $\|\varepsilon_{\bb{0}}\|_p < \infty$ for some $p>2$. Then, under the Assumptions of Theorem 17 of \cite{cuny2025weak}, it follows that 
    \[ \| \max_I |S_I|\|_p \leq C' |\bb{n}|^{1/2}, \]
    where $C'$ is independent of $\bb{n}$, and possibly dependent on $d$ and $p$.\
\end{lemma}
\begin{proof}
The result follows more-or-less straightforwardly from Theorem 17 of \cite{cuny2025weak}; nevertheless we provide a proof for completeness. 
    For $\bb{a}\in \Z^d$, $\bb{a} > \bb{0}$, let $S_{\bb{a}}^\varepsilon=\sum_{\bb{0} \leq \bb{i}\leq \bb{a}} \varepsilon_{\bb{i}}$. Note that, from equation (6.1) in \cite{cuny2025weak}, it follows that that 
    \begin{align}
        \Bigl\|\max_{\mathbf{k}\leq \mathbf{n}} |S_{\mathbf{k}}^\varepsilon|\Bigr\|_p
&\lesssim |\bb{n}|^{1/p}
\Biggl(
\|\varepsilon_{\bb{0}}\|_p
+
\Bigl[
\sum_{k_1=1}^{n_1}\!\ldots\!\sum_{k_d=1}^{n_d}
\frac{\|S_{k_1,\ldots,k_d}^\varepsilon\|_{p}^{2\delta}}
{k_1^{1+2\delta/p}\ldots k_d^{1+2\delta/p}}
\Bigr]^{\!1/(2\delta)}
\Biggr) \nonumber\\
& \leq  |\bb{n}|^{1/p}\Biggl(
\|\varepsilon_{\bb{0}}\|_p
+
\Bigl[
\sum_{k_1=1}^{n_1}\!\ldots\!\sum_{k_d=1}^{n_d}
\frac{k_1^\delta \ldots k_d^\delta }
{k_1^{1+2\delta/p}\ldots k_d^{1+2\delta/p}}
\Bigr]^{\!1/(2\delta)}
\Biggr) \nonumber\\
& \leq  |\bb{n}|^{1/p}\Biggl(
\|\varepsilon_{\bb{0}}\|_p
+
\Bigl[
\sum_{k_1=1}^{n_1}\!\ldots\!\sum_{k_d=1}^{n_d}
(k_1 \ldots k_d)^{\delta \frac{p-2}{p} -1}
\Bigr]^{\!1/(2\delta)}
\Biggr) \nonumber\\
& \leq |\bb{n}|^{1/p}(
\|\varepsilon_{\bb{0}}\|_p
+ |\bb{n}|^{1/2-1/p} ) \lesssim |\bb{n}|^{1/2}. \label{eq: left-bottom}
    \end{align}
    The result follows from \eqref{eq: left-bottom} by observing that for $\bb{a}, \bb{b}\in \Z^d$ with $\bb{a} \leq \bb{b}$, any rectangle $I_{[\bb{a},\bb{b}]}$ can be represented as:
    \[I_{[\bb{a},\bb{b}]}=\sum_{\boldsymbol{\eta}\in\{0,1\}^d}
(-1)^{\sum_j \eta_j}\,
S^{\varepsilon}_{\,\bb{b}-\boldsymbol{\eta}\odot(\bb{a}-\bb{1})},\]
where $\odot$ is component-wise dot-product. This completes the proof.
\end{proof}

Lemma \ref{lemma:anchoring} delivers control over sums over a special class of ``anchored'' rectangles that serves as building blocks in our proof of Theorem \ref{thm:epiconsistency}.

\begin{lemma}\label{lemma:anchoring}
    Grant Assumption \ref{asmp:doob}. For $\bb{a}, \bb{l}\in \Z^d_{\geq 0}$, define the \textit{anchored} rectangle
    \[ I(\bb{a}, \bb{l})= \prod_{k=1}^{d-1} [a_k, a_k+l_k] \otimes [1, l_d], \]
    where the anchoring is along the canonical axes in the $d$-th dimension. Given an integer $m>0$, consider the class
    \[ \mathcal{M}(m)= \{I: I= I(\bb{a}, \bb{l}), \bb{a}, \bb{l}\in \Z^d_{\geq 0}, |I|\leq m\}. \]
    Then it follows that 
    \[ \| \max_{I \in \mathcal{M}(m)} |S_{I}^\varepsilon| \|_p \leq C'' \sqrt{m} (\log m)^{\frac{d-1}{p}} (\prod_{k=1}^{d-1}n_k)^{1/p}.\]
\end{lemma}

\begin{proof}
    For each $\bb{r} = (r_1, r_2, \ldots, r_{d-1})\in \Z^{d-1}_{\geq 0}$, define $I(\bb{r})=\{I\in \mathcal{M}(m): 2^{r_k} \leq l_k < 2^{r_k+1},k \in [d-1] \}.$ For each $I(\bb{a}, \bb{l})\in I(\bb{r})$, it is evident that $$l_d\leq H(\bb{r}):= \lceil \frac{m}{2^{\sum_{k=1}^{d-1}r_k}} \rceil.$$ For each $k\in [d]$, let 
    \[ \mathcal{B}_{k,s}(\bb{r})= [s 2^{r_k+1}, (s+1)2^{r_k+1}], \ s\in \{1 , \ldots, \lceil \frac{n_k}{2^{r_k+1}}\rceil\}, \]
    denote a partition of $[1, n_k]$ into intervals of length $2^{r_k+1}$. Finally, to complete our notational preparation, for $\bb{t}=(t_1, \ldots, t_{d-1})\in \prod_{k=1}^{d-1}\{1 , \ldots, \lceil \frac{n_k}{2^{r_k+1}}\rceil \}$, let us define the rectangles
    \[ \mathcal{Q}(\bb{r}, \bb{t})= \prod_{k=1}^{d-1} \mathcal{B}_{k, t_k}(\bb{r}) \times [1, H(\bb{r})]. \]
     Take any rectangle $I(\bb{a},\bb{l}) \in I(\bb{r})$. Define
\[
t_k := \left\lfloor \frac{a_k-1}{2^{r_k+1}} \right\rfloor, k\in [d-1] .
\]
Then the interval $\{a_k,\ldots,a_k+\ell_k-1\}\subseteq \mathcal{B}_{k, t_k}(\bb{r})$, and consequently, $I(\bb{a},\bb{l})\subseteq \mathcal{Q}(\bb{r}, \bb{t}).$ Therefore, one writes
\begin{align}
     \| \max_{I \in \mathcal{M}(m)} |S_{I}^\varepsilon| \|_p \leq  \bigg\|\max_{\bb{r}}\max_{ \bb{t} \in \{1 , \ldots, \lceil \frac{n_k}{2^{r_k+1}}\rceil \}^{d-1}}\sup_{I \subseteq \mathcal{Q}(\bb{r}, \bb{t})} |S_{I}^\varepsilon| \bigg\|_p. \label{eq:spatial-dyadic}
\end{align}
Let us deal with the right-hand side of \eqref{eq:spatial-dyadic}. Firstly, for a fixed $\bb{r}, \bb{t}$, Assumption \ref{asmp:doob} instructs that
\begin{align}
    \|\sup_{I \subseteq \mathcal{Q}(\bb{r}, \bb{t})} |S_{I}^\varepsilon| \|_p \lesssim |\mathcal{Q}(\bb{r}, \bb{t})|^{1/2}\leq \bigg(H(\bb{r}) \prod_{k=1}^{d-1}2^{r_k+1}\bigg)^{1/2} \lesssim \sqrt{m}, \label{eq:innermost-dyadic}
\end{align}
where, $\lesssim$ hides constants pertaining to $d$. On the other hand, for each $\bb{r}$, the number of $\bb{t}$'s are at most $\prod_{k=1}^{d-1} n_k$/ Finally, noting that $|\bb{l}|\leq m$ if $I(\bb{a}, \bb{l}) \in \mathcal{M}(m)$, the number of possible $\bb{r}$'s are at most $(\log_2 m)^{d-1}$. The proof is completed by invoking an union bound on \eqref{eq:spatial-dyadic} in view of \eqref{eq:innermost-dyadic}.
\end{proof}

Proposition \ref{prop:tsconsistency} is arguably the most vital cog in the general strategy of our proof of Theorem \ref{thm:epiconsistency}, and derives a weaker upper bound that is then leveraged in a finer analysis in Appendix \ref{app:proof-1} to conclude the theorem.  

\begin{proposition}\label{prop:tsconsistency}
    Under the assumptions of Theorem \ref{thm:epiconsistency}, it holds that 
     \begin{equation}\label{eq:weakerrate}
        |I_0 \Delta \hat{I}_{LS}(\mathcal{C}_0c_n, 1 - \mathcal{C}_1c_n)| = O_{\IP}(n^{1/2}\delta^{-1}).
    \end{equation}
\end{proposition}
\begin{proof}[ of Proposition \ref{prop:tsconsistency}]
We borrow notation from the proof of Theorem \ref{thm:epiconsistency}. One can simplify $V_I^{\mu}$ as
    \begin{align} \label{eq: expression for V_I_mu}
        b(|I|)^{-1}V_I^{\mu} = \delta \bigg( \frac{x_3}{x_2+x_3} - \frac{x_4}{x_1+x_4}\bigg).
    \end{align}
    Consider the following series of simplification:
    \begin{align}
        &|V_{I_0}^{\mu}|^2-  |V_{I}^{\mu}|^2\nonumber\\ =& n^{-2}\delta^2\bigg((x_3+x_4)(x_1+x_2)-(x_2+x_3)(x_1+x_4)( \frac{x_3}{x_2+x_3} - \frac{x_4}{x_1+x_4})^2 \bigg) \nonumber\\
        =&n^{-2}\delta^2 (x_2+x_3)(x_1+x_4)\bigg((1 -\frac{x_2}{x_2+x_3}+ \frac{x_4}{x_2+x_3})(1- \frac{x_4}{x_1+x_4}+ \frac{x_2}{x_1+x_4}) \nonumber\\ & \hspace*{4cm} - (1- \frac{x_2}{x_2+x_3}- \frac{x_4}{x_1+x_4})(1-\frac{x_4}{x_1+x_4}- \frac{x_2}{x_2+x_3} )\bigg) \nonumber\\
        =& n^{-2}\delta^2 (x_2+x_3)(x_1+x_4)\bigg(\frac{x_2 x_3 + x_1x_4}{(x_2+x_3)(x_1+x_4)} + \frac{x_2x_3}{(x_2+x_3)^2} + \frac{x_1x_4}{(x_1+x_4)^2} \bigg) \nonumber\\
        =& n^{-2}\delta^2 (x_2+x_3)(x_1+x_4)\bigg( (\frac{x_2x_3}{x_2+x_3} + \frac{x_1x_4}{x_1+x_4})\frac{n}{(x_2+x_3)(x_1+x_4)} \bigg) \nonumber\\
        =& n^{-1} \delta^2\bigg( \frac{x_2x_3}{x_2+x_3} + \frac{x_1x_4}{x_1+x_4}\bigg)\label{eq:sqdiff}.
    \end{align}
Apart from characterizing the explicit difference between $|V_{I_0}^{\mu}|^2$ and $|V_{I}^{\mu}|^2$, \eqref{eq:sqdiff} also gives a very useful information: that $|V_{I_0}^{\mu}| \geq |V_{I}^{\mu}|$; this is of course expected, since at the population level, $|V_I^{\mu}|$ should be maximizing at the true interval $I_0$. We can exploit this equality as follows.
\begin{align}
    |V_{I_0}^{\mu}|-  |V_{I}^{\mu}| = \frac{|V_{I_0}^{\mu}|^2-  |V_{I}^{\mu}|^2}{|V_{I_0}^{\mu}|+  |V_{I}^{\mu}|} &\geq \frac{n^{-1} \delta^2( \frac{x_2x_3}{x_2+x_3} + \frac{x_1x_4}{x_1+x_4})}{2\delta \sqrt{\tau_n(1-\tau_n)}} \nonumber\\
    &\geq 4^{-1}\frac{\delta}{n\sqrt{\tau_n(1-\tau_n)}}\min\{x_I, x_1+x_3, nc_n \}, \label{eq:absdiff}
\end{align}
where $c_n := \min\{\tau_n, 1-\tau_n\}$. Next, we will show that there exists a constant $C_0>0$ such that $x_1+x_3\geq C_0nc_n$. Otherwise, there exist a sequence $\{r_n\}_{n\geq 1} \in \N$, $r_n \to \infty$, and a sequence of rectangles $(I_{r_n})_{n \geq 1}\in \mathcal{R}$ such that \begin{align}\frac{x_1 + x_3}{r_n c_{r_n}} \to 0,\label{eq:falseass}\end{align} as $n \to \infty$; note that here we have kept the dependence of $x_i$ on $n$ implicit; nevertheless, they are still sequences varying with $n$. Without loss of generality assume that $\tau_n \leq 1/2$ for all sufficiently large $n$, which implies $c_{r_n}=\tau_{r_n}$ for all sufficiently large $n$. Therefore, by definition of $x_i$'s, we also must have $\frac{x_3+x_4}{r_nc_{r_n}}\to 1$ as $n \to \infty$. Therefore, in light of \eqref{eq:falseass}, we have
\begin{align*}
   \frac{|I^c_{r_n}|}{|I_{0, r_n}|}=\frac{x_1+x_4}{r_nc_{r_n}}\to 1, \text{ as $n\to \infty$},
\end{align*}
which is in direct contradiction $\frac{|I_{r_n}|}{|I_{0,r_n}|} \geq \mathcal{C}_0$. Therefore, from \eqref{eq:absdiff}, one has
\begin{align}\label{eq:absdiff2}
    |V_{I_0}^{\mu}|-  |V_{I}^{\mu}| \geq 4^{-1} \frac{\delta}{n\sqrt{\tau_n(1-\tau_n)}}\min\{x_I, \mathcal{C}_0nc_n \}.
\end{align}
Consider equation (10) of \cite{bai1994} , which yields
\begin{align}\label{eq:keyinequality}
        |V_{\hat{I}_0}^{\varepsilon}| + |V_{I_0}^{\varepsilon}| \geq |V_{I_0}^{\mu}| - |V_{\hat{I}_0}^{\mu}|.
\end{align}
Let $\kappa$ be given. Denoting by $\hat{x}_I=|I_0 \Delta \hat{I}_0(\lambda_n)|$, \eqref{eq:absdiff2} and \eqref{eq:keyinequality} provides, for some approprately chosen $M_{\kappa}$, that
\begin{align}
    \IP(\hat{x}_I > M_{\kappa}n^{1/2}\delta^{-1}) &\leq \IP(|V_{\hat{I}_0}^{\varepsilon}| + |V_{I_0}^{\varepsilon}| \geq \delta \frac{c_n}{\sqrt{\tau_n(1-\tau_n)}}) \nonumber\\ & \hspace{2cm} + \IP(\sup_{I \in \mathcal{R}: n(1-\lambda_n)>|I|> n\lambda_n} |V_I^{\varepsilon}| >  M_{\kappa} (n^{1-4/p}c_n)^{-1/2}) \nonumber\\ &:= \eqref{I-II}\text{.1} + \eqref{I-II}\text{.2}. \label{I-II}
\end{align}
For \eqref{I-II}\text{.1}, note that  
\begin{align*}
    \eqref{I-II}\text{.1} &\leq \IP\bigg(\sup_{I \in \mathcal{R}: |I|> n\lambda_n} b(|I|)|\frac{S_I}{|I|} - \frac{S_{I^c}}{|I^c|}| \geq \delta \frac{c_n}{\sqrt{\tau_n(1-\tau_n)}}\bigg)\nonumber\\
    &\leq \IP\bigg(\sup_{I \in \mathcal{R}: n(1-\lambda_n)>|I|> n\lambda_n} \sqrt{\frac{n-|I|}{n}} \frac{|S_I|}{\sqrt{|I|}} \geq \sqrt{n}\delta \frac{c_n}{2\sqrt{\tau_n(1-\tau_n)}} \bigg) + \nonumber\\ & \hspace{2cm}\IP\bigg(\sup_{I \in \mathcal{R}: n(1-\lambda_n)>|I|> n\lambda_n} \sqrt{\frac{|I|}{n}} \frac{|S_{I^c}|}{\sqrt{|I^c|}} \geq \sqrt{n}\delta \frac{c_n}{2\sqrt{\tau_n(1-\tau_n)}} \bigg)\nonumber\\
    &= \eqref{I-II}\text{.1.1} + \eqref{I-II}\text{.1.2} 
\end{align*}
Let us first focus on \eqref{I-II}\text{.1.1}. In view of $\tau_n(1-\tau_n) \asymp c_n$ and Assumption \ref{asmp:doob}, we obtain
\begin{align}
        &\IP\bigg(\sup_{n(1-\mathcal{C}_1 c_n)>|I|> n\mathcal{C}_0 c_n } \sqrt{\frac{n-|I|}{n}} \frac{|S_I|}{\sqrt{|I|}} \geq C\sqrt{n}\delta \sqrt{c_n} \bigg) \nonumber\\
        \leq & \IP\bigg(\sup_{n(1-\mathcal{C}_1 c_n)>|I|> n\mathcal{C}_0 c_n } |S_I| \geq C{n}\delta \sqrt{\frac{c_n^2}{1-\mathcal{C}_1 c_n}}\ \bigg)
        \leq  C \frac{n^{p/2} (1-\mathcal{C}_1 c_n)^{p/2}}{n^p \delta^p c_n^p} \lesssim(nc_n^2\delta^2)^{-p/2}. \label{eq:7.13.trt}
    \end{align}
Note that \eqref{I-II}\text{.1.2} can be similarly tackled by noting that $n(1-\mathcal{C}_0 c_n)>|I|> n\mathcal{C}_1 c_n $ and $\sup_I |S_{I^c}| \leq |S_{I_{[\bb{1}, \bb{n}]}}|+ \sup_I |S_I|$. Henceforth, we shift focus to \eqref{I-II}\text{.2}. Note that, for a given $\kappa>0$ and appropriately chosen $M_{\kappa}$, 
    \begin{align}
        &\ \IP(\sup_{n(1-\mathcal{C}_1 c_n)>|I|> n\mathcal{C}_0 c_n } |V_I^{\varepsilon}| >  M_{\kappa} (n^{1-4/p}c_n)^{-1/2} ) \nonumber\\
        \leq &\ \IP(\sup_{n(1-\mathcal{C}_1 c_n)>|I|> n\mathcal{C}_0 c_n }\sqrt{\frac{n-|I|}{n}} \frac{|S_I|}{\sqrt{|I|}} >  M_{\kappa} c_n^{-1/2} ) + \IP(\sup_{n(1-\mathcal{C}_1 c_n)>|I|> n\mathcal{C}_0 c_n } \sqrt{\frac{|I|}{n}} \frac{|S_{I^c}|}{\sqrt{|I^c|}} > M_{\kappa} c_n^{-1/2} ) \nonumber\\
        \leq & \ 2\IP(\sup_{n(1-\mathcal{C}_1 c_n)>|I|> n\mathcal{C}_0 c_n } |S_{I}| > \sqrt{C}_0 M_{\kappa}n^{1/2 }) + o(1) \nonumber\\
        \leq & \  4C_0^{-p/2} M_{\kappa}^{-p} \label{eq:II.I-II}.
    \end{align}
Finally, \eqref{eq:7.13.trt} and \eqref{eq:II.I-II} in conjunction with \eqref{I-II} completes the proof.
\end{proof}

\section{Proof of Theorem \ref{thm:epiconsistency}}\label{app:proof-1}
Without loss of generality, we assume $\delta>0$, since otherwise we can replace $X_i$ with $-X_i$ to leave \eqref{eq:lsspatial} unchanged. Let $I_0:= \prod_{j=1}^d [n_j \tau_{1, j}, n_j \tau_{2,j}]$. For ease of exposition, we will also omit the niceties of $\lceil n_j\tau_j \rceil$, and pretend that $n_j{\tau}_j \in \N$. This of course, raises no issue in asymptotic analysis, since for $\gamma \in (0,1)$, $ n^{-1}\lfloor n\gamma \rfloor \asymp \gamma \asymp n^{-1}\lceil n\gamma \rceil$. Let $\bb{\tau}_1 = (\tau_{1,1}, \ldots, \tau_{1,d})$, and likewise $\bb{\tau}_2 = (\tau_{2,1}, \ldots, \tau_{2,d})$. Without loss of generality, we further assume that $\tau_n:= |\bb{\tau}_2- \bb{\tau}_1|<\frac{1}{2}$, so that $c_n=\tau_n$. The other case can be treated similarly. To explain our argument effectively, we require some notations. For a candidate interval $I$, denote
    \[ x_1= |I^c \cap I_0^c| , \ x_2= |I \cap I_0^c|, \ x_3=|I \cap I_0|, \text{ and } x_4=|I^c \cap I_0|, \]
    where $A^c=I_{[1,n]}\setminus A$. For $I \in \mathcal{R}$, call $x_I=|I \Delta I_0|$. Note that $x_I=x_2+x_4$, $x_1+x_2=n(1-\tau_n)$, and $x_3+x_4=n\tau_n$. Denote by $V_I^X=b(|I|)(\bar{X}_I- \bar{X}_{I^c})$, where $b(k)=\sqrt{k(n-k)n^{-2}}$. Let us further define $V_I^{\mu}=b(|I|)(\bar{\mu}_I - \bar{\mu}_{I^c})$. Note that, for a fixed $I \in \mathcal{R}$, $V_I^{\mu}=\IE[V_I^X]$, and in particular, $V_{I_0}^{\mu}= b(|I_0|)\delta$. We likewise define $V_I^{\varepsilon}$. 
    Recall $r_{n, \delta}$ from Theorem \ref{thm:epiconsistency}. Consider the sets 
    \begin{align*}
        \mathcal{D}&:=\{I : \mathcal{C}_0nc_n < |I| < n(1-c_n), \  x_I>M_{\eta}r_{n, \delta}^{-1} \}, \ 0<\mathcal{C}_0<1 \text{ is a small constant, and},\\
        \mathcal{D}_0 &:= \mathcal{D} \cap \{I:x_I< C_{\eta} n^{1/2}\delta^{-1}\},
    \end{align*}
    where, the choice of $M_{\eta}>0$ will be specified later, and $C_{\eta}$ is such that $\IP(|x_I|> C_{\eta} n^{1/2}\delta^{-1})<\eta$ upon invoking Proposition \ref{prop:tsconsistency}. Therefore, it is sufficient to control the probability $\IP(\sup_{I \in \mathcal{D}_0} |V_I| \geq |V_{I_0}|)$.
Clearly, 
    \begin{align}
        \IP(\sup_{I \in \mathcal{D}_0} |V_I| \geq |V_{I_0}|)&\leq \IP(\sup_{I \in \mathcal{D}_0} V_I-V_{I_0} \geq 0) + \IP(\sup_{I \in \mathcal{D}_0} V_I + V_{I_0} \leq 0)\color{black}:= \eqref{eq: abs-decomp}.1 + \eqref{eq: abs-decomp}.2. \label{eq: abs-decomp}
    \end{align}
We deal with the two terms sequentially.
\subsection{Control on $\eqref{eq: abs-decomp}.1$}
We write
\begin{align}
    \IP(\sup_{I\in \mathcal{D}_0} V_I - V_{I_0}\geq 0) & \leq \IP(\sup_{I\in \mathcal{D}_0}V_I^\varepsilon -  V_{I_0}^\varepsilon - V_{I_0}^\mu + V_I^{\mu}\geq 0)\nonumber\\
    &\leq \IP(\sup_{I\in \mathcal{D}_0}V_I^\varepsilon -  V_{I_0}^\varepsilon - 2^{-1}\frac{\delta}{n\sqrt{\tau_n(1-\tau_n)}}\min\{x_I, C_0nc_n \} \geq 0 ) \label{eq: use of abs-diff2}\\
    &\leq \IP(\sup_{I\in \mathcal{D}_0}x_I > C_0nc_n) + \IP(\sup_{I\in \mathcal{D}_0} \frac{V_I^\varepsilon -  V_{I_0}^\varepsilon}{x_I} \geq 2^{-1}\frac{\delta}{n\sqrt{c_n}})\nonumber\\
    & := \eqref{eq: abs-decomp}.1.1 + \eqref{eq: abs-decomp}.1.2 \nonumber,
\end{align}
where \eqref{eq: use of abs-diff2} follows from \eqref{eq:absdiff2}. Now, since $|x_I|< C_{\eta}n^{1/2}\delta^{-1}$, it follows $\eqref{eq: abs-decomp}.1.1=0$ for all sufficiently large $n$, as $n^{}c_n^2\delta^2 \to \infty$. So we move on to tackling \eqref{eq: abs-decomp}.1.2. Following the notations of $x_i, i=1(1)4$, we define the following for a candidate interval $I\in \mathcal{D}_0$. Let
\[ S_1^\varepsilon(I) = \sum_{i \in I^c \cap I_0^c} \varepsilon_i; \ S_2^\varepsilon(I) = \sum_{i \in I \cap I_0^c} \varepsilon_i; \ S_3^\varepsilon(I) = \sum_{i \in I \cap I_0} \varepsilon_i; \ S_4^\varepsilon(I) = \sum_{i \in I^c \cap I_0} \varepsilon_i. \]
For convenience, subsequently we keep the dependence of $S^{\varepsilon}_i$'s on $I$ implicit. With these notations in place, let us rewrite \eqref{eq: abs-decomp}.1.2 as follows
\begin{align}
  \frac{V_I^\varepsilon -  V_{I_0}^\varepsilon}{x_I} &= x_I^{-1} \bigg( \Big( b(x_2+x_3) \frac{S^{\varepsilon}_2+S^{\varepsilon}_3}{x_2+x_3} - b(x_3+x_4)\frac{S^{\varepsilon}_3 + S^{\varepsilon}_4}{x_3+x_4}\Big)- \Big( b(x_3+x_4) \frac{S^{\varepsilon}_1+S^{\varepsilon}_2}{x_1+x_2} - b(x_2+x_3)\frac{S^{\varepsilon}_1 + S^{\varepsilon}_4}{x_1+x_4}\Big) \bigg)  \nonumber\\
  &:= \frac{G(I)}{x_I} + \frac{H(I)}{x_I}. \nonumber
\end{align}
In the following, we establish a control over $\frac{G(I)}{x_I}$; the term $\frac{H(I)}{x_I}$ can be dealt with similarly. To that end, we further express $G(I)$ as follows.
\begin{align}
    G(I)&=b(|I|)\frac{S^{\varepsilon}_I}{|I|} - b(|I_0|) \frac{S^{\varepsilon}_{I_0}}{|I_0|} \nonumber\\
    &= b(|I_0|)S^{\varepsilon}_{I_0} ( \frac{1}{|I|}-\frac{1}{|I_0|}) -  b(|I_0|)\frac{S^{\varepsilon}_{I} - (S_2- S_4)}{|I|} + b(|I|)\frac{S_I}{|I|}\nonumber\\
    &=  b(|I_0|)S^{\varepsilon}_{I_0} \frac{x_4-x_2}{|I||I_0|} + ( b(|I|) -  b(|I_0|))\frac{S_I}{|I|} + b(|I_0|)\frac{S_2- S_4}{|I|} \nonumber\\
    &:= G_1(I) + G_2(I) + G_3(I). \nonumber 
\end{align}
Clearly,
\begin{align}
    \IP(\sup_{I\in \mathcal{D}_0}\frac{G_1(I)}{x_I} \geq \frac{\delta}{n\sqrt{c_n}}) \leq & \IP(\frac{S^{\varepsilon}_{I_0}}{\sqrt{|I_0|}} \geq \sqrt{n} \delta \mathcal{C}_0) = O((nc_n^2\delta^2)^{-1})=o(1), \label{eq:G1}
\end{align}
where \eqref{eq:G1} follows from $n^{1-4/p}c_n^2 \delta^2 \to \infty$. On the other hand, note that $|b(|I|) -  b(|I_0|)| \leq \frac{||I| - |I_0||}{n}\leq \frac{x_I}{n}$. Therefore,
\begin{align}
     \IP(\sup_{I\in \mathcal{D}_0} \frac{G_2(I)}{x_I} \geq \frac{\delta}{n\sqrt{\tau_n}}) \leq & \IP(\sup_{I\in \mathcal{D}_0} \frac{S^{\varepsilon}_I}{|I|} \geq \frac{\delta}{\sqrt{\tau_n}}) = O((nc_n^2\delta^2)^{-1})=o(1). \label{eq:G2} 
\end{align}
    Finally, for $G_3$, we proceed via considering a carefully orchestrated partitioning argument. To introduce this, let us first consider a coordinate-wise partition of intervals in the $k$-th dimension. For an interval $I$ with $I^k:=[a_k, b_k]$ denoting its slice in the $k$-th dimension, let
\begin{align*}
    L_k(I) = [a_k, \ b_k \wedge n_k \tau_{1,k}]; \ M_k(I) = [ a_k \vee n_k \tau_{1,k}, \ b_k \wedge n_k \tau_{2,k}], \text{ and } R_k(I) =[a_k \vee n_k \tau_{2,k}, b_k].
\end{align*}
Here, for notational convenience we assume that $[a,b]$ is empty is $b <a$. Clearly, for each $k$, $L_k(I)$, $M_k(I)$ and $R_k(I)$ are disjoint, and $I_k = L_k(I) \cup M_k(I) \cup R_k(I)$.  Observe that $I \cap I_0= \prod_{k=1}^d M_k(I)$. Therefore, if $\sigma= (\sigma_1, \ldots, \sigma_d) \in \{L, M, R\}^d$ and $$\sigma^{(k)}(I) =\begin{cases}
    &L_k(I), \text{ if } \sigma_k=L\\
     &M_k(I), \text{ if } \sigma_k=M\\
      &R_k(I), \text{ if } \sigma_k=R,
\end{cases}$$ it follows that 
\begin{align} \label{eq:peeling}
    I \cap I_0 = \cup_{\sigma \neq (M, \ldots, M)} \prod_{k=1}^d \sigma^{(k)}(I) .
\end{align}
Note that the $\sigma^{(k)}(I)$'s depend on $I$ through the $L_k(I), M_k(I)$ and $R_k(I)$'s. The representation \eqref{eq:peeling} facilitates a piecemeal application of Assumption \ref{asmp:doob} and \ref{lemma:anchoring}. We call \eqref{eq:peeling} a \textit{Peeling} representation of $I\cap I_0$, in that it resembles peeling $I\cap I_0$ into an union of disjoint rectangles along different axes; the name is also justified since it is on the individual rectangles $I(\sigma):= \prod_{k=1}^d \sigma^{(k)}(I)$ that we will apply the ``peeling'' trick of dyadic decompositions. To that end, let us introduce another, more general partition of the set of rectangles $\mathcal{R}:=\{I_{[a,b]}: \bb{1}\leq a \leq \bb{n}, a,b\in \Z^d\}$. Let     \[ \mathcal{R}^k = \{I: I=[a,b], 1\leq a < b \leq n\}, \ k\in [d]\]
be the set of corresponding $k$-th dimension slice. In particular, with $I_0^k=I_{[n_k \tau_{1,k}, n_k \tau_{2,k}]}$, $k\in [d]$, let
\begin{itemize}
        \item $\mathcal{P}_1^k:=\{I:I \subseteq I_0^k\}$; 
        \item  $\mathcal{P}_2^k:=\{I:I \supseteq I_0^k\}$;
        \item  $\mathcal{P}_3^k:=\{I:I \cap I_0^k=\phi\}$;
        \item  $\mathcal{P}_4^k:=\{I:a< n_k \tau_{1,k} < b < n_k \tau_{2,k}\}$;
        \item  $\mathcal{P}_5^k:=\{I: n_k \tau_{1,k} <a < n_k \tau_{2,k} <b\}$; 
    \end{itemize}
Clearly, $\mathcal{R}^k=\cup_{i=1}^5P_{i}^k$.
Then, a partition of $\mathcal{R}$ can be represented as 
\begin{align} \label{eq:big-partition}
    \bigg\{ \prod_{k=1}^d \mathcal{P}_{\alpha_k}^k:  \bb{\alpha}:=(\alpha_1, \alpha_2, \ldots, \alpha_d)\in \{1,\ldots, 5\}^d\bigg\}.
\end{align}
In view of \eqref{eq:big-partition}, we essentially have to deal with $5^d$ cases. Fix some $\bb{\alpha} \in [5]^d$, and let $\mathcal{P}_{\alpha}=\prod_{k=1}^d \mathcal{P}_{\alpha_k}^k$. Observe that if $\alpha_k=3$ for some $k\in [d]$, then $I \cap I_0 = \phi$. Therefore, for the sake of exposition, we consider the hardest case $\bb{\alpha}\in \{1,2,4,5\}^d$.
\begin{align}
    \IP(\sup_{I \in \mathcal{D}_0 \cap \mathcal{P}_{\bb{\alpha}}} \frac{G_3(I)}{x_I} \geq \frac{\delta}{n\sqrt{\tau_n}} ) & \leq \IP(\sup_{I \in \mathcal{D}_0\cap \mathcal{P}_{\bb{\alpha}}} \frac{S^{\varepsilon}_2 - S^{\varepsilon}_4}{|x_I|} \geq \mathcal{C}_0 \delta ) \nonumber\\
    & \leq \sum_{j=\lfloor\log_2 M_\eta r_{n, \delta}^{-1} \rfloor}^{\lceil \log_2 \sqrt{n}\delta^{-1}\rceil} \bigg( \IP(\sup_{\substack{I \in \mathcal{D}_0 \cap \mathcal{P}_{\bb{\alpha}}\\2^j<|x_I|<2^{j+1} }
    } |S_2^\varepsilon| > \mathcal{C}_0 \delta 2^j ) + \IP(\sup_{\substack{I \in \mathcal{D}_0 \cap \mathcal{P}_{\bb{\alpha}}\\2^j<|x_I|<2^{j+1} }
    } |S_4^\varepsilon| > \mathcal{C}_0 \delta 2^j ) \bigg).
   \label{eq:G3} 
\end{align}
Due to the similarity of the two terms in \eqref{eq:G3}, we only elaborate on the treatment of \begin{align} \label{eq:S2-big}
    \IP\Big(\sup_{\substack{I \in \mathcal{D}_0 \cap \mathcal{P}_{\bb{\alpha}}\\2^j<|x_I|<2^{j+1} }
    } |S_2^\varepsilon| > \mathcal{C}_0 \delta 2^j \Big). 
\end{align} 
 Depending on $\alpha_k$, we can further restrict the set of $\sigma$'s in the corresponding \textit{Peeling} representation of $I\cap I_0^c$ for a candidate rectangle $I$. Let $\sigma_k(\alpha_k)$ denote the particular set of choices for $\sigma_k$ given a $ \alpha=(\alpha_1, \ldots, \alpha_d)$. 
\begin{itemize}
    \item $\alpha_k=1 \implies \sigma_k(\alpha_k)=\{M\}, \  \sigma^{(k)}(I) = M_k(I) = [a_k, b_k]$.
    \item $\alpha_k=2 \implies \sigma_k(\alpha_k) = \{L, M, R \}, \  \sigma^{(k)}(I) \in \{L_k(I), M_k(I), R_k(I)\}$, where
    \[ L_k(I) = [a_k, n_k \tau_{1,k}]; \ M_k(I) = [ n_k \tau_{1,k}, n_k \tau_{2,k}], \text{ and } R_k(I) =[ n_k \tau_{2,k}, b_k].\]
    \item $\alpha_k=4 \implies \sigma_k(\alpha_k) = \{L, M\}, \ \sigma^{(k)}(I) \in \{L_k(I), M_k(I)\}$, where  
    \[ L_k(I) = [a_k, n_k \tau_{1,k}]; \ M_k(I) = [ n_k \tau_{1,k}, b_k].\]
    \item $\alpha_k=5 \implies \sigma_k(\alpha_k)= \{M, R\},\ \sigma^{(k)}(I) \in \{M_k(I), R_k(I)\}$, where 
     \[ M_k(I) = [a_k, n_k \tau_{2,k}],  R_k(I) =[ n_k \tau_{2,k}, b_k].\]
\end{itemize}
 Further, denote
\[ \sigma(\bb{\alpha}):= \{ (\sigma_1, \ldots, \sigma_d): \sigma_k \in \sigma_k(\alpha_k)\} \setminus \{M, \ldots, M\} .\]
Note that, due to the \textit{Peeling} representation \eqref{eq:peeling},
\begin{align}
    \sup_{I \in \mathcal{P}_\alpha}|S_2^\varepsilon| \leq \sum_{\sigma \in \sigma(\bb{\alpha})} \sup_{I \in \mathcal{P}_\alpha}|S_{I(\sigma)}^\varepsilon|, \ I(\sigma):= \prod_{k=1}^d \sigma^{(k)}(I). \label{eq:peeling-decomp}
\end{align}
For each $\sigma$, there exists $k_0$ such that $\sigma_{k_0} = L$ or $\sigma_k =R$. We focus on the first case, since the other case can be tackled symmetrically. Without loss of generality, let $k_0=d$. Evidently, then $I(\sigma) \in \mathcal{M}(2^{j+1})$, where $\mathcal{M}(m)$ is defined as in Lemma \ref{lemma:anchoring}. At this point, the end-points $\bb{a}, \bb{l}$ of the rectangles $I(\bb{a}, \bb{l})$ in $\mathcal{M}(m)$ may seem to be unrestricted, but we can further restrict the box by exploiting the condition $|x_I| < 2^{j+1}$.

 Indeed, for a generic $I:= I_{[a,b]}$ with $I \cap I_0\neq \phi$, $|x_I|< 2^{j+1}$ immediately implies that $\max_{j\in [d]}(|a_j - n_j \tau_{1,j}| \ \vee \ |b_j - n_j \tau_{2,j}|) < 2^{j+1}$. Consequently, $I(\sigma)$ can be enclosed in the following box:
\begin{align}
    I \subseteq \mathcal{B}:= \prod_{k=1}^d [n_k \tau_{2,k} - 2^{j+2}, n_k \tau_{2,k} + 2^{j+1}].  \label{eq:big-box}
\end{align} 
 From \eqref{eq:S2-big}-\eqref{eq:big-box}, one obtains, for each $\sigma \in \sigma(\bb{\alpha})$, that
\begin{align}
    \IP\Big( \sup_{\substack{I \in \mathcal{D}_0 \cap \mathcal{P}_{\bb{\alpha}}\\|x_I|<2^{j+1} }}|S_{I(\sigma)}^\varepsilon| > \mathcal{C}_0 \delta 2^j \Big) 
   \leq & \ \IP\Big( \sup_{\substack{I \in \mathcal{M}(2^{j+1})\\ I \subseteq \mathcal{B}}}|S_{I}^\varepsilon|  > \mathcal{C}_0 \delta 2^j\Big) \nonumber\\
   \lesssim & \frac{2^{jp/2} (\log 2^j)^{\frac{d-1}{p}} (2^j)^{\frac{d-1}{p}} }{\delta^p 2^{jp}} \label{eq:appl-of-new-lemma}\\
   \lesssim & \frac{j^{\frac{d-1}{p}} (2^j)^{\frac{d-1}{p} - \frac{p}{2}} }{\delta^p}, \label{eq:HR}
\end{align}
where, \eqref{eq:appl-of-new-lemma} involves an application of Lemma \ref{lemma:anchoring}.
Plugging \eqref{eq:HR} into \eqref{eq:G3} yields that
\begin{align}
     \IP(\sup_{I \in \mathcal{D}_0 \cap \mathcal{P}_{\bb{\alpha}}} \frac{G_3(I)}{x_I} \geq \frac{\delta}{n\sqrt{\tau_n}} )   & \lesssim \delta^{-p}\sum_{j=\lfloor\log_2 M_\eta r_{n, \delta}^{-1} \rfloor}^{\lceil \log_2 \sqrt{n}\delta^{-1}\rceil} j^{\frac{d-1}{p}} (2^j)^{\frac{d-1}{p} - \frac{p}{2}} \nonumber\\
     & \lesssim \delta^{-p} (\log_2 \sqrt{n}\delta^{-1})^{\frac{d-1}{p}} (M_{\eta}r_{n, \delta}^{-1})^{\frac{d-1}{p} - \frac{p}{2}}\label{eq:penultimate-HR}\\
     & \leq M_{\eta}^{\frac{d-1}{p} - \frac{p}{2}},\label{eq:final-HR}
\end{align}
where we have used $$r_{n,\delta}=\delta^{\frac{2}{1- 2(d-1)/p^2}} (\log_2 \frac{\sqrt{n}}{\delta})^{- \frac{2}{p^2/(d-1) -2}},$$ to simplify \eqref{eq:penultimate-HR} into \eqref{eq:final-HR}. Finally, since $p>\sqrt{2(d-1)}$ implies that $\frac{d-1}{p} - \frac{p}{2}<0$, hence, $M_{\eta}$ can be chosen to make \eqref{eq:final-HR} arbitrarily small. Therefore, from \eqref{eq:G1}, \eqref{eq:G2}, \eqref{eq:G3} and \eqref{eq:final-HR}, one obtains $\IP(\sup_{I \in \mathcal{D}_0}\frac{G(I)}{x_I} > \frac{\delta}{n\sqrt{c_n}})<\eta$ for all sufficiently large $n$. This completes the proof by establishing a control on \eqref{eq: abs-decomp}.1.2.

\subsection{Control on $\eqref{eq: abs-decomp}.2$}
We further divide this into two sub-cases as follows.
\begin{align}
    \IP(\sup_{I \in \mathcal{D}_0} V_I + V_{I_0} \leq 0) &\leq \IP(\sup_{I \in \mathcal{D}_0: V_I\geq 0} V_I + V_{I_0} \leq 0) + \IP(\sup_{I \in \mathcal{D}_0: V_I\leq 0} V_I + V_{I_0} \leq 0):= \eqref{eq: abs-decomp}.2.1+ \eqref{eq: abs-decomp}.2.2. \label{eq:second-case-decomp}
\end{align}
Write $V_I= V_I^{\varepsilon}+ V_I^\mu$. When $V_I>0$, it is immediate that $V_I + V_{I_0} \leq 0 \iff V_{I}^\varepsilon+V_{I_0}^\varepsilon \leq - V_I^\mu - V_{I_0}^\mu \leq -V_{I_0}^{\mu}$, and therefore, from $V_{I_0}^\mu=\delta \tau_n(1-\tau_n)$, one obtains,
\begin{align}
    \eqref{eq: abs-decomp}.2.1 &\leq \IP(\sup_{I \in \mathcal{D}_0: V_I\geq 0}|V_I^{\varepsilon}| \geq \frac{\delta \sqrt{\tau_n(1-\tau_n)}}{2})+ \IP(|V_{I_0}^{\varepsilon}| \geq \frac{\delta \sqrt{\tau_n(1-\tau_n)}}{2}) \nonumber\\ &\leq 2\IP(\sup_{I: |I|> nc_n}  |V_I^{\varepsilon}| \geq \frac{\delta \sqrt{\tau_n(1-\tau_n)}}{2})\nonumber\\=& o(1) \nonumber,
\end{align}
where $o(1)$ bound occurs by a treatment following verbatim from the corresponding analysis of the term \eqref{I-II}\text{.1} in \eqref{I-II}. Next, we show that \eqref{eq: abs-decomp}.2.2 is exactly zero. Indeed, from \eqref{eq: expression for V_I_mu}, $V_I^{\mu}<0$ trivially reduces to \begin{equation*}
 x_3x_1 < x_4x_2 \iff \frac{x_3}{x_2}<\frac{\tau_n}{1-\tau_n} \iff   |I| <  \frac{x_2}{1-\tau_n},
\end{equation*}
which, light of $I \in \mathcal{D}_0$, implies that 
\begin{align}\label{eq:neg-contra}
\mathcal{C}_0 n\tau_n < |I|< C_{\eta} (1-\tau_n)^{-1} \sqrt{n}\delta^{-1} \leq   2^{-1}C_{\eta}\sqrt{n}\delta^{-1}.
\end{align}
Clearly, in view of $nc_n^2\delta^2 \to \infty$, \eqref{eq:neg-contra} constitutes a contradiction for all sufficiently large $n$, showing that $\eqref{eq: abs-decomp}.2.2=0$.

\section{Proof of Theorem \ref{thm:single-consistency}}\label{app:proof-2}
    Recall $m$ from Algorithm \ref{algo:subsample}. A direct application of Theorem \ref{thm:epiconsistency} yields the error rate for our first stage estimators. More formally, let $r_{m, \delta}$ be defined as $r_{n, \delta}$ in Theorem \ref{thm:epiconsistency}, but with $n$ replaced by $m \asymp n^{1-\alpha}$. Since $m c_n^2 \to \infty$, given $\varepsilon >0$, Theorem \ref{thm:epiconsistency} instructs that there exists $M_{\eta}$ sufficiently large, such that for every $k\in [d]$,
     \begin{align}
         a_{I,k} &\in \mathcal{L}_k:=[M_k\tau_{1,k} - M_{\eta}r_{m, \delta}^{-1}, M_k\tau_{1,k} + M_{\eta} r_{m, \delta}^{-1}], \text{ and } \nonumber\\
         b_{I,k} &\in \mathcal{R}_k:= [M_k\tau_{2,k} - M_{\eta}r_{m, \delta}^{-1}, M_k\tau_{2,k} + M_{\eta} r_{m, \delta}^{-1}], \label{eq:1ststage}
     \end{align}
     holds with probability $\geq 1-\eta$, i.e. $\IP(\mathcal{A}) \geq 1-\eta$, where the event in \eqref{eq:1ststage} is denoted by $\mathcal{A}$. Let $\mathcal{L}_B= \prod_{k=1}^d \mathcal{L}_k$, and  $\mathcal{R}_B= \prod_{k=1}^d \mathcal{R}_k$. Clearly, at the second stage, it holds that 
     \begin{align}\label{eq:1ststage-decomp}
         \IP(|\Tilde{I} \Delta I_0|> G_\eta r_{n, \delta}^{-1}) \leq \sup_{ \bb{s} \in \mathcal{L}_B, \bb{t} \in \mathcal{R}_B} \IP(|\Tilde{I} \Delta I_0|> G_\eta r_{n, \delta}^{-1} \ | a_I =\bb{s}, b_I= \bb{t}) + \eta.
     \end{align}
     We note that conditional on $\mathcal{A}$, \eqref{eq:1ststage} instructs
     \begin{align}
         |b_{I,k}-a_{I,k}| \geq M_k (\tau_{2,k}-\tau_{1,k}) - 2 M_{\eta} r_{m, \delta}^{-1}. \label{eq:a-I-lb}
     \end{align}
     Conditional on the event $\{a_I =\bb{s}, b_I=\bb{t} \}$, define the set of rectangles 
     \begin{align}
     \mathcal{P}_{\bb{s},\bb{t}}&=\Bigg\{I_{[\bb{i},\bb{j}]}: \bb{i} \in \prod_{k=1}^d\Big[s_k L_k- CL_k n_k^{\kappa}(\log n)^{1/d}, s_k L_k + CL_k n_k^{\kappa}(\log n)^{1/d}\Big], \ \nonumber\\ & \hspace{2cm} \bb{j} \in \prod_{k=1}^d\Big[t_k L_k- CL_k n_k^{\kappa}(\log n)^{1/d}, t_k L_k + CL_k n_k^{\kappa}(\log n)^{1/d}\Big]  \Bigg\}.
     \end{align}
     Evidently, $\mathcal{P}_{\bb{s},\bb{t}}$ is motivated directly from the definitions of $\hat{L}_B$ and $\hat{R}_B$ from Algorithm \ref{algo:subsample}.
     Let $\bb{n}$ be sufficiently large that $(\min_k n_k)^\kappa > M_\eta r_{m, \delta}^{-1}.$ Crucially, note that $r_{m, \delta} \asymp r_{n, \delta}$, since $\log m \asymp \log n$. Then, conditional on $\mathcal{A} \cap \{a_I =\bb{s}, b_I=\bb{t} \}$, $I_0 \in \mathcal{P}_{\bb{s},\bb{t}}$. Let $\mathcal{D}_{\bb{s},\bb{t}}:= \{I: I \in \mathcal{P}_{\bb{s},\bb{t}}, \ |I \Delta I_0|> G_\eta r_{n, \delta}^{-1}\}$. Then, for a fixed $s<t$, it holds that 
     \begin{align}
         \IP(|\Tilde{I} \Delta I_0|> G_\eta r_{n, \delta}^{-1} \ | a_I =s, b_I= t) \leq \IP(\sup_{I \in \mathcal{D}_{\bb{s},\bb{t}}} |V_I| \geq |V_{I_0}| \ | a_I =\bb{s}, b_I=\bb{t}), \nonumber
     \end{align}
     where $V_I=b(|I|)(\bar{X}_I- \bar{X}_{I^c})$, where $b(k)=\sqrt{k(n-k)n^{-2}}$. Note that, if $I \in \mathcal{P}_{\bb{s}, \bb{t}}$, then with $\tau_k = \tau_{2,k}-\tau_{1,k}$ it follows 
     \begin{align}
         |I| &\geq \prod_{k=1}^d \Big( (t_k -s_k)L_k - 2CL_k n_k^{\kappa}(\log n)^{1/d} \Big)\nonumber\\
        & \overset{(a)}{\geq} \prod_{k=1}^d (n_k \tau_k - 2 M_{\eta} L_k r_{n, \delta}^{-1} - 2L_k n_k^\kappa (\log n)^{1/d} )\nonumber\\
        & \overset{(b)}{\geq} \prod_{k=1}^k (n_k \tau_k -  4L_k n_k^\kappa (\log n)^{1/d} )\nonumber\\
        & \geq \mathcal{C}_0 n c_n, \label{eq:condn-lower-bound} \ \text{a.s.}
     \end{align} 
     where, $(a)$ follows from \eqref{eq:a-I-lb}; $(b)$ follows from $(\min_k n_k)^\kappa > M_\eta r_{m, \delta}^{-1}$, and \eqref{eq:condn-lower-bound} is the consequence of the choice of $\alpha$ guaranteeing $\alpha+\kappa<1$, culminating in the final bound for a small enough constant $\mathcal{C}_0$. Similarly, it can be shown that if $I\in \mathcal{P}_{s,t}$, then $|I|\leq n(1- \mathcal{C}_1 c_n)$ for a small constant $\mathcal{C}_1$. Therefore, it follows that conditional on $\mathcal{A} \cap \{a_I=\bb{s}, b_I = \bb{t} \}$, $\mathcal{D}_{s,t}\subseteq \mathcal{D}$, where we recall $\mathcal{D}$ from the proof of Theorem \ref{thm:epiconsistency}. Clearly, Theorem \ref{thm:epiconsistency} instructs that, 
     \[ \IP(\sup_{I \in \mathcal{D}_{\bb{s},\bb{t}}} |V_I| \geq |V_{I_0}| \ | a_I =\bb{s}, b_I= \bb{t}) \leq \IP(\sup_{I \in \mathcal{D}} |V_I|\geq |V_{I_0}| ) < \eta, \]
     upon choosing $G_{\eta}$ appropriately. This completes the proof in light of \eqref{eq:1ststage-decomp}. 

\section{Proof of Theorem \ref{thm:multiple}} \label{app:proof-3}
The proof of Theorem \ref{app:proof-3}, while quite involved, mostly consists of sequential validation of each step of Algorithm \ref{algo:multiple-subsample}. In particular, in Step 1, we establish the validity of our testing mechanism in identifying small blocks inside anomalous rectangles. In Step 2, we leverage the uniform Gaussian approximation assumption \ref{ass:fclt} along with our careful deletion steps, to argue that the blocks with significant intersection with the background noise will not be selected by our mechanism. Steps 1 and 2 together show that we will select the correct number of patches in our algorithm with probability approaching $1$. Finally, in Step 3, we provide individual level localization rate for each anomalous patch, rounding off the theoretical analysis of SPLADE. 
\subsection{Step 1} Let for each $j\in [K]$,  $\widetilde{\mathcal{B}}_j=\{\bb{s}: B_{\bb{s}} \subseteq I_j\}$, and let $\widetilde{\mathcal{B}} = \cup_{j=1}^K \mathcal{B}_j$. Recall that, $|B_{\bb{s}}|\asymp n^{\alpha}$, and from Assumption \ref{ass:min-sep}, $\min_{j\in [K]} |b_{jl}- a_{jl}| \gg  n_l^{\alpha} \log^{1/d} n$ for all $l \in [d]$. Therefore for all sufficiently large $n$, it must hold that for each $j\in [K]$, $|\widetilde{\mathcal{B}}_j|\gg \log n$. Let $\bar{X}_{\bb{s}}$ be defined the same as in Step 9 of Algorithm \ref{algo:multiple-subsample}. In the following, we show that 
\begin{align}
    \IP\bigg(\min_{j\in [K]}\min_{\bb{s} \in \widetilde{\mathcal{B}}_j}|\bar{X}_{\bb{s}}| > \mathbf{Q}\bigg) \to 1, \ \text{as $n \to \infty$}. \label{eq:step-1-mult}
\end{align}
 Indeed, it follows that for a fixed $j \in [K]$, and $\bb{s}\in \widetilde{\mathcal{B}}_j$
\begin{align}
    \IP(|\bar{X}_{\bb{s}}| < \mathbf{Q}) &\leq \IP(|\bar{\varepsilon}_{\bb{s}}| > |\delta_j|- \mathbf{Q}) \leq \frac{O(n^{-p\alpha/2})}{(|\delta_j|- \mathbf{Q})^p},  \nonumber
\end{align}
which directly implies
\begin{align}
    \IP\bigg(\min_{j\in [K]}\min_{\bb{s} \in \widetilde{\mathcal{B}}_j}|\bar{X}_{\bb{s}}| \leq \mathbf{Q}\bigg) \lesssim \max_{j\in [K]} \frac{|I_j|}{n^{\alpha}} \frac{O(n^{-p\alpha/2})}{(|\delta_j|- \mathbf{Q})^p} \overset{(a)}{=} O (\max_{j\in [K]}n^{1-\alpha(p-1)/2}c_{nj} \delta_j^{-p}  ) \to 0, \nonumber
\end{align}
where in $(a)$ we used \eqref{eq:jump-min} together with the Gaussian tail bound $\mathbf{Q}\asymp \sqrt{\frac{\log n}{n^{\alpha}}}$ to conclude that $\min_{j\in[K]}|\delta_j|\gg \mathbf{Q}.$
This shows \eqref{eq:step-1-mult}.

\subsection{Step 2} In this step, we show that $\IP(\hat{K}=K)\to 1$ as $n \to \infty$. To that end, note that by the construction of $\bar{M}$, and since $C_j$'s and $\widetilde{\mathcal{B}}_j$'s are respectively disjoint, it holds
\begin{align}
    &\IP\bigg(\min_{j\in [K]}\min_{\bb{s} \in \widetilde{\mathcal{B}}_j}|\bar{X}_{\bb{s}}| > \mathbf{Q}\bigg) \leq \IP(\Acal_{n1}), \nonumber\\
    & \text{where } \Acal_{n1}:=\{\text{For every $j \in [K]$, there exists $i_j\in [\hat{K}]$, such that $\widetilde{\mathcal{B}}_j \subseteq C_{i_j}$}\}. \label{eq:ub-mcp-2},
\end{align}
Let $\IP_{E,F}(\cdot)=\IP(\cdot \ \cap E \ \cap F )$ for any events $E$, $F$. At this stage, the relationship between $\hat{K}$ and $K$ is still not entirely clear. Subsequently, we will show that under the event $\Acal_{n1}$, the mapping $j\mapsto i_j$ is injective, establishing that $\hat{K}\geq K$ with high probability. To that end, suppose there exists $k_1 < k_2\in [K]$ such that $i_{k_1}=i_{k_2}$. Let the common component $C_{i_{k_1}}= C_{i_{k_2}}$ be denoted by $C$. Without loss of generality, suppose that $\nu^{\star}_{I_{k_1} , I_{k_2}}=1$. Further, without loss of generality, we can assume that $\max\{0, a_{k_2,1} - b_{k_1,1}, a_{k_1,1} - b_{k_2,1}\}:= a_{k_2,1} - b_{k_1,1}$. Let $s_{k_1,1} = \lceil b_{k_1,1}/ L_1\rceil$, $s_{k_2,1}=\lceil a_{k_2,1}/ L_1\rceil$. Note that $s_{k_2,1} - s_{k_1,1}$ computes the gap between the rectangles $I_{k_1}$ and $I_{k_2}$ projected into the first dimension. Moreover, by Assumption \ref{ass:min-sep}, $s_{k_2,1} - s_{k_1,1} \geq n_1^{\alpha} \log^{1/d}n$. Consider the set $$A= \{\bb{l} \in \prod_{k=1}^d [M_k(I)]: s_{1,1} < l_1 < s_{2,1} \}.$$ Because $C$ is connected in $\mathbb{R}^d$, its projection $\pi_1(C)=\{x_1:x\in C\}\subset \mathbb{R}$ is connected and hence an interval; moreover, we have $[b_{1,1},\,a_{2,1}] \subset \pi_1(C)$. Hence, for each integer $r$ with $s_{1,1}\le r\le s_{2,1}$, there exists $\bb{s}_r\in [\bb{n}]$ with $\bb{s}_{r,1}=r$ and $B_{\bb{s}_r} \in C$. Clearly, by definition of rectangles, $s_{k_1,1}$ and $s_{k_2,1}$, $B_{\bb{s}_r} \notin I_{k_1} \cup I_{k_2}$ for all $s_{1,1}\le r\le s_{2,1}$. However, Assumption \ref{ass:fclt} instructs that
\begin{align}
    & \hspace{0.5cm} \IP_{\Acal_{n1}}(\text{There exists $\bb{s}_{s_{1,1}}, \ldots, \bb{s}_{s_{2,1}}$ such that $B_{\bb{s}_r} \in I_1^c \cap I_2^c \cap C$}) \nonumber\\
     & \leq \IP(\text{There exists at least $\log n$ many points $\bb{s} \in \prod_{k=1}^d [M_k(I)]$ with $\bar{\varepsilon}_{\bb{s}} > \mathbf{Q}$}) \nonumber\\
     & \overset{(b)}{\leq}  o(1) + \IP(\text{There exists at least $\log n$ many points $\bb{s} \in \prod_{k=1}^d [M_k(I)]$ with $|B_{\bb{s}}|^{-1} \mathbb{W}_{\bb{s}} > \mathbf{Q}/\sigma$}) \nonumber\\
      &\lesssim o(1) + (\log n)^{-\log n} = o(1) \label{eq:ub-mcp-3},
\end{align}
where $(b)$ follows from Assumption \ref{ass:fclt} and $n^{\alpha/2} \gg (\log n)^{-1/2} |\bb{n}|_{\infty}^{d/q}$. 
 Therefore, from \eqref{eq:ub-mcp-2} and \eqref{eq:ub-mcp-3}, jointly with Step 1, it follows that $\IP(\hat{K} \geq K)\to 1$. Before we show the other direction, we recalibrate by letting $\Acal_{n2}=\{\text{$C_{i_k}$'s are mutually disjoint for $k\in [K]$}\}$, and realizing that we have shown 
 \[ \IP(\Acal_{n1} \cap \Acal_{n2})\to 1, \text{ as $n\to \infty$}. \]
Now if $\hat{K}> K$, then under the event $\Acal_{n1}\cap \Acal_{n2}$, there exists $j\in [\hat{K}]$ such that $C_j$ and $ \cup_{s \in \mathcal{B}} B_s$ are disjoint. Consequently, it must be true that $|C_j \cap (\cup_{k=1}^K I_k)| \leq n^{\alpha}.$ Note that, by construction of $C_j$'s in Algorithm \ref{algo:multiple-subsample}, $|C_j| \geq c n^\alpha\sqrt{\log n}$. Therefore it must be true that there are at least $2^{-1}c \sqrt{\log n}$ many $\bb{s}$'s such that $B_s \cap C_j \cap (\cup_{k=1}^K I_j)=\phi$, and $\bar{X}_{\bb{s}} > \mathbf{Q}$. Hence it follows similar to \eqref{eq:ub-mcp-3} that
\[ \IP_{\Acal_{n1}, \Acal_{n2}}(\hat{K}>K) \to 0, \text{ as $n\to \infty$}, \]
which immediately implies \eqref{eq:correct-number-est}.
\subsection{Step 3}
In this step, we show \eqref{eq:multpl-consistency} conditional on $\Acal_{n3}:=\{\hat{K}=K\}$. Under $\Acal_{n3},$ without loss of generality, we can assume $i_j=j$ , $j \in [K]$. Conditional on $\Acal_{n3}$, in this step we establish the piecewise consistency of $\hat{I}_j$ in estimating the true rectangle $I_j$. Recall that $\hat{I}_j$ is obtained by implementing Algorithm \ref{algo:subsample} on the random rectangle $D_j$. The sets $D_j$'s themselves can be thought of as an enlargement of the random sets $C_j$'s into rectangles, so as to enable an application of Algorithm \ref{algo:subsample}. To facilitate further analysis, it is imperative that the sets $D_j$'s are disjoint with high probability. 

To that end, first observe that under $\Acal_{n1}$, $I_j \subseteq C_j$, and therefore, $I_j \subseteq D_j$. Consider the deterministic rectangles 
\[ D_j^\dagger:= \prod_{k=1}^d [n_k\tau_{1,k}^j- \frac{c_0}{2}n_k^\alpha \log^{3/2}n, n_k\tau_{2,k}^j + \frac{c_0}{2}n_k^\alpha \log^{3/2}n], \]
where $c_0$ is as in Assumption \ref{ass:min-sep}. Clearly, $D_j^\dagger$ are disjoint by invoking Assumption \ref{ass:min-sep}. We will show that $D_j \subseteq D_j^\dagger$ with high probability. Observe that, a proof similar to \eqref{eq:ub-mcp-3} can be employed to deduce that under, $\IP_{\Acal_{n1},\Acal_{n2}, \Acal_{n3}}(|C_j \setminus I_j|\leq \log n)\to 1$. Let $\Acal_{n4}:=\{ |C_j \setminus I_j|\leq \log n\}$.  Recall from Algorithm \ref{algo:multiple-subsample} that
\[ \ell_k^j = \min_{\bb{s}\in C_j} s_k, \ r_k^j= \max_{\bb{s}\in C_j} s_k. \]
Under $\mathcal{A}_{n4}$ it follows that 
\begin{align}
\min\{ |L_k \ell_k^j  - n_k\tau_{1,k}^j|, |L_k r_k^j  - n_k\tau_{2,k}^j| \}< \log n. \label{eq:Dj-gap} 
\end{align}
From the definition of $D_j$ and $D_j^\dagger$, it follows from \eqref{eq:Dj-gap} that
\begin{align}
    \IP_{\Acal_{n1},\Acal_{n2}, \Acal_{n3}, \Acal_{n4}}(I_j \subseteq D_j \subseteq D_j^\dagger) \to 1, \ \text{as $n\to \infty$.} \label{eq:D_j-consistency}
\end{align}
Let $\Acal_{n5}:=\{I_j \subseteq D_j \subseteq D_j^\dagger, j\in [K]\}$. For our final step, we analyze Algorithm \ref{algo:subsample} in the context of $D_j$ and $D_j^\dagger$. To that end, let us consider the naive, least-square based estimator on $D_j$. More formally, let
\begin{align}
    \tilde{I}_j&:= \argmax_{I \subset D_j, |D_j|\lambda_2>|I|> |D_j|\lambda_1} \sqrt{\frac{|I|(|D_j|-|I|)}{|D_j|^2}} |\bar{X}_I - \bar{X}_{I^c}| 
\end{align} 
Fix $\varepsilon>0$. Observe that, in light of $|I_j|\gg n^\alpha \log n$, one derives $r_{|D_j^\dagger|, \delta_j} \asymp r_{|I_j|, \delta_j} \asymp r_{n, \delta_j}$. Consequently, it is immediate that
\begin{align}
    \IP_{\cap_{u=1}^4 \Acal_{n,u}}(| \tilde{I}_j \Delta I_j|> M_{\eta}r_{|I_j|, \delta_j}) &\leq \IP_{\cap_{u=1}^4 \Acal_{n,u}}\big(\sup_{I\subseteq D_j: |I \Delta I_j|> M_{\eta}r_{|I_j|, \delta_j}} |V_I| \geq |V_{I_j}| \big) \nonumber\\ &\leq \IP(\sup_{I\subseteq D_j^\dagger: |I \Delta I_j|> M_{\eta}r_{|I_j|, \delta_j}} |V_I| \geq |V_{I_j}|)<\eta \label{eq:I-tilde-analysis}
\end{align}
where the choice of $M_\eta$ ascertains the control by $\eta$ via Theorem \ref{thm:epiconsistency}. Therefore, for the sub-sampling step of Algorithm 1, as long as the choice of $\alpha_j$ in Algorithm \ref{algo:subsample} satisfies \eqref{eq:indv-I_j-tradeoff}, the first stage localization around the end-points of $I_j$, similar to \eqref{eq:1ststage}, is achieved with high probability, conditional on $\cap_{u=1}^5\Acal_{n,u}$. Thereafter, an argument verbatim to that of Theorem \ref{thm:single-consistency} can be employed to conclude \eqref{eq:multpl-consistency}, and thus we omit the details.

\section{Deferred implementation details for SPLADE: How to get $\mu_0$ and $\sigma$?}\label{ssc:musigma}
An important aspect of Theorem~\ref{thm:multiple} is the requirement that the baseline mean $\mu_0$ and long-run variance $\sigma^2$ be known, which may not hold in many practical applications. Therefore, we briefly discuss a procedure for estimating both the parameters even in presence of anomalous patches. Consider the \textit{boundary layer with thickness $\beta$}
\[I_n^{\mathrm{bdry}}
= \left\{
\mathbf{i} = (i_1,\ldots,i_d) \in I_n :
\exists\, j \in [d] \ \text{s.t.}\ i_j \le n_j^\beta
\ \text{or}\ i_j \ge n_j - n_j^\beta + 1
\right\}.
\]
Figure \ref{fig: est-mu-sigma} depicts $I_n^{\mathrm{bdry}}$ in the case of $d=2$. In light of Assumption \ref{asmp:away-from-boundary}, $I_n^{\mathrm{bdry}}$ is disjoint from any of the anomalous patches, and therefore, can be safely employed to estimate both $\mu_0$ and $\sigma^2$. In particular, we replace $\mu_0$ by the corresponding sample mean over $I_n^{\mathrm{bdry}}$, which is consistent via Proposition 1 of \cite{el2013central}. On the other hand, for $\sigma^2$, we employ the Kernel-based estimators from \cite{steland2025inference}, which can also be understood as generalizations of HAC estimator (see \cite{newey-west, andrews1992}). 
\begin{figure}[htbp]
    \centering
    \includegraphics[width=0.5\linewidth]{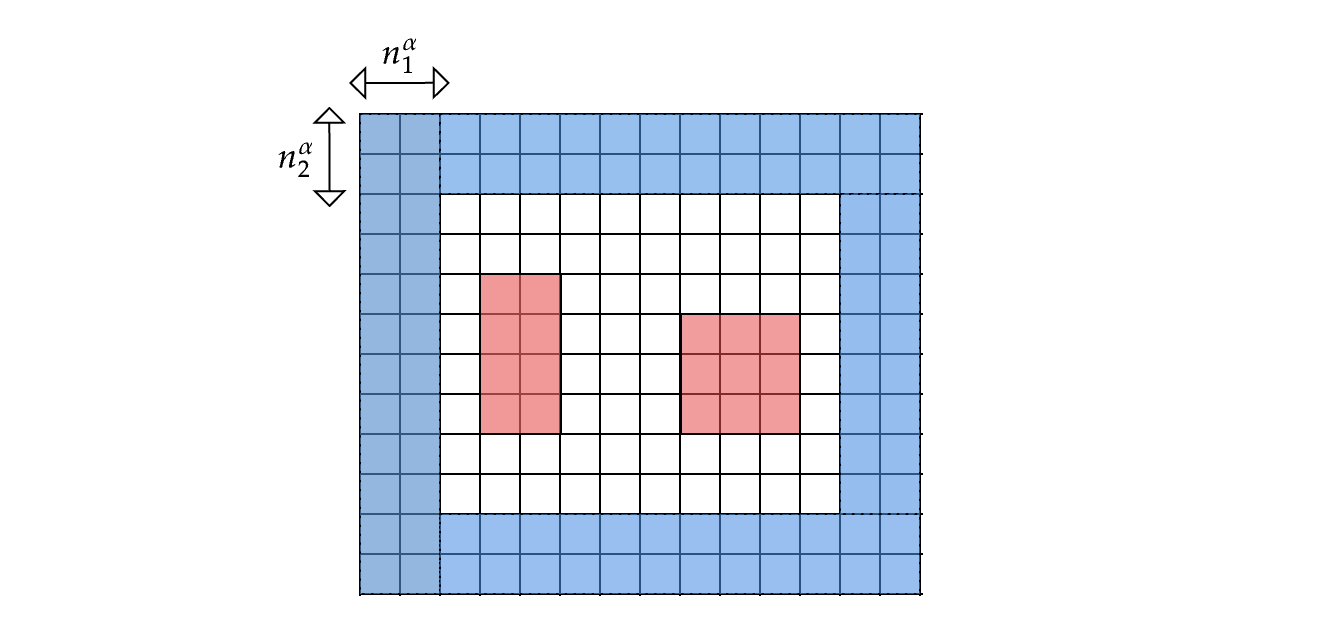}
    \caption{Example for $d=2$: estimate $\mu_0$ and $\sigma^2$ based on the \textcolor{blue}{blue shaded area}.}
    \label{fig: est-mu-sigma}
\end{figure}
Formally, let $K: \mathbb{R} \rightarrow \mathbb{R}$ be a symmetric kernel with bounded support $[-\omega, \omega]$, with $K \in \mathcal{C}^1$, and $\sup _x\left|K^{\prime}(x)\right| \leqslant C$. With a slight abuse of notation, for $v \in \mathbb{R}^d$, let $K(v):=K\left(v_1\right) \ldots K\left(v_d\right)$. Our long-run variance estimator reads
\[
\widehat{\sigma}^2
\;=\;
\frac{1}{\lvert I_n^{\mathrm{bdry}} \rvert}\sum_{i,j\in S}
K\!\big( (i-j)/ \mathbf{B}_n \big)\,
\big(X_i-\overline{X}_{I_n^{\mathrm{bdry}}}\big)\,\big(X_j-\overline{X}_{I_n^{\mathrm{bdry}}}\big),
\]
where $\mathbf{B}_n=(B_{n,1},\dots,B_{n,d})$ with $B_{n,k}\to\infty$ and $B_{n,k}/n_k^{1/d}\to 0$. Since the theoretical properties of $\hat{\sigma}^2$ for different choices of kernel functions and bandwidths $B_{n,k}$ follows directly from \cite{steland2025inference}, we omit that discussion for brevity.


\section{Additional Simulation results}\label{se:appendix simulation}
\noindent In this section we provide some more sensitivity analysis and an extensive comparative study for a non-linear spatial distribution. For $\tilde{\alpha}>2$, let $P_\varepsilon(\tilde{\alpha})$ denote the distribution of the random variable $Z - \IE[Z]$, where $Z \sim \text{Fréchet}(\tilde{\alpha})$. Let $\sum_{\bb{s}\in \Z^d} |a_{\bb{s}}| < \infty$. Define max stable distribution as
\[ Y_{\bb{t}} = \max_{\bb{s} \in \Z^d} a_{\bb{s}} \varepsilon_{\bb{t}- \bb{s}}, \ \varepsilon_{\bb{t}} \overset{i.i.d.}{\sim} P_\varepsilon(\tilde{\alpha}). \]
\noindent We set, in this part of our simulation, $a_{\bb{s}}= 0.6^{s_1+s_2}$. In contrast to the experiments on SAR model in \S\ref{sec:simu}, where an increasing $\rho$ indicates increasing spatial dependence, the relation between dependence and the parameter $\alpha$ is more nuanced for the Fréchet scenario. In particular, a larger $\alpha$ means lighter tails, so naively speaking, increasing $\alpha$ results in less extremal values under no anomalous patch, which might mean overall a weaker level of dependence. To further investigate this, we choose $\tilde{\alpha}= 2.75 \text{ and } 3$ for the sensitivity analysis and $\tilde{\alpha}= 2.5,2.75 \text{ and } 3$ for comparative study. In both cases we demean the data to keep it comparable with the mean-zero SAR ($\rho$) cases in the main draft. The sensitivity analysis results for $\delta_{\mu}=1$ are presented in Table \ref{tab:splade_ablation_config1_stable} and the comparative studies are deferred to Tables \ref{tab:grid256_512_frechet_stacked_k3} and \ref{tab:grid256_512_frechet_stacked_k5}.

\subsection{Sensitivity analysis for SPLADE- Max Stable Distribution}\label{ssc:simuablation_app}
One sees, from Table \ref{tab:splade_ablation_config1_stable} that our method SPLADE enjoys reasonable robustness across different choices of $\alpha$ parameter in the first-stage of Algorithm \ref{algo:multiple-subsample} SPLADE. In particular, ARI is consistently high and close to 1 across grid size and Fr\'echet parameter $\tilde{\alpha}$. When contrasted with Table \ref{tab:splade_ablation_config1_sar}, Table \ref{tab:splade_ablation_config1_stable} provides an interesting insight into the performance of SPLADE in the heavy-tailed set-ups, along with potentially hinting at the optimality of the theoretical assumptions in Theorem \ref{thm:multiple}. In particular, the max stable distribution is heavy-tailed, and therefore, the moment-based prescriptions in \eqref{eq:range of alpha} and \eqref{eq:final-choice} are perhaps more accurately applicable in this scenario. Note that \eqref{eq:final-choice} will be satisfied only if $\alpha>1/2$, which corresponds to the performance  boost SPLADE enjoys in Table \ref{tab:splade_ablation_config1_stable} for $\alpha=0.5$ and $0.6$ (especially in estimating $K$, and in terms of Hausdorff distance). As $N$ increases, the asymptotic regime kicks in, and for $N=1000$, the performance of SPLADE stabilizes for $\alpha>0.5$, reflecting back the same robustness property displayed in Table \ref{tab:splade_ablation_config1_sar}. This further vindicates our choice of $\alpha=1/2$ in all our experiments.
\begin{table}[!htbp]
\centering
\scriptsize
\setlength{\tabcolsep}{4.5pt}
\renewcommand{\arraystretch}{1.15}
\caption{Ablation study of SPLADE (on $\alpha$ from Stage -1 of Algorithm \ref{algo:multiple-subsample}) for Config.~1 with $\delta_{\mu}=1$ under Fr\'echet ($\tilde{\alpha}$). Each cell reports average over 100 replicates in the order $\alpha=0.4$ / $\alpha=0.5$ / $\alpha=0.6$.}
\label{tab:splade_ablation_config1_stable}
\begin{tabular}{|c| c| c| c| c| c|}
\hline
$N$ & Fr\'echet ($\tilde{\alpha}$) & $\hat K$ mean & $I(\hat K=3)$ & ARI & Hausdorff \\
\hline
\multirow{2}{*}{500}
& 2.75
& 4.19 / 3.34 / 2.91
& 0.33 / 0.7 / 0.91
& 0.991 / 0.974 / 0.976
& 0.674 / 0.338 / 0.093 \\
\cline{2-6}
& 3.00
& 4.06 / 3.38 / 2.94
& 0.36 / 0.72 / 0.978
& 0.993 / 0.974 / 0.982
& 0.643 / 0.352 / 0.064 \\
\hline
\multirow{2}{*}{750}
& 2.75
& 4.73 / 3.50 / 3.00
& 0.19 / 0.61 / 1.00
& 0.984 / 0.982 / 0.991
& 0.816 / 0.402 / 0.022 \\
\cline{2-6}
& 3.00
& 4.50 / 3.41 / 3.00
& 0.23 / 0.69 / 1.00
& 0.985 / 0.982 / 0.991
& 0.770 / 0.327 / 0.022 \\
\hline
\multirow{2}{*}{1000}
& 2.75
& 4.86 / 3.09 / 2.83
& 0.17 / 0.91 / 0.72
& 0.993 / 0.998 / 0.933
& 0.829 / 0.093 / 0.260 \\
\cline{2-6}
& 3.00
& 4.67 / 3.08 / 2.83
& 0.19 / 0.92 / 0.75
& 0.994 / 0.998 / 0.942
& 0.813 / 0.077 / 0.232 \\
\hline
\end{tabular}
\end{table}
\subsection{Comparing SPLADE with other competing methods-Max Stable Distribution}\label{ssc:simucomparativestable}
For both the grid sizes $256 \times 256 $ and $512 \times 512$, the performance of SPLADE is comparable across different $\tilde{\alpha}$ value; specifically for the larger grid, SPLADE may even seem to perform better as $\tilde{\alpha}$ increases. This vindicates our earlier notion that increasing $\tilde{\alpha}$ might mean weakening the dependency structure. To compare with other methods, firstly we focus on configuration 1, presented in Table \ref{tab:grid256_512_frechet_stacked_k3}. Here, the Hausdorff metric for SPLADE shows dramatic improvement compared to other methods while maintaining great ARI and accuracy for number of patches. 
Moreover, speed-wise, SPLADE beats DCART uniformly. However, it is slower than TV in all cases across two tables. One could also see that TV fails in other accuracy metrics compared to SPLADE almost everywhere. The speed-up of SPLADE compared to DCART is roughly between 1.2x-2x. As expected, SPLADE outperforms both methods in all 4 accuracy metrics almost uniformly. 

For configuration 2, displayed in Table \ref{tab:grid256_512_frechet_stacked_k5}, although sometimes for larger jump sizes, DCART catches up or marginally beats SPLADE in ARI or Hausdorff metric, they have a tendency to overestimate number of patches uniformly. 
\begin{table}[htbp]
\centering
\scriptsize
\setlength{\tabcolsep}{4.2pt}
\renewcommand{\arraystretch}{1.18}
\caption{Comparison of DCART, SPLADE and TV across grid sizes, jump sizes, and Fr\'echet $\tilde{\alpha}$ for Config.\ 1 (3 patches). Each cell reports avg.\ over 100 replicates in the order DCART / SPLADE / TV.}
\label{tab:grid256_512_frechet_stacked_k3}
\resizebox{\textwidth}{!}{%
\begin{tabular}{|c|c|c|c|c|c|}
\hline
Jump & $\hat K$ & $I(\hat K=3)$ & ARI & Hausdorff distance & time/iter (sec) \\
\hline

\multicolumn{6}{|c|}{\textbf{Grid = $256 \times 256$}} \\
\hline
\multicolumn{6}{|c|}{\textbf{Fr\'echet $\tilde{\alpha} = 2.50$}} \\
\hline
0.2 & 6.11 / \textbf{3.20} / 2.30 & 0.08 / \textbf{0.70} / 0.18 & 0.233 / \textbf{0.895} / 0.036 & 0.98 / \textbf{0.39} / 0.98 & 10.54 / 7.88 / \textbf{2.29} \\
\hline
0.4 & 5.98 / 3.11 / \textbf{2.95} & 0.06 / \textbf{0.67} / 0.25 & 0.286 / \textbf{0.899} / 0.206 & 0.99 / \textbf{0.41} / 0.98 & 9.89 / 8.54 / \textbf{2.29} \\
\hline
0.6 & 6.02 / \textbf{3.04} / 3.38 & 0.09 / \textbf{0.63} / 0.36 & 0.355 / \textbf{0.882} / 0.313 & 0.99 / \textbf{0.45} / 0.98 & 13.15 / 11.40 / \textbf{3.14} \\
\hline
0.8 & 7.59 / \textbf{3.01} / 3.96 & 0.01 / \textbf{0.62} / 0.26 & 0.500 / \textbf{0.872} / 0.593 & 1.00 / \textbf{0.46} / 0.98 & 17.30 / 14.64 / \textbf{4.00} \\
\hline
1.0 & 6.63 / \textbf{2.98} / 4.53 & 0.01 / \textbf{0.65} / 0.24 & 0.511 / \textbf{0.869} / 0.839 & 1.00 / \textbf{0.44} / 0.98 & 10.04 / 9.12 / \textbf{2.48} \\
\hline
\multicolumn{6}{|c|}{\textbf{Fr\'echet $\tilde{\alpha} = 2.75$}} \\
\hline
0.2 & 5.04 / \textbf{3.14} / 1.77 & 0.14 / \textbf{0.75} / 0.12 & 0.337 / \textbf{0.910} / 0.032 & 0.94 / \textbf{0.34} / 0.98 & 9.75 / 8.06 / \textbf{2.00} \\
\hline
0.4 & 4.37 / \textbf{2.99} / 2.53 & 0.23 / \textbf{0.70} / 0.22 & 0.372 / \textbf{0.885} / 0.261 & 0.98 / \textbf{0.39} / 0.97 & 9.87 / 8.92 / \textbf{2.13} \\
\hline
0.6 & 4.23 / \textbf{2.96} / 3.09 & 0.32 / \textbf{0.71} / 0.43 & 0.413 / \textbf{0.871} / 0.526 & 0.98 / \textbf{0.40} / 0.97 & 17.34 / 14.71 / \textbf{3.83} \\
\hline
0.8 & 6.28 / \textbf{2.94} / 3.68 & 0.00 / \textbf{0.71} / 0.49 & 0.614 / \textbf{0.874} / 0.836 & 1.00 / \textbf{0.40} / 0.97 & 15.32 / 13.20 / \textbf{3.40} \\
\hline
1.0 & 4.91 / \textbf{2.92} / 3.97 & 0.04 / \textbf{0.69} / 0.37 & 0.572 / 0.871 / \textbf{0.890} & 0.99 / \textbf{0.41} / 0.98 & 10.00 / 9.06 / \textbf{2.36} \\
\hline
\multicolumn{6}{|c|}{\textbf{Fr\'echet $\tilde{\alpha} = 3.00$}} \\
\hline
0.2 & 4.62 / \textbf{3.07} / 1.45 & 0.13 / \textbf{0.75} / 0.05 & 0.392 / \textbf{0.902} / 0.027 & 0.92 / \textbf{0.34} / 0.97 & 9.77 / 8.48 / \textbf{1.91} \\
\hline
0.4 & 3.99 / \textbf{2.96} / 2.37 & 0.30 / \textbf{0.71} / 0.15 & 0.452 / \textbf{0.877} / 0.315 & 0.97 / \textbf{0.39} / 0.97 & 9.84 / 9.00 / \textbf{2.06} \\
\hline
0.6 & 3.76 / \textbf{2.92} / 3.10 & 0.29 / \textbf{0.69} / 0.66 & 0.428 / \textbf{0.871} / 0.770 & 0.97 / \textbf{0.40} / 0.96 & 14.15 / 12.25 / \textbf{3.02} \\
\hline
0.8 & 5.88 / \textbf{2.91} / 3.47 & 0.01 / \textbf{0.68} / 0.63 & 0.639 / 0.869 / \textbf{0.890} & 0.99 / \textbf{0.41} / 0.97 & 9.94 / 9.01 / \textbf{2.20} \\
\hline
1.0 & 4.44 / \textbf{2.91} / 3.57 & 0.07 / \textbf{0.68} / 0.58 & 0.594 / 0.869 / \textbf{0.897} & 0.99 / \textbf{0.41} / 0.97 & 9.99 / 8.98 / \textbf{2.23} \\
\hline

\multicolumn{6}{|c|}{\textbf{Grid = $512 \times 512$}} \\
\hline
\multicolumn{6}{|c|}{\textbf{Fr\'echet $\tilde{\alpha} = 2.50$}} \\
\hline
0.2 & 12.30 / 3.15 / \textbf{3.10} & 0.00 / \textbf{0.77} / 0.19 & 0.132 / \textbf{0.927} / 0.008 & 1.00 / \textbf{0.25} / 1.00 & 50.78 / 29.19 / \textbf{11.12} \\
\hline
0.4 & 12.23 / \textbf{3.09} / 3.10 & 0.00 / \textbf{0.84} / 0.22 & 0.192 / \textbf{0.953} / 0.023 & 1.00 / \textbf{0.19} / 1.00 & 72.38 / 48.84 / \textbf{15.60} \\
\hline
0.6 & 12.70 / \textbf{3.14} / 3.67 & 0.00 / \textbf{0.86} / 0.32 & 0.434 / \textbf{0.972} / 0.198 & 1.00 / \textbf{0.17} / 1.00 & 73.36 / 51.69 / \textbf{16.44} \\
\hline
0.8 & 12.90 / \textbf{3.14} / 4.18 & 0.00 / \textbf{0.86} / 0.28 & 0.554 / \textbf{0.972} / 0.300 & 1.00 / \textbf{0.17} / 1.00 & 72.35 / 51.92 / \textbf{16.77} \\
\hline
1.0 & 13.29 / \textbf{3.13} / 4.53 & 0.00 / \textbf{0.87} / 0.19 & 0.667 / \textbf{0.972} / 0.421 & 1.00 / \textbf{0.16} / 1.00 & 42.14 / 31.57 / \textbf{10.70} \\
\hline
\multicolumn{6}{|c|}{\textbf{Fr\'echet $\tilde{\alpha} = 2.75$}} \\
\hline
0.2 & 5.92 / \textbf{3.04} / 2.27 & 0.15 / \textbf{0.89} / 0.13 & 0.226 / \textbf{0.954} / 0.009 & 0.99 / \textbf{0.14} / 0.99 & 68.27 / 44.11 / \textbf{13.27} \\
\hline
0.4 & 5.77 / \textbf{3.08} / 2.47 & 0.11 / \textbf{0.92} / 0.15 & 0.297 / \textbf{0.974} / 0.110 & 0.99 / \textbf{0.11} / 0.99 & 72.34 / 51.29 / \textbf{14.34} \\
\hline
0.6 & 6.21 / \textbf{3.08} / 3.15 & 0.09 / \textbf{0.92} / 0.24 & 0.511 / \textbf{0.974} / 0.276 & 0.99 / \textbf{0.11} / 0.99 & 72.96 / 51.87 / \textbf{15.25} \\
\hline
0.8 & 6.51 / \textbf{3.08} / 3.57 & 0.10 / \textbf{0.92} / 0.32 & 0.693 / \textbf{0.974} / 0.419 & 0.96 / \textbf{0.11} / 0.99 & 68.26 / 49.50 / \textbf{15.09} \\
\hline
1.0 & 6.70 / \textbf{3.08} / 4.28 & 0.02 / \textbf{0.92} / 0.26 & 0.749 / \textbf{0.974} / 0.818 & 0.96 / \textbf{0.11} / 0.99 & 42.09 / 31.62 / \textbf{9.86} \\
\hline
\multicolumn{6}{|c|}{\textbf{Fr\'echet $\tilde{\alpha} = 3.00$}} \\
\hline
0.2 & 3.95 / \textbf{3.07} / 1.85 & 0.33 / \textbf{0.93} / 0.11 & 0.278 / \textbf{0.974} / 0.010 & 0.97 / \textbf{0.10} / 0.99 & 71.02 / 49.17 / \textbf{12.86} \\
\hline
0.4 & 3.96 / \textbf{3.06} / 2.51 & 0.22 / \textbf{0.94} / 0.23 & 0.401 / \textbf{0.974} / 0.233 & 0.98 / \textbf{0.09} / 0.99 & 72.33 / 51.37 / \textbf{13.71} \\
\hline
0.6 & 3.96 / \textbf{3.08} / 2.68 & 0.22 / \textbf{0.92} / 0.28 & 0.585 / \textbf{0.973} / 0.327 & 0.94 / \textbf{0.11} / 0.99 & 73.26 / 51.91 / \textbf{14.59} \\
\hline
0.8 & 4.45 / \textbf{3.05} / 3.65 & 0.34 / \textbf{0.95} / 0.43 & 0.765 / \textbf{0.975} / 0.808 & 0.84 / \textbf{0.08} / 0.99 & 41.66 / 31.61 / \textbf{8.83} \\
\hline
1.0 & 4.52 / \textbf{3.06} / 4.00 & 0.32 / \textbf{0.94} / 0.46 & 0.774 / \textbf{0.976} / 0.934 & 0.95 / \textbf{0.09} / 0.99 & 41.85 / 31.72 / \textbf{9.23} \\
\hline
\end{tabular}%
}
\end{table}

\begin{table}[htbp]
\centering
\scriptsize
\setlength{\tabcolsep}{4.2pt}
\renewcommand{\arraystretch}{1.18}
\caption{Comparison of DCART, SPLADE and TV across grid sizes, jump sizes, and Fr\'echet $\tilde{\alpha}$ for Config.\ 2 (5 patches). Each cell reports avg.\ over 100 replicates in the order DCART / SPLADE / TV.}
\label{tab:grid256_512_frechet_stacked_k5}
\resizebox{\textwidth}{!}{%
\begin{tabular}{|c|c|c|c|c|c|}
\hline
Jump & $\hat K$ & $I(\hat K=5)$ & ARI & Hausdorff distance & time/iter (sec) \\
\hline

\multicolumn{6}{|c|}{\textbf{Grid = $256 \times 256$}} \\
\hline
\multicolumn{6}{|c|}{\textbf{Fr\'echet $\tilde{\alpha} = 2.50$}} \\
\hline
0.2 & 7.64 / \textbf{5.04} / 2.02 & 0.11 / \textbf{0.86} / 0.03 & 0.725 / \textbf{0.946} / 0.093 & 0.98 / \textbf{0.22} / 0.98 & 9.49 / 5.81 / \textbf{2.66} \\
\hline
0.4 & 8.31 / \textbf{5.01} / 4.21 & 0.03 / \textbf{0.99} / 0.22 & 0.898 / \textbf{0.963} / 0.661 & 0.91 / \textbf{0.08} / 0.97 & 9.67 / 5.70 / \textbf{2.87} \\
\hline
0.6 & 8.79 / \textbf{5.01} / 6.25 & 0.03 / \textbf{0.99} / 0.28 & 0.941 / \textbf{0.965} / 0.817 & 0.92 / \textbf{0.07} / 0.98 & 9.65 / 5.67 / \textbf{2.85} \\
\hline
0.8 & 8.63 / \textbf{5.00} / 6.98 & 0.00 / \textbf{1.00} / 0.07 & 0.951 / \textbf{0.965} / 0.907 & 0.89 / \textbf{0.06} / 0.97 & 9.75 / 5.69 / \textbf{2.94} \\
\hline
1.0 & 8.72 / \textbf{5.00} / 7.75 & 0.00 / \textbf{1.00} / 0.00 & 0.961 / \textbf{0.965} / 0.953 & 0.89 / \textbf{0.06} / 0.99 & 9.69 / 5.68 / \textbf{2.72} \\
\hline
\multicolumn{6}{|c|}{\textbf{Fr\'echet $\tilde{\alpha} = 2.75$}} \\
\hline
0.2 & 6.28 / \textbf{5.03} / 1.49 & 0.24 / \textbf{0.97} / 0.01 & 0.786 / \textbf{0.961} / 0.083 & 0.88 / \textbf{0.10} / 0.97 & 9.46 / 5.71 / \textbf{2.46} \\
\hline
0.4 & 6.77 / \textbf{5.01} / 4.84 & 0.06 / \textbf{0.99} / 0.49 & 0.933 / \textbf{0.965} / 0.780 & 0.63 / \textbf{0.07} / 0.96 & 9.61 / 5.69 / \textbf{2.72} \\
\hline
0.6 & 7.18 / \textbf{5.00} / 6.12 & 0.02 / \textbf{1.00} / 0.35 & 0.963 / \textbf{0.965} / 0.866 & 0.70 / \textbf{0.06} / 0.96 & 9.57 / 5.68 / \textbf{2.67} \\
\hline
0.8 & 7.07 / \textbf{5.00} / 6.77 & 0.00 / \textbf{1.00} / 0.00 & \textbf{0.972} / 0.965 / 0.957 & 0.60 / \textbf{0.06} / 0.98 & 9.63 / 5.67 / \textbf{2.61} \\
\hline
1.0 & 7.11 / \textbf{5.00} / 7.28 & 0.00 / \textbf{1.00} / 0.00 & \textbf{0.975} / 0.965 / 0.967 & 0.61 / \textbf{0.06} / 0.99 & 9.66 / 5.67 / \textbf{2.55} \\
\hline
\multicolumn{6}{|c|}{\textbf{Fr\'echet $\tilde{\alpha} = 3.00$}} \\
\hline
0.2 & 5.61 / \textbf{5.01} / 1.27 & 0.36 / \textbf{0.99} / 0.00 & 0.811 / \textbf{0.964} / 0.089 & 0.78 / \textbf{0.07} / 0.96 & 9.54 / 5.68 / \textbf{2.34} \\
\hline
0.4 & 6.20 / \textbf{5.00} / 5.13 & 0.11 / \textbf{1.00} / 0.71 & 0.946 / \textbf{0.965} / 0.818 & 0.40 / \textbf{0.06} / 0.96 & 9.59 / 5.68 / \textbf{2.56} \\
\hline
0.6 & 6.58 / \textbf{5.00} / 6.26 & 0.03 / \textbf{1.00} / 0.07 & \textbf{0.971} / 0.965 / 0.943 & 0.49 / \textbf{0.06} / 0.96 & 9.59 / 5.67 / \textbf{2.49} \\
\hline
0.8 & 6.37 / \textbf{5.00} / 6.49 & 0.02 / \textbf{1.00} / 0.00 & \textbf{0.974} / 0.965 / 0.968 & 0.32 / \textbf{0.06} / 0.99 & 9.66 / 5.69 / \textbf{2.50} \\
\hline
1.0 & 6.39 / \textbf{5.00} / 6.76 & 0.01 / \textbf{1.00} / 0.00 & \textbf{0.975} / 0.965 / 0.971 & 0.33 / \textbf{0.06} / 0.99 & 9.66 / 5.66 / \textbf{2.47} \\
\hline

\multicolumn{6}{|c|}{\textbf{Grid = $512 \times 512$}} \\
\hline
\multicolumn{6}{|c|}{\textbf{Fr\'echet $\tilde{\alpha} = 2.50$}} \\
\hline
0.2 & 14.61 / \textbf{5.06} / 2.64 & 0.00 / \textbf{0.83} / 0.04 & 0.798 / \textbf{0.974} / 0.024 & 1.00 / \textbf{0.19} / 0.99 & 43.0 / 30.3 / \textbf{11.4} \\
\hline
0.4 & 15.70 / \textbf{5.16} / 2.30 & 0.00 / \textbf{0.84} / 0.04 & 0.897 / \textbf{0.984} / 0.072 & 1.00 / \textbf{0.18} / 0.99 & 38.6 / 28.2 / \textbf{11.6} \\
\hline
0.6 & 15.94 / \textbf{5.16} / 4.60 & 0.00 / \textbf{0.84} / 0.22 & 0.950 / \textbf{0.984} / 0.676 & 1.00 / \textbf{0.17} / 0.99 & 38.8 / 28.3 / \textbf{11.8} \\
\hline
0.8 & 15.56 / \textbf{5.14} / 6.88 & 0.00 / \textbf{0.88} / 0.15 & 0.966 / \textbf{0.986} / 0.826 & 1.00 / \textbf{0.14} / 0.99 & 39.0 / 28.4 / \textbf{12.0} \\
\hline
1.0 & 15.66 / \textbf{5.14} / 8.24 & 0.00 / \textbf{0.88} / 0.04 & 0.975 / \textbf{0.986} / 0.864 & 1.00 / \textbf{0.14} / 1.00 & 38.9 / 28.3 / \textbf{12.5} \\
\hline
\multicolumn{6}{|c|}{\textbf{Fr\'echet $\tilde{\alpha} = 2.75$}} \\
\hline
0.2 & 8.46 / \textbf{5.08} / 1.81 & 0.06 / \textbf{0.91} / 0.00 & 0.850 / \textbf{0.983} / 0.025 & 0.96 / \textbf{0.12} / 0.99 & 38.3 / 28.0 / \textbf{9.8} \\
\hline
0.4 & 9.14 / \textbf{5.10} / 2.71 & 0.02 / \textbf{0.91} / 0.09 & 0.958 / \textbf{0.986} / 0.473 & 0.94 / \textbf{0.11} / 0.97 & 38.4 / 28.2 / \textbf{10.7} \\
\hline
0.6 & 9.24 / \textbf{5.10} / 5.31 & 0.01 / \textbf{0.91} / 0.38 & 0.981 / \textbf{0.986} / 0.812 & 0.92 / \textbf{0.11} / 0.98 & 38.5 / 28.2 / \textbf{10.8} \\
\hline
0.8 & 9.30 / \textbf{5.10} / 6.65 & 0.00 / \textbf{0.91} / 0.19 & 0.985 / \textbf{0.986} / 0.865 & 0.91 / \textbf{0.11} / 0.99 & 38.7 / 28.2 / \textbf{10.9} \\
\hline
1.0 & 9.30 / \textbf{5.10} / 7.77 & 0.00 / \textbf{0.91} / 0.02 & \textbf{0.986} / 0.986 / 0.953 & 0.91 / \textbf{0.11} / 0.99 & 38.7 / 28.3 / \textbf{11.6} \\
\hline
\multicolumn{6}{|c|}{\textbf{Fr\'echet $\tilde{\alpha} = 3.00$}} \\
\hline
0.2 & 6.52 / \textbf{5.10} / 1.32 & 0.23 / \textbf{0.91} / 0.01 & 0.881 / \textbf{0.986} / 0.025 & 0.80 / \textbf{0.11} / 0.98 & 38.3 / 28.0 / \textbf{9.4} \\
\hline
0.4 & 7.13 / \textbf{5.10} / 4.04 & 0.03 / \textbf{0.91} / 0.24 & 0.978 / \textbf{0.986} / 0.727 & 0.63 / \textbf{0.11} / 0.97 & 38.4 / 28.1 / \textbf{9.8} \\
\hline
0.6 & 7.19 / \textbf{5.10} / 5.55 & 0.00 / \textbf{0.91} / 0.54 & \textbf{0.988} / 0.986 / 0.853 & 0.61 / \textbf{0.11} / 0.98 & 38.5 / 28.2 / \textbf{10.1} \\
\hline
0.8 & 7.18 / \textbf{5.10} / 6.66 & 0.00 / \textbf{0.91} / 0.09 & \textbf{0.988} / 0.986 / 0.947 & 0.60 / \textbf{0.11} / 0.99 & 38.6 / 28.1 / \textbf{10.4} \\
\hline
1.0 & 7.18 / \textbf{5.10} / 7.17 & 0.00 / \textbf{0.91} / 0.00 & \textbf{0.989} / 0.986 / 0.978 & 0.59 / \textbf{0.11} / 0.99 & 38.7 / 28.1 / \textbf{11.0} \\
\hline
\end{tabular}%
}
\end{table}
\section{Additional Real-world dataset analysis}\label{se:appendix real data}

Compression in glass fibre-reinforced polymers often leads to unreliable or potentially deformed fibre clusters. Recently, multi-computed Tomography has been heavily used to produce three-dimensional voxelized images of the fibre microstructure \citep{emerson2017individual, garcea2018x} , which are further processed via MAVI (Modular Algorithms for Volume Images - \citet{wirjadi2016estimating}) over a scanning window to produce three-dimensional fibre directions. The key idea, as illustrated in \cite{dresvyanskiy2020detecting} is that fibres are just cylinders without a distinct ``head" or ``tail";  therefore, on a high level, MAVI constructs a cube around some voxel $(m_1, m_2, m_3)$, and averages the local fibre direction vectors $(x,y,z)$ inside the cube. Finally, corresponding to each of the directions $x,y,$ and $z$, the absolute value of the corresponding direction is assigned to the voxel $(m_1, m_2, m_3)$. Concretely, given a tomographic image of size $m_1 \times m_2 \times m_3$ voxels, a MAVI scan using equal-sized cubic blocks of side length $b$ produces three 3-dimensional datasets, corresponding to fibre directions parallel to the $x$-, $y$-, and $z$-axes, each of size
\[
\left(\frac{m_1}{b}\right)\times\left(\frac{m_2}{b}\right)\times\left(\frac{m_3}{b}\right).
\]
Based on the three-dimensional fibre-detection datasets, one can employ algorithms for anomalous patch detections to identify deformations on fibre systems. This key idea was analyzed in \cite{dresvyanskiy2019application, dresvyanskiy2020detecting} with the spatial dependence between the fibre directions being assumed to be $m$-dependent. In contrast, we employ Algorithm \ref{algo:multiple-subsample} that allows general form of spatial dependence, along with providing an estimate of the anomalous region, rather than tackling a testing problem. To that end, we consider the fibre directions datasets\footnote{The authors gratefully thank Prof. Claudia Redenbach for providing this dataset in personal communication. The dataset is available upon request contingent on permission from Prof. Redenbach.} corresponding to 3D-images of a glass fibre reinforced polymer, collected by the Institute for Composite Materials(IVW) in Kaiserslautern. In particular, we look at two images:
\begin{itemize}
    \item a \textit{simulated} glass-fibre image of $2000\times 2000\times 2100$ voxels, analyzed by MAVI with block-size $b=24$ to provide three datasets (corresponding to $3$ fibre-directions), each of size $83\times83\times87$ voxels. This simulated dataset acts as a baseline sanity check for the performance of Algorithm \ref{algo:multiple-subsample} against the detection algorithms in \cite{dresvyanskiy2020detecting}.
    \item Moving beyond simulated dataset, we also analyze a real glass-fibre image of $970\times 1469\times 1217$ voxels, which was further processed by MAVI with $b=15$ to produce three fibre-directions datasets, each of size $65\times 98 \times 81$ voxels. 
\end{itemize}
Mimicking main drafts setting, for each of the dataset, we employ Algorithm \ref{algo:multiple-subsample} SPLADE with $\alpha=0.5$, with the corresponding parameters for the application of Algorithm \ref{algo:subsample} inside Algorithm \ref{algo:multiple-subsample} being $\alpha_j=0.5, \ \kappa_j=0.01,$ for every $j\in [\hat{K}]$. The identified anomalous patches for both the simulated and real datasets are shown in Figure \ref{fig:fibre}. 

Firstly, we discuss the results on simulated dataset (Figure \ref{fig:fibre}, subfigures (a)-(b)-(c)). As a sanity check, we note that the anomalous hyper-rectangular patch is uniformly localized across each of three fibre-directions, which also corresponds almost identically to Figures 7 and 8 in \cite{dresvyanskiy2020detecting}. This indicates the accuracy of our algorithm, albeit in synthetic settings where fibre deformations are uniform across different directions. Moreover, as can be seen in Table 5 in the aforementioned paper, their cluster-based anomaly detection algorithm fails to detect any anomalous patch along the direction of $z$-axis. On the other hand, Algorithm \ref{algo:multiple-subsample} discovers the same anomalous rectangle for the $z$-axis fibre direction as with the other two directions, indicating the increased accuracy and power of our algorithm.

The situation is much more nuanced for the real glass-fibre reinforced polymer dataset, corresponding to \S 6.2 in \cite{dresvyanskiy2020detecting}. Therein, in Table 6, the authors show that they fail to detect any anomalous patch along $x$-axis direction, whereas anomalous patches are discovered along the $y$- and $z$- axes directions. In contrast, our results for the $x$-axis direction (depicted in Figure \ref{fig:fibre}, subfigures (d)) recover the same anomalous patch as highlighted by \cite{dresvyanskiy2020detecting}, yet again highlighting the elevated accuracy of our algorithm compared to available methods. Moving on, in \cite{dresvyanskiy2020detecting}, the discovered patches are then combined to produce Figures 14-15-16 therein, where the discovered anomalous regions are mostly identical.  However, our results show that such combinations of means of local fibre directions may be too simplistic to represent the total characteristics  of the anomalous region. For example, Figure \ref{fig:fibre}, subfigures (d)-(e)-(f) clearly show that the anomaly is anisotropic, and different spatial fibre directions exhibit different directional biases. This may highlight local shear or flow causing fibres to rotate mainly in one planar direction; compression or warping causing more change in vertical orientation; layered structure where one component picks up the boundary more strongly than another, or even diffuse orientation disorder - requiring further investigations and in-depth physical analysis. In light of this, Algorithm \ref{algo:multiple-subsample} can accurately predict complex anomalies occurring even in three-dimensional systems, marrying scalability with performance. 

\begin{figure}[htbp]
    \centering
    \begin{subfigure}[b]{0.32\textwidth}
        \centering
        \includegraphics[width=\textwidth]{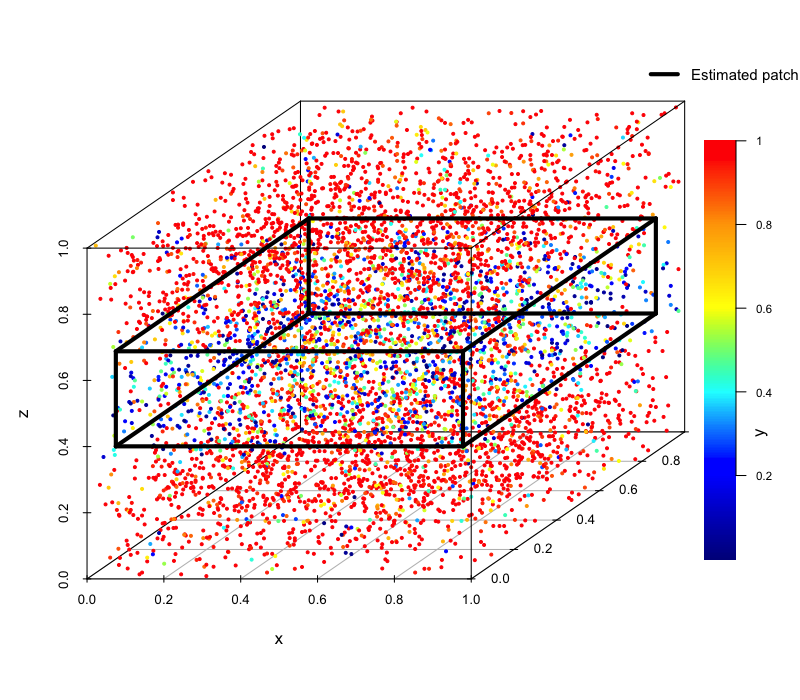}
        \caption{Simulated; fibre-direction: $x$-axis}
        \label{fig:sub1_fibre}
    \end{subfigure}
    \hfill
    \begin{subfigure}[b]{0.32\textwidth}
        \centering
        \includegraphics[width=\textwidth]{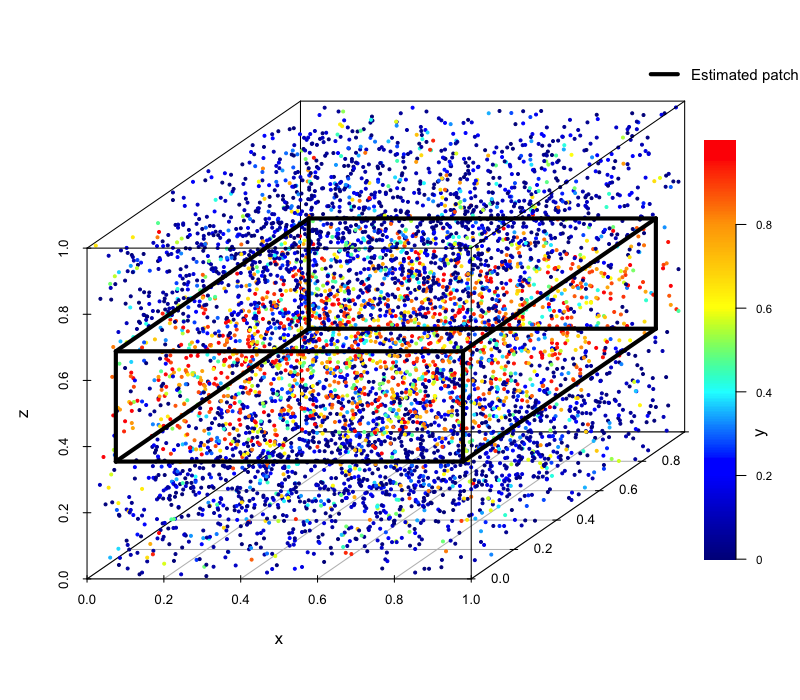}
        \caption{Simulated; fibre-direction: $y$-axis}
        \label{fig:sub2_fibre}
    \end{subfigure}
    \hfill
    \begin{subfigure}[b]{0.32\textwidth}
        \centering
        \includegraphics[width=\textwidth]{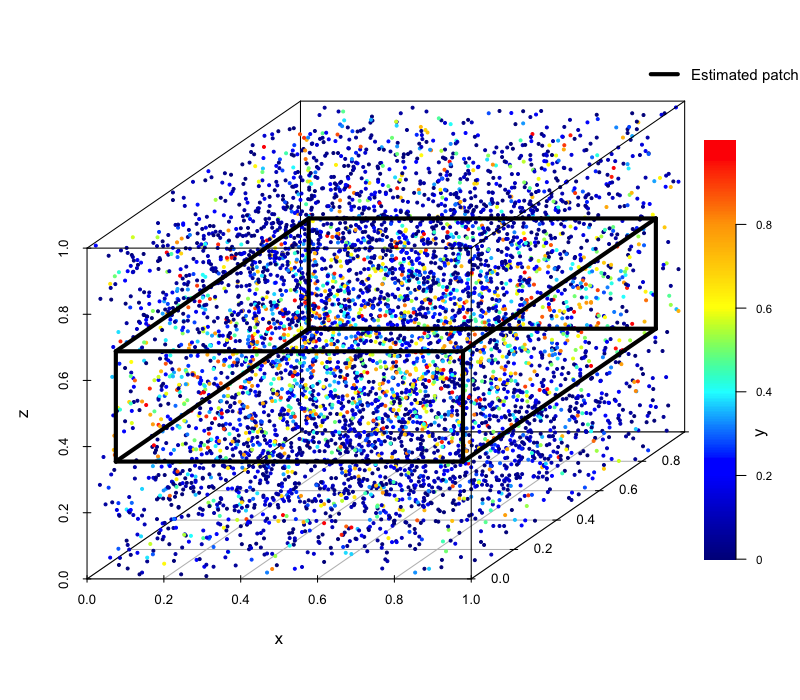}
        \caption{Simulated; fibre-direction: $z$-axis}
        \label{fig:sub3_fibre}
    \end{subfigure}

    \vspace{0.5em}

    \begin{subfigure}[b]{0.32\textwidth}
        \centering
        \includegraphics[width=\textwidth]{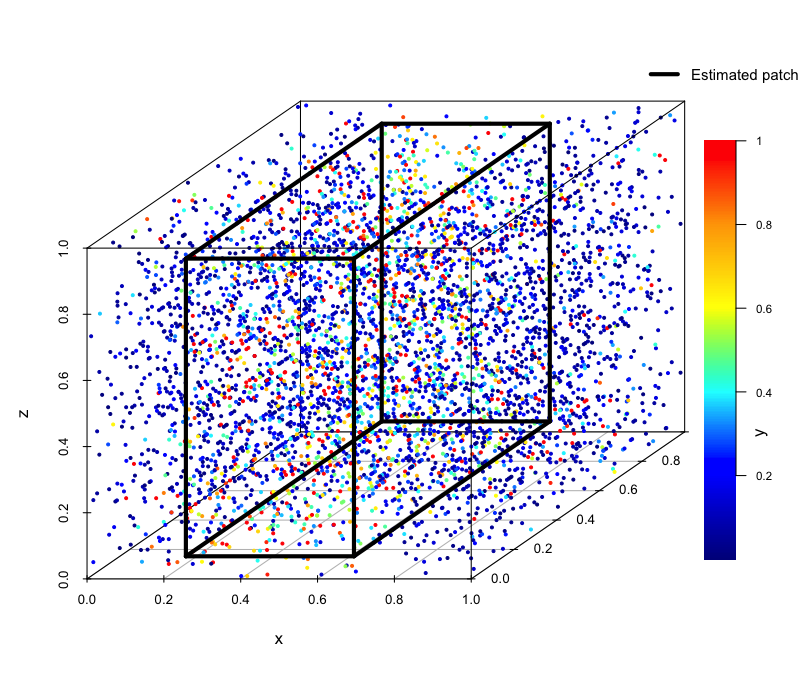}
        \caption{Real data; fibre-direction: $x$-axis}
        \label{fig:sub4_fibre}
    \end{subfigure}
    \hfill
    \begin{subfigure}[b]{0.32\textwidth}
        \centering
        \includegraphics[width=\textwidth]{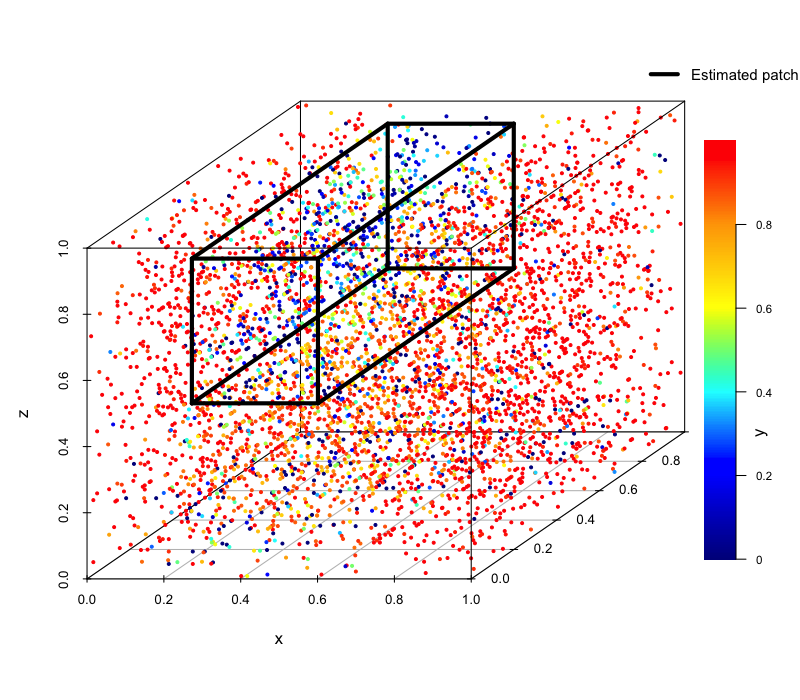}
        \caption{Real data; fibre-direction: $y$-axis}
        \label{fig:sub5_fibre}
    \end{subfigure}
    \hfill
    \begin{subfigure}[b]{0.32\textwidth}
        \centering
        \includegraphics[width=\textwidth]{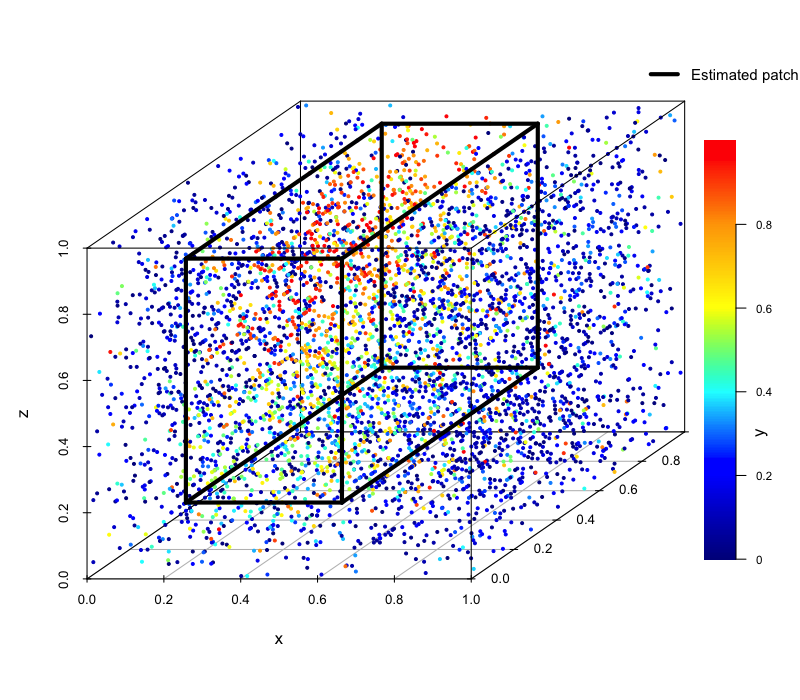}
        \caption{Real data; fibre-direction: $z$-axis}
        \label{fig:sub6_fibre}
    \end{subfigure}

    \caption{Application of Algorithm \ref{algo:multiple-subsample} to the fibre systems dataset.}
    \label{fig:fibre}
\end{figure}

\end{document}